\pdfoutput=1

\documentclass[aps,prx,onecolumn,superscriptaddress,nofootinbib,10pt]{revtex4-2}

\usepackage[T1]{fontenc}
\usepackage[utf8]{inputenc}
\usepackage{lmodern}

\usepackage{amsmath,amssymb,amsthm,mathtools,bm}
\usepackage{graphicx}
\usepackage{microtype}
\usepackage[hidelinks]{hyperref}
\usepackage[nameinlink,noabbrev]{cleveref}
\hypersetup{pdftitle={Quantum Natural Gradient on Quotient Spaces},pdfauthor={Zeyu Chen}}

\DeclareMathOperator{\Tr}{Tr}
\DeclareMathOperator{\rank}{rank}
\DeclareMathOperator{\Var}{Var}
\DeclareMathOperator{\Cov}{Cov}
\DeclareMathOperator{\diag}{diag}
\DeclareMathOperator{\spanop}{span}
\DeclareMathOperator{\End}{End}

\DeclareMathOperator{\grad}{grad}

\theoremstyle{plain}
\newtheorem{theorem}{Theorem}[section]
\newtheorem{proposition}[theorem]{Proposition}
\newtheorem{lemma}[theorem]{Lemma}
\newtheorem{corollary}[theorem]{Corollary}

\theoremstyle{definition}
\newtheorem{definition}[theorem]{Definition}

\newcommand{\cH}{\mathcal{H}}
\newcommand{\cK}{\mathcal{K}}
\newcommand{\cL}{\mathcal{L}}

\newcommand{\cV}{\mathcal{V}}

\newcommand{\cO}{\mathcal{O}}
\newcommand{\cQ}{\mathcal{Q}}
\newcommand{\bbC}{\mathbb{C}}
\newcommand{\bbR}{\mathbb{R}}

\newcommand{\bbN}{\mathbb{N}}

\newcommand{\ket}[1]{\left|#1\right\rangle}
\newcommand{\bra}[1]{\left\langle#1\right|}
\newcommand{\braket}[2]{\left\langle#1\,\middle|\,#2\right\rangle}

\begin{document}

\title{Quantum Natural Gradient on Quotient Spaces}

\author{Zeyu Chen}
\noaffiliation

\begin{abstract}
A parametrized quantum circuit reports its state geometry through a quantum Fisher information matrix (QFIM), often singular. A small Fisher value can reflect exact state-preserving redundancy, compression by the circuit chart, or weak intrinsic distinguishability, and these mechanisms call for different numerical treatments. We show that the circuit metric factors as $F=B^{*}MB$, where $B$ is the state-level circuit-to-orbit differential and $M$ is the intrinsic Fisher operator on the reachable orbit. The factorization identifies the exact kernel as $\ker B$, separates coordinate transfer from intrinsic geometry, and yields the condition for a circuit to realize an orbit-level quantum natural-gradient (QNG) direction. When the prescribed redundancy exhausts the Fisher kernel, the Moore--Penrose update is the minimum-norm horizontal lift of the quotient Riemannian gradient. At critical points with a locally diffeomorphic quotient-to-orbit map, chart singular values cancel from the linearized QNG operator while intrinsic anisotropy remains; in the trace-orthonormal full-control generator frame, excitation-gap anisotropy gives $\kappa_{\mathrm{QNG}}=\kappa_{\mathrm{Eucl}}$. Representation theory makes $M$ explicit on highest-weight, Slater, and fermionic-Gaussian orbits, and cominuscule fidelity flow becomes integrable, with conserved principal-defect ratios, cubic Lie-retracted convergence at $\eta=2$, and stability boundary $\eta=4$. Finite data impose a second boundary: an estimated QFIM and its confidence radius alone cannot distinguish an exact zero from a small physical mode, so the estimated spectrum alone cannot license hard projection. Under depolarization, inverse-Fisher scaling amplifies mean updates and fluctuations together and cannot restore update signal-to-noise. A redundant Slater/Givens circuit confirms exact transfer identities and illustrates finite-shot tradeoffs.
\end{abstract}

\maketitle

\section{Introduction}

Variational quantum algorithms encode physical states in gate parameters and optimize a state-dependent objective through a classical outer loop \cite{Cerezo2021VQA}.
Quantum natural gradient (QNG) replaces Euclidean parameter distance by the pullback of quantum statistical distinguishability, making the infinitesimal physical update invariant under regular reparameterizations \cite{Amari1998Natural,ProvostVallee1980,BraunsteinCaves1994,Stokes2020QNG,Martens2020NaturalGradient,Yamamoto2019NaturalGradient}.
The pullback is often singular because several parameter velocities can generate the same state velocity.
An exact state-preserving redundancy, a rank-changing family, and a physically weak but nonzero direction all produce small or zero Fisher eigenvalues, yet they require different algorithmic treatments.
Pseudoinversion is intrinsic in the first case only after the redundant representatives have been identified as a quotient; hard truncation in the third case can erase genuine descent directions.

Existing results characterize several parts of this picture separately.
QFIM rank measures effective model dimension and overparametrization, and singular quantum Fisher information governs nonregular estimation problems \cite{HaugBhartiKim2021,Larocca2023Overparametrization,Goldberg2021SingularQFI,Mihailescu2026MetrologicalSymmetries}.
Dynamical Lie algebras organize reachable sectors and also influence gradient concentration, although barren plateaus concern the objective gradient rather than coordinate redundancy itself \cite{McClean2018,Larocca2022DLA,Ragone2024,Larocca2025Barren}.
Wilson et al. established orbitwise QFIM spectral invariance and the fixed nonzero spectrum of the full-observable pure-state case \cite{Wilson2026GeometricInvariants}, while quotient-manifold optimization explains how exact group redundancies should be removed \cite{Lee2013,AbsilMahonySepulchre2008,Boumal2023}.
These results leave a coordinate-level gap: an intrinsic orbit spectrum does not determine the QFIM seen by a gate circuit, because the circuit differential can restrict and distort the orbit tangent space.
The same gap prevents orbit-level dynamics from being transferred automatically to a circuit and makes finite-shot null-space decisions ambiguous.

The practical problem is sharper than a request for unification. A small eigenvalue of a measured QFIM does not itself tell whether one should remove a state-preserving parameter direction, compensate distortion introduced by the circuit chart, or retain a weak physical direction. These interventions act on different mathematical objects, so the missing ingredient is a decomposition that separates where degeneracy enters.

Three structures must therefore be kept distinct. The physical symmetry group $K$ decomposes the Hilbert space into sectors, the parameter-redundancy group $\Gamma$ identifies coordinates that prepare the same state, and the sector dynamical group $G_\lambda$ generates the reachable orbit. On a $\Gamma$-invariant regular stratum for which the quotient is smooth, they form the chain
\begin{equation}
\Theta \xrightarrow{\ \pi\ } \cQ=\Theta/\Gamma
\xrightarrow{\ \bar\rho\ } \cO_\lambda=G_\lambda/H_\rho
\hookrightarrow \mathcal D(\cK_\lambda),
\label{eq:structural_chain}
\end{equation}
where $H_\rho$ is the state stabilizer. The first arrow removes redundant representatives, and the second maps quotient velocities to physical orbit tangents. We prove the operator factorization $F_\theta=B_\theta^{*}M_\rho B_\theta$, with $B_\theta$ the state-level circuit-to-orbit differential and $M_\rho$ the intrinsic Fisher operator on the reachable orbit. This is the organizing mechanism of the paper: $\ker B_\theta$ gives exactly the state-invisible circuit velocities, $\operatorname{im}B_\theta$ determines which orbit tangents the circuit can realize, the singular values of $B_\theta$ quantify coordinate distortion on those directions, and $M_\rho$ contains the intrinsic distinguishability. Because $B_\theta$ and $M_\rho$ need not share singular directions, the separation is operator-level rather than a unique assignment of each small matrix eigenvalue to one cause. The same factorization gives the pointwise transfer criterion: an orbit-level QNG direction is realized without projection exactly when it lies in the circuit-reachable tangent space; surjectivity of $B_\theta$ makes this automatic for every orbit objective.

Two consequences turn the factorization into a capability boundary. At a critical point where the quotient-to-orbit map is locally diffeomorphic, the QNG linearization is similar to the intrinsic operator $M_*^{-1}H_{\cO,*}$, so the singular values of a faithful circuit chart disappear while physical anisotropy remains. In the trace-orthonormal full-control generator frame, this residual anisotropy is exactly the excitation-gap spread and gives $\kappa_{\mathrm{QNG}}=\kappa_{\mathrm{Eucl}}$. Finite data impose a different boundary: an estimated QFIM and an operator-norm confidence radius alone cannot distinguish an exact zero from a small physical mode inside the same confidence ball. The estimated spectrum alone therefore cannot license hard projection. Known state-preserving structure can do so directly, while independently supplied rank-and-gap information can certify the local Fisher kernel without by itself identifying a global gauge action; unresolved directions are retained through soft regularization.

Representation theory makes the intrinsic factor explicit after the physical sector has been selected \cite{Larocca2022DLA,Ragone2024,Nguyen2024}. Under irreducible partial control, the moment-map identity identifies the total Fisher weight with the quadratic generalized-entanglement purity deficit; highest-weight roots then resolve the individual Fisher scales, while minuscule representations produce isotropic Slater and fermionic-Gaussian manifolds. For cominuscule projective embeddings, strongly orthogonal root directions factor the fidelity into principal defects and an isotropic maximal-flat Fisher block forces all active defects to share one contraction factor. The continuous flow is therefore integrable, whereas finite Lie-retracted steps break the defect-ratio invariants at second order; the embedding index controls the finite map, and tangent-space coverage determines whether a circuit inherits the intrinsic flow.

The statistical analysis converts the identifiability boundary into an implementation rule. Known exact gauges are projected before Fisher estimation and regularization; rank-and-gap certificates control kernel recovery when their model assumptions are independently supplied; unresolved low-curvature directions remain active and are softly damped. Perturbation and stationarity bounds show how the Fisher gap governs update amplification, projector leakage, and cumulative gauge drift, while a Schur-complement law and an independent-pilot ridge rule treat coupled soft modes without declaring them redundant. The accompanying 28-parameter Slater/Givens experiment tests the transfer law and these finite-shot tradeoffs at equal total shot budgets on a classically simulable free-fermion model.

\subsection{Contributions}

The paper makes three connected advances.
\begin{enumerate}
\item It establishes an operator-level diagnosis of singular circuit geometry. The factorization $F_\theta=B_\theta^*M_\rho B_\theta$ identifies exact state-level null directions as $\ker B_\theta$, separates circuit-coordinate transfer from intrinsic distinguishability, and gives the pointwise criterion for realizing the full orbit-QNG direction. In the faithful quotient regime, the Moore--Penrose step is the minimum-norm horizontal lift; when the quotient-to-orbit map is locally diffeomorphic, the linearized QNG operator is insensitive to singular-value distortion of the chart.
\item It makes the intrinsic factor explicit and solves the resulting fidelity dynamics on broad symmetry-reduced families. Under irreducible partial control, the identity $\Tr F=4E_{\mathfrak g_\lambda}$ links total Fisher weight to the quadratic generalized-entanglement purity deficit; highest-weight roots resolve the spectrum, and minuscule Slater and fermionic-Gaussian orbits collapse to one intrinsic scale. On cominuscule embeddings, fidelity QNG is completely integrable, while finite Lie-retracted steps have second-order defect-ratio drift, cubic local convergence at $\eta=2$, and exact stability boundary $\eta=4$.
\item It identifies limits that determine the finite-shot algorithm. In the trace-orthonormal full-control generator frame, QNG leaves excitation-gap conditioning unchanged, $\kappa_{\mathrm{QNG}}=\kappa_{\mathrm{Eucl}}$; a QFIM estimate within a confidence ball cannot by itself distinguish an exact zero from a small physical mode; and under depolarization inverse-Fisher scaling cannot recover update signal-to-noise. These boundaries motivate structural gauge projection, soft regularization of unresolved modes, and conditional recovery of the local Fisher kernel when independent rank-and-gap information is available. The Slater/Givens experiment tests the transfer mechanism and the resulting resource--accuracy tradeoffs.
\end{enumerate}

\subsection{Related work and scope}

Metric-aware optimizers, stochastic Fisher estimators, and randomized measurement protocols reduce the estimation or implementation cost of QNG \cite{Wierichs2020Avoiding,vanStraaten2021Measurement,Gacon2021SPSA,Kolotouros2024RandomNG,Halla2025QNGGeodesic,Halla2025SteinQFI,Rath2021RandomizedQFI}.
Rank-based analyses identify effective model dimension \cite{HaugBhartiKim2021,Larocca2023Overparametrization}, and horizontal-gate constructions remove symmetry directions at the ansatz level through homogeneous-space decompositions \cite{Wiersema2025Horizontal}.
The transfer theorem complements these approaches by determining, for any smooth state map into an orbit, how intrinsic Fisher scales and tangent directions appear in a redundant gate chart.

Manifold and Grassmann optimization provide computational primitives for the orbit examples studied here \cite{AbsilMahonySepulchre2008,Absil2002GrassmannRQI,Absil2004CubicGrassmann,Bendokat2024Grassmann}.
The cubic rank-one retraction at $\eta=2$ shares an order of convergence with Grassmann Rayleigh-quotient iterations, while arising from an embedded fidelity loss rather than a matrix eigenproblem.
The coherent-state and Hermitian-symmetric-space geometry used in the pure-state analysis is classical \cite{Perelomov1986,Berceanu1997GeometryCoherentStates,Berceanu2004GeometricalPhases}; compatible Fisher tensors on mixed-state unitary coadjoint orbits provide the corresponding geometric setting beyond pure states \cite{ContrerasSchiavina2022}.
Our results use this structure to connect orbit geometry to circuit pullbacks, optimization dynamics, and finite-shot stability.

Symmetry testing and learning can propose structural candidates upstream \cite{LaBorde2023TestingSymmetry,Lu2024LearningSymmetries,Sauvage2024SymmetricShadows,Meyer2023,Nguyen2024}.
The present framework begins with a specified candidate action or subspace and determines whether it is a state-level gauge, how the resulting quotient is represented by the circuit, and how uncertainty should be handled when exact structure cannot be certified.
Its conditioning statements remove coordinate and gauge effects while retaining physical excitation gaps, estimator-dependent measurement costs, and decoherence-induced signal loss.

\section{From redundant circuit coordinates to quotient state orbits}
\label{sec:quotient}

A singular circuit QFIM becomes an intrinsic optimization metric only after physical sector reduction is separated from parameter redundancy.
The physical symmetry selects the active Hilbert-space sector, the redundancy action removes equivalent parameter representatives, and the circuit differential transfers quotient velocities to the reachable state orbit.
The pullback mechanism connecting these three levels controls both the Fisher kernel and the pseudoinverse update.
Standard background on quantum Fisher information, natural gradients, smooth group actions, and quotient tensors is collected in Appendix~\ref{app:qfi_geometry} \cite{BraunsteinCaves1994,Petz2008,Lee2013,AbsilMahonySepulchre2008,Boumal2023}.

\subsection{Parametric models and symmetries}

Variational optimization over parametrized quantum circuits separates naturally into two ingredients: a map from parameters to quantum states and a scalar objective defined on those states.
The state map determines the Fisher geometry and the redundancies of the parametrization, while the objective selects the direction of optimization within that geometry.
Thus the quotient is determined by the represented state family rather than by the particular loss optimized on it.

Let $\cH$ be a finite-dimensional Hilbert space.
We work locally on an open parameter set $\Theta\subset\bbR^{p}$ equipped with its standard Euclidean inner product; orthogonality, pseudoinverses, and horizontal complements refer to this background structure unless stated otherwise.
A parameter vector $\theta\in\Theta$ specifies a unitary circuit $U(\theta)\in U(\cH)$ and an induced quantum state
\begin{equation}
\rho(\theta)=U(\theta)\rho_{0}U(\theta)^{\dagger},
\end{equation}
with a fixed reference state $\rho_{0}$.
For concreteness, we write the objective in expectation-value form,
\begin{equation}
\cL(\theta)=\Tr\!\left(O\,\rho(\theta)\right),
\end{equation}
where $O=O^{\dagger}$ is a target observable.
All quotient and pseudoinverse statements apply to any smooth loss that depends on the parameters only through the represented quantum state.

Within this state-space formulation, physical symmetries and parameter redundancies act at different levels.
We use $K$ for physical symmetries acting on the Hilbert space and $\Gamma$ for redundancy groups acting on the parameter manifold.
Physical symmetries such as conserved charges, permutation symmetry, or fixed-momentum constraints first reduce the effective Hilbert-space sector.
Independently, a chosen parametrization of the resulting variational family may contain smooth redundancies: different parameter values can induce the same represented quantum state.
$\Gamma$ describes this second structure and defines the parameter quotient.
Thus $K$ appears in representation-theoretic reductions and commutant ans\"atze, whereas $\Gamma$ appears in the quotient $\Theta/\Gamma$, the vertical tangent space $\cV_{\theta}$, and the gauge-drift discussion.
We use the state-level equivalence $\rho(g\cdot\theta)=\rho(\theta)$; a measurement-level quotient would instead identify only the statistics of a fixed measurement model.

\subsection{QFIM and QNG conventions}

Quantum Fisher information measures infinitesimal state distinguishability before a circuit parametrization is chosen.
For a parametrized family, the quantum Fisher information matrix (QFIM) is the coordinate representation of this state-space geometry after pullback through the state map; the symmetric-logarithmic-derivative (SLD) construction and its statistical interpretation are reviewed in Appendix~\ref{app:sld_qfi} \cite{BraunsteinCaves1994,Petz2008}.
Writing $g^{\mathrm{SLD}}$ for the state-space SLD quantum Fisher form, the pulled-back QFIM tensor is
\begin{equation}
F_{\theta}(v,w)
:=g^{\mathrm{SLD}}_{\rho(\theta)}\!\left(
 d\rho_{\theta}(v),d\rho_{\theta}(w)
\right),
\qquad v,w\in T_{\theta}\Theta.
\label{eq:qfim_pullback}
\end{equation}
We use $F_{\theta}$ for this bilinear form and $F(\theta)=[F_{ij}(\theta)]$ for its matrix in the coordinate basis $\{\partial_i\}$, so that
$F_{\theta}(v,w)=v^{\mathsf T}F(\theta)w$.
The equivalent SLD matrix formula and its behavior under reparametrization are given in Appendix~\ref{app:pullback_geometry}.

Most results concern pure-state sectors, where $\rho(\theta)=\ket{\psi(\theta)}\!\bra{\psi(\theta)}$ and the same convention gives the derivative-overlap and generator-covariance formulas.
Introduce the local Hermitian generators
\begin{equation}
A_{i}(\theta):= i \left(\partial_{i}U(\theta)\right)U(\theta)^{\dagger},
\qquad
\partial_i\ket{\psi}=-iA_i\ket{\psi},
\qquad
\partial_{i}:=\frac{\partial}{\partial \theta_{i}}.
\end{equation}
Then \cite{ProvostVallee1980,Stokes2020QNG}
\begin{align}
F_{ij}(\theta)
&=4\,\mathrm{Re}\!\left(
\braket{\partial_i\psi}{\partial_j\psi}
-\braket{\partial_i\psi}{\psi}\braket{\psi}{\partial_j\psi}
\right) \notag\\
&=4\,\mathrm{Re}\!\left(
\bra{\psi}A_iA_j\ket{\psi}
-\bra{\psi}A_i\ket{\psi}\bra{\psi}A_j\ket{\psi}
\right) \notag\\
&=2\,\bra{\psi}\{\Delta A_i,\Delta A_j\}\ket{\psi},
\label{eq:QFIM_pure}
\end{align}
where $\Delta A_i:=A_i-\bra{\psi}A_i\ket{\psi}I$ and $\{X,Y\}=XY+YX$.
Appendix~\ref{app:pure_qfi} derives these equivalent expressions and records the normalization convention.
Unless stated otherwise, Sections~\ref{sec:quotient}--\ref{sec:landscape_dynamics} and Section~\ref{sec:noise_trainability} use this pure-state specialization, while Section~\ref{sec:mixed_extension} develops the mixed-state SLD geometry on unitary orbits.

Quantum natural gradient (QNG) uses the QFIM to convert the coordinate differential of the loss into a state-geometrically meaningful update.
Where $F_{\theta}$ is nondegenerate, the coordinate step is
\begin{equation}
\Delta \theta_{\mathrm{QNG}}
=-\eta\,F(\theta)^{-1}\nabla \cL(\theta),
\end{equation}
with learning rate $\eta>0$.
When $F(\theta)$ is singular, the standard coordinate-level extension is
\begin{equation}
\Delta \theta_{\mathrm{QNG}}
=-\eta\,F(\theta)^{+}\nabla \cL(\theta),
\end{equation}
where $F(\theta)^{+}$ is the Moore--Penrose pseudoinverse \cite{Stokes2020QNG}.
The natural-gradient and minimum-norm interpretations of these formulas are summarized in Appendix~\ref{app:natural_gradient_background}.
The quotient construction identifies the precise condition under which this pseudoinverse step is the intrinsic quotient natural gradient: its kernel must contain exactly the redundant velocities.

\subsection{Representation-theoretic sector reduction}
\label{sec:rep}

Physical symmetry reduction isolates the Hilbert-space sectors within which state motion and Fisher geometry decouple.
Let $K$ be a compact group with a finite-dimensional unitary representation $V:K\to U(\cH)$.
Complete reducibility yields the isotypic decomposition
\begin{equation}
\cH \cong \bigoplus_{\lambda\in \widehat{K}} \left(V_{\lambda}\otimes \bbC^{m_{\lambda}}\right),
\label{eq:isotypic}
\end{equation}
where $V_{\lambda}$ runs over irreducible representations and $m_{\lambda}$ is its multiplicity.
The summand $\cH_{\lambda}:=V_{\lambda}\otimes \bbC^{m_{\lambda}}$ is the $\lambda$-\emph{isotypic component}; $V_\lambda$ is the irreducible representation space and $\bbC^{m_\lambda}$ its multiplicity space.
An ansatz restricted to the commutant acts trivially on $V_\lambda$ and mixes only the multiplicity space, which therefore becomes the effective training sector at fixed $\lambda$.
The commutant algebra has the block form \cite{FultonHarris1991,KeylWerner2001,BaconChuangHarrow2006}
\begin{equation}
\End_{K}(\cH):=\{X\in \End(\cH):[X,V(k)]=0\ \text{for all }k\in K\}
\cong \bigoplus_{\lambda\in \widehat{K}} \left(I_{V_{\lambda}}\otimes \End(\bbC^{m_{\lambda}})\right).
\label{eq:commutant}
\end{equation}

A $K$-invariant circuit family may consequently be written as
\begin{equation}
U(\theta)=\bigoplus_{\lambda}\left(I_{V_{\lambda}}\otimes U_{\lambda}(\theta^{(\lambda)})\right),
\qquad U_{\lambda}(\theta^{(\lambda)})\in U(m_{\lambda}),
\label{eq:block_ansatz}
\end{equation}
with independent parameter blocks $\theta^{(\lambda)}$.
The QFIM inherits the same separation.

\begin{proposition}
\label{prop:sector_blocks}
Consider the block ansatz \cref{eq:block_ansatz}.
\begin{enumerate}
\item[(a)]
If the pure reference state $\rho_{0}=\ket{\psi_{0}}\!\bra{\psi_{0}}$ is supported on a single isotypic component $\cH_{\lambda_{0}}$, then the pure-state QFIM for $\rho(\theta)=U(\theta)\rho_{0}U(\theta)^{\dagger}$ has only one possibly nonzero block:
\begin{equation}
F(\theta)=\bigoplus_{\lambda} F^{(\lambda)}(\theta^{(\lambda)}),\qquad
F^{(\lambda)} \equiv 0 \text{ for } \lambda\neq \lambda_{0}.
\label{eq:qfim_block}
\end{equation}
\item[(b)]
If $\rho_{0}=\bigoplus_{\lambda}\rho_{0,\lambda}$ is block diagonal with respect to \cref{eq:isotypic}, then the SLD QFIM is block diagonal across $\lambda$, and each block depends only on the corresponding pair $(U_{\lambda},\rho_{0,\lambda})$.
\end{enumerate}
\end{proposition}

Thus a pure state supported in one isotypic component has a single active Fisher block, while a block-diagonal mixed state retains independent sectorwise estimation and inversion.
Schur--Weyl duality gives the canonical example for $\cH=(\bbC^{d})^{\otimes n}$ and $K=SU(d)$.
For partitions $\lambda\vdash n$ with at most $d$ rows,
$\cH\simeq \bigoplus_{\lambda}Q_{\lambda}\otimes P_{\lambda}$, where $Q_\lambda$ and $P_\lambda$ are irreducible $SU(d)$ and symmetric-group representations, respectively.
The commutant is generated by the symmetric-group action, so commutant ans\"atze act on the multiplicity spaces $P_\lambda$; efficient circuits for the associated Schur transform are available \cite{BaconChuangHarrow2006,Krovi2019}.

Fixing an active $\lambda$-sector leaves a smaller state space $\cK_\lambda$ and its restricted dynamical group.
Parameter redundancies inside this sector are governed independently by the action of $\Gamma$.

\subsection{Redundancy actions and quotient metrics}

A QFIM singularity induced by parameter redundancy is the differential signature of a quotient geometry.
State-preserving parameter orbits generate vertical directions in the QFIM kernel; when these directions exhaust the kernel, optimization is governed by a Riemannian metric on the quotient, and the pseudoinverse update is its ambient-coordinate horizontal lift.
Assume that a Lie group $\Gamma$ acts smoothly on the parameter manifold through $\Phi:\Gamma\times\Theta\to\Theta$, written $g\cdot\theta$, and that each orbit represents one quantum state:
\begin{equation}
\rho(g\cdot \theta)=\rho(\theta)\qquad \forall g\in \Gamma,\ \theta\in\Theta.
\label{eq:state_invariant}
\end{equation}
Equivalently, the state map descends from the parameter manifold to the orbit space $\Theta/\Gamma$.
Since the objectives considered here depend on the parameters only through the represented state, they are constant along the same orbits:
\begin{equation}
\cL(g\cdot \theta)=\cL(\theta)\qquad \forall g\in \Gamma,\ \theta\in\Theta.
\label{eq:loss_invariant}
\end{equation}
The orbit through $\theta$ defines an equivalence class of parameters that induce the same state and therefore the same predictions for any observable.
We work on a $\Gamma$-invariant open set $U\subseteq\Theta$ on which the orbit type is constant, and we assume that $U/\Gamma$ is a smooth quotient manifold with canonical projection $\pi:U\to U/\Gamma$ a submersion \cite{Lee2013,AbsilMahonySepulchre2008,Boumal2023}.
When $\ker F(\theta)=\cV_{\theta}$, the descended tensor is positive definite on $U/\Gamma$; equivalently, $F$ is the pullback of the quotient metric along $\pi$.

The infinitesimal action of $\Gamma$ identifies the vertical directions that must lie in the QFIM kernel.
Let $\mathfrak{r}$ be the Lie algebra of $\Gamma$. Each $\xi\in\mathfrak{r}$ induces a fundamental vector field
\begin{equation}
X_{\xi}(\theta):=\left.\frac{d}{dt}\right|_{t=0}\bigl(\exp(t\xi)\cdot \theta \bigr)\in T_{\theta}\Theta.
\end{equation}
The orbit tangent space (vertical space) is
\begin{equation}
\cV_{\theta}:=\spanop\{X_{\xi}(\theta):\xi\in\mathfrak{r}\}\subseteq T_{\theta}\Theta.
\end{equation}
Differentiating \cref{eq:loss_invariant} yields $d\cL_{\theta}(v)=0$ for all $v\in \cV_{\theta}$.

By the standing state-invariance hypothesis, the state map is constant along each orbit.
Differentiating \cref{eq:state_invariant} along $g(t)=\exp(t\xi)$ gives
$d\rho_{\theta}(X_{\xi}(\theta))=0$.
Since the QFIM is the pullback in \cref{eq:qfim_pullback}, every $v\in\cV_{\theta}$ therefore satisfies, for all $w\in T_{\theta}\Theta$,
\begin{equation}
F_{\theta}(v,w)
=g^{\mathrm{SLD}}_{\rho(\theta)}\!\left(0,d\rho_{\theta}(w)\right)
=0.
\label{eq:vertical_qfi_kernel}
\end{equation}
Hence $\cV_{\theta}\subseteq\ker F_{\theta}$.
This is the pullback-kernel property summarized in Appendix~\ref{app:pullback_geometry}.

To formulate the quotient geometry precisely, for each $g\in \Gamma$ let
$\Phi_g:U\to U$ denote the diffeomorphism induced by the action,
$\Phi_g(\theta)=g\cdot\theta$.
The standard descent criterion for covariant tensors on a smooth quotient is reviewed in Appendix~\ref{app:quotient_preliminaries}.
Applied to the action maps $\Phi_g$, it motivates the following definition.

\begin{definition}
\label{def:basic_fisher}
Let $U\subseteq\Theta$ be a $\Gamma$-invariant open set such that $U/\Gamma$ is a
smooth quotient manifold and $\pi:U\to U/\Gamma$ is a submersion, and let $F$
denote the QFIM tensor induced by the state family on $U$.
We call $F$ $\Gamma$-\emph{basic} if the following two conditions hold:
\begin{enumerate}
\item[(i)] $F$ is invariant under the $\Gamma$-action:
\begin{equation}
F_{g\cdot\theta}\!\left((d\Phi_g)_{\theta}u,(d\Phi_g)_{\theta}v\right)
=F_{\theta}(u,v)
\end{equation}
for all $g\in \Gamma$, $\theta\in U$, and $u,v\in T_{\theta}U$, where
$(d\Phi_g)_{\theta}:T_{\theta}U\to T_{g\cdot\theta}U$ is the tangent map.
\item[(ii)] $F$ annihilates vertical directions:
\begin{equation}
\cV_{\theta}\subseteq\ker F_{\theta}
\end{equation}
for all $\theta\in U$, where
$\ker F_{\theta}:=\{u\in T_{\theta}U:F_{\theta}(u,\cdot)=0\}$.
\end{enumerate}
A $\Gamma$-basic QFIM is called \emph{faithful to the quotient} on $U$ if
\begin{equation}
\ker F_{\theta}=\cV_{\theta}\qquad \forall \theta\in U.
\label{eq:faithful_fisher}
\end{equation}
\end{definition}

Under the standing state-invariance hypothesis, $\rho\circ\Phi_g=\rho$.
Since the QFIM is obtained by pulling back the state-space quantum Fisher
tensor along $\rho$, the identity $\rho\circ\Phi_g=\rho$ implies
$\Phi_g^{*}F=F$.
Together with \cref{eq:vertical_qfi_kernel}, this
shows that the QFIM is automatically $\Gamma$-basic on every such set $U$.
Faithfulness is the additional requirement that the prescribed redundancy
group accounts for all null directions of the QFIM.

Because $F(\theta)$ is symmetric positive semidefinite, the Moore--Penrose pseudoinverse is well defined in any local coordinate chart.
Define the QFIM-identifiable complement
\begin{equation}
\mathsf{H}_{\theta}:=(\ker F(\theta))^{\perp},
\end{equation}
where the orthogonal complement is taken with respect to the Euclidean inner product in the chosen coordinates.
When $F$ is faithful to the quotient, $\mathsf{H}_{\theta}=\cV_{\theta}^{\perp}$ is precisely the Euclidean horizontal space.
If $F$ has additional null directions, $\mathsf{H}_{\theta}$ is smaller and contains only the directions detected by the QFIM.
Then $F(\theta)$ is positive definite on $\mathsf{H}_{\theta}$ and $\mathsf{H}_{\theta}=\mathrm{im}\,F(\theta)$.
Let $P_{\mathsf{H}}$ denote the Euclidean orthogonal projection onto $\mathsf{H}_{\theta}$.
With this choice,
\begin{equation}
F(\theta)^{+}=\left(F(\theta)\vert_{\mathsf{H}_{\theta}}\right)^{-1}P_{\mathsf{H}},
\label{eq:pinv_decomp}
\end{equation}
and $F(\theta)^{+}$ vanishes on $\ker F(\theta)$ (hence on $\cV_{\theta}$).

The faithful-to-the-quotient condition \cref{eq:faithful_fisher} becomes a generic rank condition in analytic models: once it holds at one point of a connected stratum, it holds away from a proper analytic exceptional set.

\begin{proposition}
\label{prop:faithful_generic}
Let $U\subseteq\Theta$ be a connected $\Gamma$-invariant open set on which
$s:=\dim\cV_{\theta}$ is constant.
Assume that the coefficient matrix of the $\Gamma$-basic QFIM tensor $F$ is real analytic on $U$ and that $F$ is faithful to the quotient at one point of $U$.
Then there is a proper real-analytic subset $Z\subset U$ such that $F$ is faithful to the quotient on $U\setminus Z$.
In particular, $Z$ has empty interior and Lebesgue measure zero, and $\rank F=p-s$ on $U\setminus Z$.
\end{proposition}
\begin{proof}
Since $\cV_{\theta}\subseteq\ker F_{\theta}$ with $\dim\cV_{\theta}=s$, one has
$\rank F_{\theta}\le p-s$ on $U$, with equality exactly at the faithful points.
Let $\phi(\theta)$ be the sum of squares of all $(p-s)\times(p-s)$ minors of $F_{\theta}$;
$\phi$ is real analytic, and $\phi(\theta)\neq0$ if and only if $\rank F_{\theta}=p-s$.
Faithfulness at one point gives a parameter $\theta_{0}\in U$ with $\phi(\theta_{0})\neq0$, so $\phi$ does not vanish identically on the connected set $U$.
Set
\begin{equation}
Z:=\{\theta\in U:\ker F_{\theta}\neq\cV_{\theta}\}=\phi^{-1}(0).
\end{equation}
Then $Z$ is a proper real-analytic zero set;
such sets are closed, Lebesgue null, and nowhere dense \cite{KrantzParks2002}.
\end{proof}

For standard finite-depth circuit families built from Pauli rotations and fixed entangling gates, the analyticity hypothesis is local and automatic: the entries of $U(\theta)$, the induced state, and the QFIM entries are trigonometric polynomials in the parameters.
Thus, in these models, checking $\rank F(\theta)=p-\dim\cV_{\theta}$ at one regular parameter point certifies the faithful quotient regime away from a proper analytic exceptional set.
On that open dense set, the locally constant-rank hypothesis in \cref{thm:quotient_geometry_update} holds automatically.

\subsection{Pseudoinverse lift on the quotient}

Basicness makes the QFIM descend to the parameter quotient, while faithfulness determines whether the descended tensor is a Riemannian metric and whether its gradient has the standard pseudoinverse lift.

\begin{theorem}
\label{thm:quotient_geometry_update}
Let $U\subseteq\Theta$ be a $\Gamma$-invariant open set such that $U/\Gamma$ is a smooth quotient manifold and $\pi:U\to U/\Gamma$ is a submersion.
Assume that the QFIM tensor $F$ is smooth, $\Gamma$-basic, and positive semidefinite on $U$, and that $\rank F_{\theta}$ is locally constant.
Then $F$ descends to a smooth positive semidefinite tensor $\bar F$ on $U/\Gamma$.
If $F$ is faithful to the quotient on $U$, then $\bar F$ is a Riemannian metric on $U/\Gamma$, and for every objective $\cL=\bar{\cL}\circ\pi$ the quotient natural-gradient step $-\eta\,\mathrm{grad}_{\bar F}\bar{\cL}([\theta])$ has the unique minimum-Euclidean-norm horizontal lift
\begin{equation}
\Delta\theta=-\eta\,F(\theta)^{+}\nabla \cL(\theta),
\label{eq:quotient_pinv}
\end{equation}
where $F(\theta)^{+}$ is the Moore--Penrose pseudoinverse.
\end{theorem}

When $\cV_{\theta}\subsetneq\ker F_{\theta}$, the descended QFIM tensor is only a quotient semimetric.
The extra null directions are additional infinitesimal state-preserving directions not generated by the prescribed $\Gamma$-action; equivalently, the induced state map on $U/\Gamma$ fails to be an immersion there.
They therefore require a larger local redundancy quotient before the QFIM defines a Riemannian metric.

In the faithful setting, the horizontal form of \cref{eq:quotient_pinv} is explicit.
With $\mathsf{H}_{\theta}=(\ker F(\theta))^{\perp}$ and $P_{\mathsf{H}}$ the Euclidean orthogonal projector,
\begin{equation}
\Delta\theta=-\eta\left(F(\theta)\vert_{\mathsf{H}_{\theta}}\right)^{-1}P_{\mathsf{H}}\nabla \cL(\theta).
\label{eq:quotient_horizontal}
\end{equation}
This lift projects to the quotient Riemannian gradient:
\begin{equation}
d\pi_{\theta}(\Delta\theta)=-\eta\,\mathrm{grad}_{\bar F}\,\bar{\cL}([\theta]).
\label{eq:quotient_projection_gradient}
\end{equation}
Thus \cref{eq:quotient_pinv} carries no first-order motion along the redundant directions.
This is an infinitesimal statement; finite-step representative drift still requires an explicit gauge choice or retraction, as discussed in \cref{subsec:gauge_fixing}.
For the same reason, isotropic damping $F\leftarrow F+\gamma I$ should be applied only after the known gauge kernel has been projected out: damping the full ambient matrix converts vertical finite-shot noise into representative drift at rate $1/\gamma^2$, as quantified in part~(a) of \cref{thm:finite_shot_amplification}.
Any additional exact null directions reflect state-level non-identifiability beyond the prescribed $\Gamma$-action; apparent null directions created by finite-shot estimation or numerical truncation are separate statistical effects.

\subsection{Circuit-to-orbit Fisher transfer}
\label{subsec:circuit_orbit_transfer}

The parameter quotient identifies redundant velocities but does not yet determine how the remaining circuit coordinates represent the intrinsic orbit geometry.
Fix the active sector $\cK_\lambda$ selected in \cref{sec:rep}, and suppose on the regular stratum that the smooth state map $\rho$ takes values locally in a reachable orbit $\cO_\lambda$.
The state differential itself supplies the required circuit-to-orbit map, so no group-valued unitary lift is needed.
We use the Hermitian-generator realization of the restricted dynamical Lie algebra:
$\mathfrak g_\lambda$ is a real vector space of traceless Hermitian operators, closed under
$[X,Y]_{\mathrm H}:=-i(XY-YX)$, and $G_\lambda$ is the connected group generated by
$\exp(-itX)$ with $X\in\mathfrak g_\lambda$.
Symbols such as $\mathfrak{su}(D)$ refer to this Hermitian realization unless a complexified root decomposition is stated explicitly.
Equip $\mathfrak g_\lambda$ with
\begin{equation}
\langle X,Y\rangle_{\mathrm{tr}}:=\Tr(XY).
\end{equation}
At $\rho=\rho(\theta)$ let
\begin{equation}
\mathfrak h_\rho:=\{X\in\mathfrak g_\lambda:[X,\rho]=0\},
\qquad
\mathfrak p_\rho:=\mathfrak h_\rho^{\perp_{\mathrm{tr}}},
\end{equation}
and write $P_{\mathfrak p_\rho}$ for the trace-orthogonal projector.
The infinitesimal orbit map
\begin{equation}
\mathcal A_\rho:\mathfrak p_\rho\longrightarrow T_\rho\cO_\lambda,
\qquad
\mathcal A_\rho(X)=-i[X,\rho],
\end{equation}
is a linear isomorphism: the kernel of the generator action is $\mathfrak h_\rho$, and restriction to its orthogonal complement removes that kernel.
Because $d\rho_\theta(v)\in T_\rho\cO_\lambda$, the state-level circuit-to-orbit differential is therefore defined without choosing a lift:
\begin{equation}
B_\theta:=\mathcal A_\rho^{-1}\circ d\rho_\theta:
T_\theta\Theta\longrightarrow\mathfrak p_\rho.
\label{eq:B_definition}
\end{equation}
Define the positive-definite intrinsic Fisher operator $M_\rho$ on $\mathfrak p_\rho$ by
\begin{equation}
\langle X,M_\rho Y\rangle_{\mathrm{tr}}
:=g_\rho^{\mathrm{SLD}}(\mathcal A_\rho X,\mathcal A_\rho Y).
\label{eq:intrinsic_fisher_operator}
\end{equation}
In a trace-orthonormal basis of $\mathfrak p_\rho$, $M_\rho$ is the orbit-frame QFIM.

When the circuit admits a unitary lift $U(\theta)$ taking values in the represented group $G_\lambda$, its local generator gives a convenient formula for the same differential:
\begin{equation}
C_\theta v:=i(dU_\theta v)U(\theta)^\dagger\in\mathfrak g_\lambda,
\qquad
d\rho_\theta(v)=-i[C_\theta v,\rho],
\qquad
B_\theta=P_{\mathfrak p_\rho}C_\theta.
\label{eq:circuit_generator_map}
\end{equation}
This group-valued lift is sufficient for direct generator-level computation, but it is not a hypothesis of the factorization.

\begin{theorem}
\label{thm:circuit_orbit_factorization}
Let $\rho$ be any smooth circuit state map whose image lies locally in $\cO_\lambda$.
Then for all $v,w\in T_\theta\Theta$,
\begin{equation}
F_\theta(v,w)=\langle B_\theta v,M_\rho B_\theta w\rangle_{\mathrm{tr}},
\qquad
F(\theta)=B_\theta^*M_\rho B_\theta.
\label{eq:circuit_orbit_factorization}
\end{equation}
Consequently,
\begin{equation}
\ker F(\theta)=\ker B_\theta.
\label{eq:factorization_kernel}
\end{equation}
If $\Gamma$ is a state-preserving redundancy group with vertical space $\cV_\theta$, then $\cV_\theta\subseteq\ker B_\theta$, and the QFIM is faithful to the prescribed quotient exactly when
\begin{equation}
\ker B_\theta=\cV_\theta.
\label{eq:B_faithfulness}
\end{equation}
\end{theorem}
\begin{proof}
By definition, $d\rho_\theta(v)=\mathcal A_\rho(B_\theta v)$.
Substitution into the pullback definition \cref{eq:qfim_pullback} and then \cref{eq:intrinsic_fisher_operator} gives \cref{eq:circuit_orbit_factorization}.
Since $M_\rho$ is positive definite on $\mathfrak p_\rho$,
\begin{equation}
F_\theta(v,v)=\|M_\rho^{1/2}B_\theta v\|_{\mathrm{tr}}^2,
\end{equation}
which proves \cref{eq:factorization_kernel}.
State invariance along $\Gamma$-orbits gives the final assertions.
\end{proof}

On a faithful constant-rank stratum, $B_\theta$ descends to an injective map
$\bar B_{[\theta]}:T_{[\theta]}\cQ\to\mathfrak p_\rho$.

\begin{corollary}
\label{cor:quotient_orbit_local_geometry}
The map $\bar\rho:\cQ\to\cO_\lambda$ is an immersion at $[\theta]$ if and only if the QFIM is faithful there.
It is a local diffeomorphism onto an open subset of $\cO_\lambda$ if and only if, in addition,
\begin{equation}
\operatorname{im}B_\theta=\mathfrak p_\rho.
\end{equation}
If the image is proper, quotient QNG is intrinsic to the circuit-reachable submanifold rather than to the entire homogeneous orbit.
\end{corollary}

The factorization also determines the state-space motion produced by a parameter-space pseudoinverse.
Let an orbit objective $\mathcal J$ have differential represented by $b_\rho\in\mathfrak p_\rho$,
\begin{equation}
d\mathcal J_\rho(\mathcal A_\rho X)=\langle b_\rho,X\rangle_{\mathrm{tr}},
\qquad
\nabla_\theta\cL=B_\theta^*b_\rho.
\end{equation}

\begin{proposition}
\label{prop:induced_orbit_update}
For
\begin{equation}
\delta\theta=-(B_\theta^*M_\rho B_\theta)^+B_\theta^*b_\rho,
\end{equation}
the induced orbit generator is
\begin{equation}
B_\theta\delta\theta
=-M_\rho^{-1/2}P_{\operatorname{im}(M_\rho^{1/2}B_\theta)}M_\rho^{-1/2}b_\rho.
\label{eq:induced_orbit_update}
\end{equation}
Thus the circuit implements the $M_\rho$-orthogonal projection of the full orbit natural gradient onto the reachable tangent subspace.
It equals the full orbit natural-gradient generator $-M_\rho^{-1}b_\rho$ if and only if
\begin{equation}
M_\rho^{-1}b_\rho\in\operatorname{im}B_\theta.
\label{eq:objective_transfer_condition}
\end{equation}
In particular, if $B_\theta$ is surjective, then
\begin{equation}
B_\theta(B_\theta^*M_\rho B_\theta)^+B_\theta^*=M_\rho^{-1},
\qquad
\delta\theta=-B_\theta^+M_\rho^{-1}b_\rho.
\label{eq:surjective_orbit_update}
\end{equation}
\end{proposition}
\begin{proof}
Set $A=M_\rho^{1/2}B_\theta$ and use
$A(A^*A)^+A^*=P_{\operatorname{im}A}$.
Equality in \cref{eq:induced_orbit_update} with the full generator $-M_\rho^{-1}b_\rho$ holds exactly when $M_\rho^{-1/2}b_\rho\in\operatorname{im}(M_\rho^{1/2}B_\theta)$, equivalently when \cref{eq:objective_transfer_condition} holds.
Surjectivity makes this condition automatic for every $b_\rho$; the last expression is then the minimum-Euclidean-norm solution of $B_\theta\delta\theta=-M_\rho^{-1}b_\rho$.
\end{proof}

\begin{corollary}
\label{cor:isotropic_spectral_transfer}
If $M_\rho=\alpha I_{\mathfrak p_\rho}$ with $\alpha>0$, then
\begin{equation}
F(\theta)=\alpha B_\theta^*B_\theta,
\qquad
\lambda_j^+(F)=\alpha\sigma_j(B_\theta)^2,
\qquad
\kappa^+(F)=\kappa(B_\theta)^2.
\label{eq:isotropic_spectral_transfer}
\end{equation}
The parameter-space QNG is a scalar multiple of the Euclidean horizontal gradient for every objective at $\theta$ if and only if
\begin{equation}
B_\theta^*B_\theta=s(\theta)^2P_{\mathsf H_\theta}.
\label{eq:conformal_chart_criterion}
\end{equation}
The rescaling is parameter independent on a region only when $s(\theta)$ is constant there.
\end{corollary}
\begin{proof}
The spectral identities follow from \cref{eq:circuit_orbit_factorization}.
On $\mathsf H_\theta=\operatorname{im}B_\theta^*$, the pseudoinverse acts as a scalar on every horizontal gradient exactly when all positive eigenvalues of $B_\theta^*B_\theta$ coincide, which is \cref{eq:conformal_chart_criterion}.
\end{proof}

Intrinsic isotropy therefore does not imply isotropy in arbitrary gate coordinates.
The constant-rescaling benchmark is recovered in a trace-orthonormal full orbit frame, for which $B_\theta=P_{\mathfrak p_\rho}$.

\begin{theorem}
\label{thm:conditioning_similarity}
Let $q_*\in \cQ$ map to $\rho_*\in\cO_\lambda$, assume that $\bar\rho$ is a local diffeomorphism at $q_*$, and let $d\mathcal J_{\rho_*}=0$.
Write $M_*$ for the intrinsic Fisher operator, $H_{\cO,*}$ for the Hessian bilinear form of $\mathcal J$, and $B_*:T_{q_*}\cQ\to\mathfrak p_{\rho_*}$ for the circuit-to-orbit differential.
Then
\begin{equation}
F_{\cQ,*}=B_*^*M_*B_*,
\qquad
H_{\cQ,*}=B_*^*H_{\cO,*}B_*,
\label{eq:metric_hessian_congruence}
\end{equation}
and
\begin{equation}
F_{\cQ,*}^{-1}H_{\cQ,*}=B_*^{-1}M_*^{-1}H_{\cO,*}B_*.
\label{eq:qng_similarity}
\end{equation}
Hence the linearized quotient-QNG operator is similar to the intrinsic operator $M_*^{-1}H_{\cO,*}$ and is independent of the singular values of the faithful circuit chart.
\end{theorem}
\begin{proof}
The first metric identity is \cref{thm:circuit_orbit_factorization} on the quotient.
At a critical point the second-derivative chain rule loses the term containing $d\mathcal J_{\rho_*}$, which gives the Hessian congruence.
Because $B_*$ is an isomorphism, direct inversion gives \cref{eq:qng_similarity}.
\end{proof}

QNG removes conditioning caused solely by a faithful change of circuit coordinates.
It does not remove the eigenvalue spread of the intrinsic objective operator $M_*^{-1}H_{\cO,*}$; the excitation-gap dependence derived later is one instance of this physical limitation.

\subsection{Slater and Grassmann quotient}

Slater determinants realize the quotient and transfer mechanisms in a familiar matrix model: orbital-frame rotations are exact gauge directions, and the physical state is the occupied subspace.
Let $Q=(q_1,\ldots,q_k)\in\operatorname{St}(k,n)$ and define
\begin{equation}
\ket{\Psi(Q)}=q_1\wedge\cdots\wedge q_k.
\end{equation}
The right action $Q\mapsto QW$, $W\in U(k)$, changes this vector only by $\det W$, so it preserves the quantum ray and gives the exact quotient
\begin{equation}
\operatorname{St}(k,n)/U(k)\simeq\operatorname{Gr}(k,n).
\end{equation}
Write a Stiefel tangent as
\begin{equation}
\dot Q=Q\Omega+Q_\perp K,
\qquad \Omega^\dagger=-\Omega.
\end{equation}
The first term is vertical and the second is Grassmann-horizontal.
For $P=QQ^\dagger$,
\begin{equation}
\dot P=Q_\perp KQ^\dagger+QK^\dagger Q_\perp^\dagger,
\end{equation}
and \cref{eq:QFIM_pure} gives
\begin{equation}
F(\dot Q_1,\dot Q_2)=2\Tr(\dot P_1\dot P_2)
=4\operatorname{Re}\Tr(K_1^\dagger K_2).
\label{eq:slater_grassmann_metric}
\end{equation}
For gate coordinates $\vartheta$ with horizontal Jacobian $J_\vartheta$, the pullback is $F_\vartheta=4J_\vartheta^{\mathsf T}J_\vartheta$ in this normalization.
Slater circuits therefore provide a concrete case in which the intrinsic metric is isotropic while gate-coordinate conditioning is determined entirely by the circuit chart.

\section{Exact Fisher geometry on symmetry-reduced orbits}
\label{sec:exact_spectrum}

The transfer law reduces circuit-level Fisher geometry to the intrinsic operator on a selected sector orbit. Full control supplies a rigid projective-space benchmark. Under irreducible restricted control, a moment-map identity makes the total Fisher weight the quadratic generalized-entanglement purity deficit, while the root decomposition resolves how that weight is distributed among physical tangent directions. Minuscule representations equalize the active root-plane scales; for a simple simply-laced orbit algebra, this gives a single trace-normalized Fisher scale.

\subsection{Full-control orbit-frame benchmark}
\label{subsec:setup_sector}

Fix an effective training sector $\cK_{\lambda}$ associated with a chosen physical symmetry label $\lambda$, and assume
\begin{equation}
\dim\cK_{\lambda}=D,\qquad
\mathfrak g_{\lambda}=\mathfrak{su}(\cK_{\lambda}),\qquad
\Tr(T_aT_b)=\delta_{ab}.
\label{eq:sector_core_assumptions}
\end{equation}
In the commutant setting of \cref{sec:rep}, this sector is typically the multiplicity space $\bbC^{m_{\lambda}}$ inside the $\lambda$-isotypic component.
Let $\mathfrak g_{\lambda}$ be the restricted dynamical Lie algebra generated by the available sector-preserving control Hamiltonians in the Hermitian convention fixed in \cref{subsec:circuit_orbit_transfer}, and let $G_{\lambda}\subseteq U(\cK_{\lambda})$ be its connected dynamical group.
The state is a pure vector $\ket{\psi}\in\cK_{\lambda}$.
The equality $\mathfrak g_{\lambda}=\mathfrak{su}(\cK_{\lambda})\cong \mathfrak{su}(D)$ is the full-control condition.
Under this condition, the restricted dynamics acts transitively on rays in the sector, so the reachable pure-state manifold from any reference ray is the whole projective space $\bbC\mathrm{P}^{D-1}$.
This is the regime in which the entire sector orbit is available, and its real dimension $2D-2$ is also the maximal possible rank of a pure-state QFIM on the sector.

Let $\{T_a\}_{a=1}^{D^2-1}$ be an orthonormal basis of $\mathfrak{su}(D)$ consisting of traceless Hermitian generators and satisfying the normalization in \cref{eq:sector_core_assumptions}.
We regard these generators as acting on $\cK_{\lambda}$ and, by extension, on $\cH$ with support on $\cK_{\lambda}$.
In these coordinates the pure-state QFIM is
\begin{equation}
F_{ab}(\psi)
=4\,\mathrm{Re}\!\left(\bra{\psi}T_aT_b\ket{\psi}-\bra{\psi}T_a\ket{\psi}\,\bra{\psi}T_b\ket{\psi}\right)
=2\,\bra{\psi}\left\{\Delta T_a,\Delta T_b\right\}\ket{\psi}.
\label{eq:qfim_suD}
\end{equation}

Wilson et al. established the orbitwise spectral invariance and the full-observable pure-state fixed-value/zero structure for a QFIM built from a Hilbert--Schmidt-orthonormal Lie-algebra basis \cite{Wilson2026GeometricInvariants}.  In the normalization of \cref{eq:qfim_suD}, their result gives
\begin{equation}
\operatorname{spec}F(\psi)
=\left\{
\underbrace{0,\dots,0}_{(D-1)^{2}\ \mathrm{times}},\,
\underbrace{2,\dots,2}_{(2D-2)\ \mathrm{times}}
\right\},
\qquad
\rank F(\psi)=2D-2,
\qquad
F(\psi)|_{\mathrm H}=2I_{\mathrm H}.
\label{eq:full_control_spectrum}
\end{equation}
Here $\mathrm H$ is the horizontal tangent space of
$\bbC\mathrm{P}^{D-1}\cong SU(D)/S(U(1)\times U(D-1))$.

In this trace-orthonormal orbit frame, the Fisher pseudoinverse is a constant preconditioner on the physical tangent space.

\begin{corollary}
\label{cor:free_ng}
Let $\cL$ be a smooth objective that depends on a pure state $\ket{\psi}\in\cK_{\lambda}$ only through the ray $\bbC\ket{\psi}$,
and let $F(\psi)$ be the pure-state QFIM in the $\mathfrak{su}(D)$ generator coordinates.
In the full-control setting of \cref{eq:full_control_spectrum}, the Euclidean gradient of $\cL$ is orthogonal to $\ker(F(\psi))$ and therefore lies in $\mathrm{range}(F(\psi))=(\ker(F(\psi)))^{\perp}$.
Consequently the quotient-space QNG direction satisfies
\begin{equation}
F(\psi)^{+}\nabla \cL=\frac{1}{2}\,\nabla \cL,
\end{equation}
and the QNG update is $\Delta\theta_{\mathrm{QNG}}=-(\eta/2)\nabla \cL$ in the generator coordinates.
Thus QNG is the constant rescaling $1/2$ in this specific trace-orthonormal Lie-generator frame.
\end{corollary}

\begin{proof}
Let $V\in\ker(F(\psi))$ and write $V=\sum_{a} v_{a}T_{a}$.
For pure states, $V\in\ker(F(\psi))$ is equivalent to $\Var_{\psi}(V)=0$, which holds if and only if
$V\ket{\psi}=c\ket{\psi}$ for some real $c$.
The unitary curve $e^{-itV}\ket{\psi}$ therefore remains on the same ray $\bbC\ket{\psi}$.
Since $\cL$ depends only on the ray, the function $t\mapsto \cL(e^{-itV}\ket{\psi})$ is constant and its derivative at $t=0$ vanishes.
In generator coordinates, this derivative equals the directional derivative of $\cL$ along $V$, which is $\sum_{a} v_{a}\,\partial_{a}\cL$.
Hence $\sum_{a} v_{a}\,\partial_{a}\cL=0$ for every $V\in\ker(F(\psi))$, proving that $\nabla\cL\perp\ker(F(\psi))$.
By \cref{eq:full_control_spectrum}, $F=2P_{\mathrm H}$ on the horizontal tangent space, so $F^{+}=\frac12P_{\mathrm H}$ and the displayed update follows.
\end{proof}

For an expectation-value objective $\cL(\psi)=\bra{\psi}O_{\lambda}\ket{\psi}$, one has
$\partial_{a}\cL = i\,\bra{\psi}[T_{a},O_{\lambda}]\ket{\psi}$.
The cited spectrum \cref{eq:full_control_spectrum} captures a geometric rigidity of $\bbC\mathrm{P}^{D-1}$ as a rank-one compact Hermitian symmetric space in the full-control sector regime \cref{eq:sector_core_assumptions}.
The normalization in \cref{eq:sector_core_assumptions} fixes the numerical value $2$.
With alternative conventions (for example $\Tr(T_{a}T_{b})=\frac{1}{2}\delta_{ab}$),
the constant rescales accordingly.
For a general gate-coordinate chart, \cref{thm:circuit_orbit_factorization,cor:isotropic_spectral_transfer} instead give
\begin{equation}
F_\theta=2B_\theta^*B_\theta,
\qquad
\lambda_j^+(F_\theta)=2\sigma_j(B_\theta)^2.
\end{equation}
The constant-rescaling statement therefore survives pullback exactly for conformal circuit charts.

\subsection{Orbit invariants and highest-weight flag geometry}
\label{sec:trace-moment}

The exact spectrum in \cref{eq:full_control_spectrum} is special to full control.
Under irreducible partial control, individual eigenvalues may vary while their sum remains an exact function of the moment map on every group orbit.

\begin{proposition}
\label{prop:trace-moment}
Let $\mathfrak g_\lambda$ act irreducibly on $\cK_\lambda$ through Hermitian generators $\{T_a\}_{a=1}^{\dim\mathfrak g_\lambda}$ with $\Tr(T_aT_b)=\delta_{ab}$, and let $c_\lambda$ be the Casimir scalar defined by
\begin{equation}
\sum_a T_a^2=c_\lambda I.
\end{equation}
Define the moment map $\mu\colon \mathbb P(\cK_\lambda)\to\bbR^{\dim\mathfrak g_\lambda}$ by
\begin{equation}
\mu_a(\psi)=\bra{\psi}T_a\ket{\psi}.
\end{equation}
Then for every pure state $\ket{\psi}\in\cK_\lambda$,
\begin{equation}
\Tr F(\psi)=4\bigl(c_\lambda-\|\mu(\psi)\|^2\bigr).
\label{eq:trace-moment}
\end{equation}
\end{proposition}

\begin{proof}
By the convention of \cref{eq:QFIM_pure}, the diagonal entries are variances,
\begin{equation}
F_{aa}(\psi)=4\left(\bra{\psi}T_a^2\ket{\psi}-\bra{\psi}T_a\ket{\psi}^2\right).
\end{equation}
Summing over $a$ and using $\sum_aT_a^2=c_\lambda I$ gives \cref{eq:trace-moment}.
\end{proof}

\begin{corollary}
\label{cor:coherent-min}
Let $G_\lambda$ be the connected group generated by $\{\exp(-itX):X\in\mathfrak g_\lambda,\ t\in\bbR\}$.
Then $\|\mu\|^2$ is constant on every $G_\lambda$-orbit, so $\Tr F$ is an orbit invariant.  On an orbit of positive Fisher rank $r$, the mean nonzero QFIM eigenvalue equals
\begin{equation}
\frac{4(c_\lambda-\|\mu\|^2)}{r}.
\end{equation}
The maximizers of $\|\mu(\psi)\|^2$ over unit rays in $\cK_\lambda$ are exactly the generalized-coherent states, namely the highest-weight orbit \cite{DelbourgoFox1977,Perelomov1986,Barnum2003}.
Consequently that orbit is Fisher-minimal among the $G_\lambda$-orbits.
With
\begin{equation}
E_{\mathfrak g_\lambda}(\psi):=c_\lambda-\|\mu(\psi)\|^2,
\end{equation}
\cref{eq:trace-moment} reads $\Tr F=4E_{\mathfrak g_\lambda}$; up to the normalization convention for $\{T_a\}$, this is the quadratic generalized-entanglement purity deficit of Refs.~\cite{Barnum2003,Somma2004}.
\end{corollary}

\begin{proof}
The moment map is equivariant, $\mu(U\psi)=\mathrm{Ad}^{*}_U\mu(\psi)$ for $U\in G_\lambda$, and the coadjoint action is orthogonal for the trace inner product in which $\{T_a\}$ is orthonormal.
Thus $\|\mu\|$ is preserved along each orbit.
The coherent-state minimal-uncertainty statement follows from the standard variational characterization of highest-weight vectors \cite{DelbourgoFox1977,Perelomov1986,Barnum2003}.
\end{proof}

For full control, $\mathfrak g_\lambda=\mathfrak{su}(D)$ acting on $\cK_\lambda\simeq\bbC^D$, one has $c_\lambda=(D^2-1)/D$ and $\|\mu(\psi)\|^2=1-1/D$ for every pure state.
Then \cref{eq:trace-moment} gives $\Tr F=4(D-1)$, matching the $2D-2$ nonzero eigenvalues equal to $2$ in \cref{eq:full_control_spectrum}.
Under partial control, different orbit strata can have different Fisher trace and rank; \cref{prop:trace-moment,cor:coherent-min} identify the coherent stratum as the minimal-total-Fisher stratum.

This invariant also separates Fisher conditioning from gradient concentration.  For any horizontal objective gradient $g\in\operatorname{range}F$,
\begin{equation}
\|\Delta_{\mathrm{QNG}}\|
=\eta\|F^+g\|
\ge \frac{\eta\|g\|}{\lambda_{\max}(F)}
\ge \frac{\eta\|g\|}{\Tr F}
=\frac{\eta\|g\|}{4E_{\mathfrak g_\lambda}}.
\label{eq:trace_qng_amplification}
\end{equation}
Coherent orbits minimize $E_{\mathfrak g_\lambda}$ and therefore maximize this trace-based lower bound on inverse-Fisher amplification.  The bound remains proportional to $\|g\|$ and supplies no positive lower bound on the gradient itself, so it does not rule out a barren plateau \cite{McClean2018,Larocca2025Barren}.  In the full-control case $E_{\mathfrak{su}(D)}=D-1$, and the bound weakens with sector dimension.

When the restricted dynamical Lie algebra is compact semisimple but smaller than $\mathfrak{su}(D)$, a highest-weight reference ray sweeps out a generalized flag manifold rather than the whole projective space.
The QFIM may then be anisotropic, but its complete horizontal spectrum remains exactly computable in root-space coordinates.

\begin{theorem}
\label{thm:flag_root_spectrum}
Let $\mathfrak g_{\lambda}$ be a compact semisimple restricted dynamical Lie algebra acting irreducibly on $\cK_{\lambda}$ through traceless Hermitian generators, and let $\ket{\Lambda}\in\cK_{\lambda}$ be a highest-weight vector with highest weight $\Lambda$.
Choose a Cartan subalgebra $\mathfrak t\subset\mathfrak g_{\lambda}$, a positive root system $\Delta^{+}$ for the complexified algebra, and root vectors $E_{\pm\alpha}$ normalized by $[E_{\alpha},E_{-\alpha}]=H_{\alpha}$.
For each $\alpha\in\Delta^{+}$ define Hermitian root generators
\begin{equation}
X_{\alpha}:=\frac{1}{\sqrt{2}}(E_{\alpha}+E_{-\alpha}),\qquad
Y_{\alpha}:=\frac{-i}{\sqrt{2}}(E_{\alpha}-E_{-\alpha}),
\end{equation}
and set $n_{\alpha}:=\Tr_{\cK_{\lambda}}(X_{\alpha}^{2})=\Tr_{\cK_{\lambda}}(Y_{\alpha}^{2})$.
In the trace-normalized basis $\widehat X_{\alpha}:=X_{\alpha}/\sqrt{n_{\alpha}}$ and $\widehat Y_{\alpha}:=Y_{\alpha}/\sqrt{n_{\alpha}}$, the pure-state QFIM at $\ket{\Lambda}$ is diagonal in the Cartan-root decomposition.
It vanishes on the Cartan directions and on every root pair satisfying $\Lambda(H_{\alpha})=0$.
For each positive root with $\Lambda(H_{\alpha})>0$, the two horizontal directions carry the common eigenvalue
\begin{equation}
F(\widehat X_{\alpha},\widehat X_{\alpha})
=F(\widehat Y_{\alpha},\widehat Y_{\alpha})
=2\,\frac{\Lambda(H_{\alpha})}{n_{\alpha}}.
\label{eq:flag_root_spectrum}
\end{equation}
Thus the reachable orbit through $[\Lambda]$ is a generalized flag manifold whose horizontal QFIM geometry is resolved root by root.
By equivariance of the group action, the same nonzero Fisher scales hold at every point of the orbit after transporting the root-space tangent frame.
\end{theorem}

Theorem~\ref{thm:flag_root_spectrum} gives the correct replacement for the isotropic full-control case: each active root direction has its own Fisher scale, fixed by the highest weight and by the chosen generator normalization.
The full $\mathfrak{su}(D)$ projective-space theorem is recovered when the active roots are precisely the transitions from the reference line to its orthogonal complement and all normalized scales in \cref{eq:flag_root_spectrum} equal $2$.

\subsection{Minuscule root scales and fermionic isotropy}
\label{sec:minuscule}

The root formula isolates a broad partial-control class in which the active root-plane scales equalize.
For a minuscule highest weight, every active root plane has Fisher value $2$ in an $\mathfrak{sl}_2$-normalized frame.
Under trace normalization, a simply-laced semisimple algebra has one constant on each simple factor, and the full horizontal metric is isotropic exactly when these factor constants coincide.
The Slater and fermionic-Gaussian families considered here satisfy this condition while having dynamical Lie algebras of only polynomial dimension.

\begin{corollary}
\label{cor:minuscule}
Assume the setting of \cref{thm:flag_root_spectrum} with $\mathfrak g_\lambda$ simply laced and the highest weight minuscule.
Write $\mathfrak g_\lambda=\bigoplus_{s=1}^{q}\mathfrak g_{\lambda,s}$ as a direct sum of simple ideals, with root systems $\Delta_s$.
Fix for each active positive root $\alpha$ an $\mathfrak{sl}_2$-normalized pair $E_{\pm\alpha}$ with $[E_\alpha,E_{-\alpha}]=H_\alpha$, and take the Hermitian root generators
\begin{equation}
X_\alpha=\frac{E_\alpha+E_{-\alpha}}{\sqrt2},
\qquad
Y_\alpha=\frac{-i(E_\alpha-E_{-\alpha})}{\sqrt2}.
\end{equation}
Then at every point of the orbit $\mathcal O_\lambda$ the QFIM in this root-adapted normalization equals
\begin{equation}
F=2P_{\mathrm H},
\label{eq:minuscule-isotropy}
\end{equation}
where $P_{\mathrm H}$ is the horizontal projector and
\begin{equation}
\rank F=2\,\#\{\alpha>0:\Lambda(H_\alpha)\ne0\}=2\dim_{\bbC}\mathcal O_\lambda.
\end{equation}
For an active root $\alpha\in\Delta_s$, define
\begin{equation}
\iota_{\lambda,s}:=\Tr_{\cK_\lambda}(E_\alpha E_{-\alpha}).
\end{equation}
This value is independent of the choice of root inside the simple factor $\mathfrak g_{\lambda,s}$.
In the trace-orthonormal convention $\Tr(T_aT_b)=\delta_{ab}$, every active root plane in that factor has Fisher scale $2/\iota_{\lambda,s}$.
Hence the trace-normalized intrinsic QFIM is isotropic on the full horizontal space if and only if $\iota_{\lambda,s}$ is the same for all simple factors containing active roots; in particular, this holds when $\mathfrak g_\lambda$ is simple.
For gate coordinates, the actual spectrum is obtained from \cref{cor:isotropic_spectral_transfer} and need not be isotropic.
\end{corollary}

\begin{proof}
By \cref{thm:flag_root_spectrum}, the nonzero root-plane Fisher scales are $2\Lambda(H_\alpha)$ in the root-adapted normalization.
Minuscule weights satisfy $\Lambda(H_\alpha)\in\{0,1\}$ for all positive roots, giving \cref{eq:minuscule-isotropy}.
The root system of a simply-laced semisimple algebra is a disjoint union of simply-laced irreducible components.
Within each component, the Weyl group acts transitively on roots and the represented trace form is invariant, so $\Tr_{\cK_\lambda}(E_\alpha E_{-\alpha})=\iota_{\lambda,s}$ depends only on the simple factor.
Trace-orthonormal rescaling therefore gives the factorwise scale $2/\iota_{\lambda,s}$, and these scales agree on the full horizontal space exactly when the active factor constants coincide.
Equivariance transports the same spectrum along the orbit.
\end{proof}

\begin{proposition}
\label{prop:fermionic}
Let $n$ denote the number of fermionic modes.
\begin{enumerate}
\item Number-conserving free-fermion circuits have physical dynamical Lie algebra $\mathfrak u(n)$ generated by bilinears $c_j^\dagger c_k$. On the fixed $k$-particle sector, the central $\mathfrak u(1)$ acts only by a global phase on rays, so the effective traceless orbit algebra in our convention is $\mathfrak{su}(n)$ acting in the minuscule representation $\Lambda^k\bbC^n$. The orbit of a Slater determinant is the Grassmannian $\mathrm{Gr}(k,n)$ of complex dimension $k(n-k)$. With hopping generators
\begin{equation}
\frac{c_j^\dagger c_k+c_k^\dagger c_j}{\sqrt2},
\qquad
\frac{i(c_j^\dagger c_k-c_k^\dagger c_j)}{\sqrt2},
\end{equation}
the nonzero QFIM spectrum on the Slater orbit is the constant $2$ with multiplicity $2k(n-k)$. For $1\le k\le n-1$, the trace-orthonormal constant is $2/\binom{n-2}{k-1}$.
\item Matchgate circuits have dynamical Lie algebra $\mathfrak{so}(2n)$ generated by Majorana bilinears $T_{pq}=\frac i2\gamma_p\gamma_q$. On a fixed parity sector this is the minuscule Weyl-spinor representation; the orbit of the Fock vacuum is the manifold of pure fermionic Gaussian states, of complex dimension $n(n-1)/2$. In the complex-fermion frame adapted to the vacuum, the normalized active generators are
\begin{equation}
X_{jk}=\frac{c_j^\dagger c_k^\dagger+c_kc_j}{\sqrt2},
\qquad
Y_{jk}=\frac{i(c_j^\dagger c_k^\dagger-c_kc_j)}{\sqrt2},
\qquad j<k.
\end{equation}
In this root-adapted frame the nonzero QFIM spectrum on the Gaussian orbit is the constant $2$ with multiplicity $n(n-1)$. Equivalently, in the orthogonal Majorana-bilinear coefficient frame the full matrix is $F=2P_{\mathrm H}$ of rank $n(n-1)$. In the trace-orthonormal convention the constant is $2/2^{n-3}$ for $n\ge3$; for $n=2$, the fixed-parity orbit is rank one and the constant is $2$.
\end{enumerate}
\end{proposition}

The individual Majorana bilinears are not themselves a root-adapted $\mathfrak{sl}_2$ frame. At the vacuum, a single bilinear can mix an active pair-creation component with a number-conserving stabilizer component and therefore have diagonal QFIM entry $1$ rather than $2$; the spectral statement in \cref{prop:fermionic} concerns the horizontal operator after these components are resolved.

\begin{proof}
The representation-theoretic identifications are standard: $\Lambda^k\bbC^n$ is the fundamental minuscule representation of type $A_{n-1}$, while the fixed-parity fermionic Fock space carries a Weyl-spinor minuscule representation of type $D_n$ \cite{FultonHarris1991}.
The active-root counts give $\dim_{\bbC}\mathrm{Gr}(k,n)=k(n-k)$ and $\dim_{\bbC}\mathcal O_{\mathrm{Gauss}}=n(n-1)/2$.
For the spinor orbit, $X_{jk}\ket{0}=c_j^\dagger c_k^\dagger\ket{0}/\sqrt2$ and $Y_{jk}\ket{0}=i c_j^\dagger c_k^\dagger\ket{0}/\sqrt2$, so each active root-plane variance is $1/2$ and the QFIM value is $2$.  By contrast, a single $T_{pq}$ generally combines an active pair-creation direction with a number-conserving stabilizer direction, which is why its individual variance does not determine the horizontal spectrum.
For $1\le k\le n-1$, the trace factor for $\Lambda^k\bbC^n$ is $\binom{n-2}{k-1}$.  For the fixed-parity spinor representation it is $2^{n-3}$ when $n\ge3$; the $n=2$ value is $1$ by direct evaluation.  These factors give the stated trace-normalized constants through \cref{cor:minuscule}.
\end{proof}

\section{Quotient-QNG landscapes and dynamics}
\label{sec:landscape_dynamics}

The orbit Fisher operator fixes how QNG measures motion, while the objective Hessian fixes which physical directions contract.
Their interaction determines critical conditioning, convergence of parameter representatives, and the exact dynamics available on homogeneous orbits.
Exact redundancies first appear as Hessian zero modes; after they are removed on a slice, maximal-flat and full-control geometries yield closed optimization laws.

\subsection{Gauge-fixed critical geometry}
\label{subsec:gauge_fixing}

The pseudoinverse lift is intrinsically first order, but the same redundancy reappears in the Hessian at a critical point.
The Hessian is flat along every redundancy orbit and therefore has infinite ambient condition number.
A transverse slice removes these zero modes, and retraction to that slice turns convergence on the quotient into convergence of the stored parameter representative.
Let $g$ be a nondegenerate background Riemannian metric used to define the Hessian operator and the orthogonal slice; it is distinct from the generally degenerate QFIM tensor $F$.
The redundancy action is assumed to be isometric with respect to $g$, a condition available for compact $\Gamma$ by averaging.

The Riemannian Hessian below is taken with respect to the background metric $g$, using the standard Levi--Civita definition recalled in Appendix~\ref{app:natural_gradient_background}.

\begin{lemma}
\label{lem:zero_modes_transversality}
Let $(\Theta,g)$ be a Riemannian manifold and let $\Gamma$ act smoothly on $\Theta$ by isometries.
Let $f:\Theta\to\bbR$ satisfy $f(g\cdot\theta)=f(\theta)$ for all $g\in \Gamma$.
\begin{enumerate}
\item[(a)]
At any critical point $\theta_{\star}$ of $f$, the orbit tangent space $\cV_{\theta_{\star}}$ lies in the Hessian kernel:
\begin{equation}
\mathrm{Hess}\,f(\theta_{\star})(v,\cdot)=0\qquad \text{for all } v\in \cV_{\theta_{\star}}.
\end{equation}
Consequently, the Hessian has at least $\dim(\cV_{\theta_{\star}})$ zero eigenvalues.
\item[(b)]
Assume in addition that the quotient Hessian at $[\theta_{\star}]\in \Theta/\Gamma$ is positive definite.
Equivalently, if $H$ denotes the Hessian operator at $\theta_{\star}$ and
$T_{\theta_{\star}}\Theta=\cV_{\theta_{\star}}\oplus \mathsf{H}_{\theta_{\star}}$ is the $g$-orthogonal decomposition, assume that $H\vert_{\mathsf{H}_{\theta_{\star}}}$ is positive definite.
By part~(a), this is equivalent to $H\succeq0$ and $\ker H=\cV_{\theta_{\star}}$.
Then the condition number on the gauge-fixed slice is
\begin{equation}
\kappa_{\mathrm{slice}}=\frac{\lambda_{\max}(H\vert_{\mathsf{H}_{\theta_{\star}}})}{\lambda_{\min}(H\vert_{\mathsf{H}_{\theta_{\star}}})},
\end{equation}
whereas the ambient condition number is infinite.
\end{enumerate}
\end{lemma}

The lemma identifies exactly which directions carry no second-order information at a critical point.
A quadratic gauge penalty may lift those directions without changing the physical Hessian spectrum.

A quadratic gauge penalty can lift the vertical Hessian zero modes without changing the physical Hessian spectrum; the exact condition-number formula is recorded in \cref{prop:soft_gauge}.

\begin{theorem}
\label{thm:representative_convergence}
Assume the faithful quotient setting of \cref{thm:quotient_geometry_update}, let $\theta_{\star}\in U$, and choose a smooth local slice $S\subset U$ through $\theta_{\star}$ such that
\[
\pi\vert_{S}:S\to V
\]
is a diffeomorphism onto an open set $V\subset U/\Gamma$ containing $[\theta_{\star}]$.
Write $s:=(\pi\vert_{S})^{-1}:V\to S$.
\begin{enumerate}
\item[(a)]
Let $K\subset V$ be compact, and let $\gamma:[0,\infty)\to K$ be an absolutely continuous quotient trajectory of finite Riemannian length.
Then the gauge-fixed representative $\vartheta(t):=s(\gamma(t))$ has finite Euclidean length, with
\begin{equation}
\operatorname{Len}_{\bbR^{p}}(\vartheta)
\le C_{K}\,\operatorname{Len}_{\bar F}(\gamma)
\end{equation}
for some $C_{K}<\infty$ depending only on $K$ and the chosen slice.
If $\gamma(t)\to [\theta_{\star}]$ in $V$, then $\vartheta(t)\to \theta_{\star}$ in ambient coordinates.
\item[(b)]
Suppose in addition that $V$, the descended metric $\bar F$, and the descended objective $\bar\cL:V\to\bbR$ are real analytic.
Let $\gamma:[0,\infty)\to V$ solve
\begin{equation}
\dot\gamma(t)=-\eta\,\grad_{\bar F}\bar\cL(\gamma(t)),
\qquad \eta>0,
\end{equation}
and assume that $\gamma([0,\infty))$ has compact closure in $V$.
Then $\gamma$ has finite $\bar F$-length and converges to a single critical point $[\theta_{\infty}]$ of $\bar\cL$.
The gauge-fixed representative $\vartheta(t)=s(\gamma(t))$ consequently has finite Euclidean length and converges to $s([\theta_{\infty}])$ in ambient coordinates.
\end{enumerate}
\end{theorem}

Finite-depth circuit expectation losses supply the analytic input in part~(b) on the principal stratum, so the conclusion follows from the Lojasiewicz--Simon theory for analytic gradient flows \cite{Simon1983,AbsilMahonyAndrews2005}.
With representative zero modes separated, the remaining convergence laws are intrinsic to the state orbit.

\subsection{Integrable fidelity flow on cominuscule orbits}
\label{sec:cominuscule_fidelity}

For rank $r>1$, target fidelity depends on $r$ principal angles, so a closed scalar evolution is not automatic. On irreducible compact Hermitian symmetric orbits with a cominuscule projective embedding, the maximal-flat Fisher block is isotropic and the strongly orthogonal root rotations commute. These two facts force all active principal defects to share one contraction factor: the higher-rank fidelity flow becomes completely integrable, with the familiar full-control logistic law as its rank-one face.

\begin{lemma}
\label{lem:hc_frame}
Let $\mathcal O_{\lambda}=G_{\lambda}\cdot[\Lambda]\subset\mathbb P(\cK_{\lambda})$ be an irreducible compact Hermitian symmetric orbit of rank $r$, embedded by the irreducible representation with highest weight
\begin{equation}
\Lambda=\ell\omega_{\mathrm c},\qquad \ell\in\bbN,
\label{eq:embedding_index}
\end{equation}
where $\omega_{\mathrm c}$ is the cominuscule fundamental weight, and choose $[\Lambda]$ as the target ray.
There are pairwise strongly orthogonal noncompact roots $\gamma_{1},\dots,\gamma_{r}$, all of the same maximal root length, such that no signed sum with distinct indices
\begin{equation}
\pm\gamma_{a_{1}}\pm\cdots\pm\gamma_{a_{k}},
\qquad k\ge2,
\end{equation}
is a root, and
\begin{equation}
\langle\Lambda,\gamma_{a}^{\vee}\rangle=\ell
\qquad (a=1,\dots,r).
\end{equation}
Moreover every ray in $\mathcal O_{\lambda}$ can be represented, after applying an element of the target stabilizer, in the flat form
\begin{equation}
\psi(\theta)=
\exp\!\left(\sum_{a=1}^{r}\theta_{a}(E_{-\gamma_{a}}-E_{\gamma_{a}})\right)\ket{\Lambda},
\qquad
0\le\theta_{a}\le\frac{\pi}{2},
\label{eq:hc_flat_state}
\end{equation}
and its fidelity with the target factorizes as
\begin{equation}
\braket{\Lambda}{\psi(\theta)}
=\prod_{a=1}^{r}\cos^{\ell}\theta_{a},
\qquad
|\braket{\Lambda}{\psi(\theta)}|^{2}
=\prod_{a=1}^{r}\cos^{2\ell}\theta_{a}.
\label{eq:hc_fidelity_product}
\end{equation}
\end{lemma}

\begin{lemma}
\label{lem:flat_fisher_block}
On the flat \eqref{eq:hc_flat_state}, the differential of
\begin{equation}
\cL_{\mathrm{fid}}(\psi)=1-|\braket{\Lambda}{\psi}|^{2}
\end{equation}
has no component along directions orthogonal to the flat.
In the principal-angle coordinates $\theta_{1},\dots,\theta_{r}$, the pure-state QFIM has block form
\begin{equation}
F = 4\ell I_{r}\oplus F_{\perp},
\label{eq:hc_fisher_block}
\end{equation}
with vanishing flat--offflat cross block.
Consequently the quotient-QNG vector field of $\cL_{\mathrm{fid}}$ is tangent to the flat.
\end{lemma}

\begin{theorem}
\label{thm:fidelity_integrability}
Let $\mathcal O_{\lambda}$ be as in \cref{lem:hc_frame}, and let the initial ray have positive fidelity with the target, so that $\cos^{2}\theta_{a}(0)>0$ for every $a$.
Set
\begin{equation}
m_{a}:=\sin^{2}\theta_{a},
\qquad
p:=|\braket{\Lambda}{\psi}|^{2}=\prod_{a=1}^{r}(1-m_{a})^{\ell},
\end{equation}
and let $A=\{a:m_{a}(0)>0\}$ be the active defect set.
For the continuous-time quotient-QNG flow
\begin{equation}
\dot\gamma(t)=-\eta\,\grad_{\bar F}\cL_{\mathrm{fid}}(\gamma(t)),
\qquad \eta>0,
\end{equation}
the principal defects obey
\begin{equation}
\dot m_{a}=-\eta\,p\,m_{a},
\qquad a=1,\dots,r.
\label{eq:cominuscule_defect_flow}
\end{equation}
Hence all ratios $m_{a}/m_{b}$ with $a,b\in A$ are first integrals.
If $A$ is nonempty, then
\begin{equation}
m(t)=u(t)\,m(0),
\qquad u(0)=1,
\end{equation}
where $u$ satisfies the scalar equation
\begin{equation}
\dot u=-\eta\,u\prod_{a=1}^{r}(1-\beta_{a}u)^{\ell},
\qquad \beta_{a}:=m_{a}(0).
\label{eq:cominuscule_scalar_reduction}
\end{equation}
The loss is
\begin{equation}
\cL_{\mathrm{fid}}(t)=1-\prod_{a=1}^{r}(1-\beta_{a}u(t))^{\ell}.
\end{equation}
Moreover $u(t)=Ce^{-\eta t}(1+o(1))$, and therefore
\begin{equation}
\cL_{\mathrm{fid}}(t)
=\ell\left(\sum_{a=1}^{r}\beta_{a}\right)Ce^{-\eta t}+o(e^{-\eta t}).
\end{equation}
Thus the asymptotic exponent is independent of rank, ambient dimension, active defect distribution, and embedding index; the leading loss amplitude retains the factor $\ell$.
If some $\cos\theta_{a}=0$, then $p=0$ and the vector field vanishes; this excluded set is the zero-fidelity critical locus.
\end{theorem}

The transfer criterion determines when a circuit follows this full-orbit vector field.
At a point $\theta$, exact realization occurs precisely when \cref{eq:objective_transfer_condition} holds with $b_\rho$ representing the fidelity differential.
Surjectivity of $B_\theta$ guarantees this condition for every orbit objective.
If \cref{eq:objective_transfer_condition} fails, \cref{eq:induced_orbit_update} gives the metric projection onto the circuit-reachable tangent subspace; the projected vector field need not preserve the defect ratios or the scalar reduction.

\begin{corollary}
\label{cor:cominuscule_discrete}
For the finite-step quotient-QNG update implemented by the Lie-exponential retraction, the principal angles close exactly as
\begin{equation}
\theta_{a}^{+}=
\left|\theta_{a}-\frac{\eta}{2}\sin\theta_a\,
\cos^{2\ell-1}\theta_a
\prod_{b\ne a}\cos^{2\ell}\theta_b\right|.
\label{eq:cominuscule_discrete_map}
\end{equation}
This exactness is relative to the Lie-exponential retraction; a first-order retraction gives the same map up to $O(\|\Delta\theta\|^{2})$. If $0<\eta\le4$ and the initial fidelity is positive, the iteration converges globally to the target; for $0<\eta\le2$ it does so without overshoot in any principal angle.
\end{corollary}

For $0<\eta<4$ with $\eta\ne2$, the local convergence ratio in $\|\theta\|$ is $|1-\eta/2|$. At $\eta=2$ the linear term vanishes and
\begin{equation}
\theta_{a}^{+}
=\left(\ell-\frac{1}{3}\right)\theta_{a}^{3}
+\ell\theta_{a}\sum_{b\ne a}\theta_{b}^{2}
+O(\|\theta\|^{5}),
\label{eq:cominuscule_newton_cubic}
\end{equation}
so the convergence is cubic in the principal-angle norm.

The first integrals in \cref{thm:fidelity_integrability} belong to the continuous flow and are generally not preserved by the finite map.  Indeed, before any overshoot and with $p$ evaluated at the current iterate, a small-$\eta$ expansion of \cref{eq:cominuscule_discrete_map} gives
\begin{equation}
\log\frac{m_a^+}{m_b^+}-\log\frac{m_a}{m_b}
=-\frac{\eta^2p^2}{4}
\left(\sec^2\theta_a-\sec^2\theta_b\right)+O(\eta^3).
\label{eq:discrete_defect_ratio_drift}
\end{equation}
Thus complete integrability is a continuous-time property: one-step ratio drift is second order in the step size and accumulates to first order over a fixed continuous-flow time.

For rank one at $\eta=2$, the signed local distance map is
\begin{equation}
\Phi_\ell(d)=d-\sin d\cos^{2\ell-1}d
=\left(\ell-\frac13\right)d^3
-\left(\frac{(2\ell-1)^2}{8}+\frac1{120}\right)d^5+O(d^7).
\label{eq:embedding_rank_one_map}
\end{equation}

At the stability boundary $\eta=4$, the absolute rank-one distance instead obeys
\begin{equation}
d_{t+1}=d_t-2\left(\ell-\frac13\right)d_t^3+O(d_t^5),
\qquad
d_t\sim\frac{1}{\sqrt{4(\ell-1/3)t}},
\label{eq:embedding_rank_one_boundary}
\end{equation}
so convergence is sublinear rather than linear.

A physical family with $\ell>1$ is supplied by $N$ identical qubits restricted to the permutation-symmetric sector $\operatorname{Sym}^N(\bbC^2)$.  The collective $SU(2)$ action has highest weight $\Lambda=N\omega_1$, so $\ell=N$, and the orbit of $\ket{0}^{\otimes N}$ is the spin-coherent copy of $\bbC\mathrm P^1$ under the degree-$N$ Veronese embedding \cite{Perelomov1986}.  Along
\begin{equation}
\ket{\psi(\theta)}=
(\cos\theta\ket0+\sin\theta\ket1)^{\otimes N},
\end{equation}
the target fidelity is $\cos^{2N}\theta$, the flat QFIM coefficient is $4N$, and the $\eta=2$ local map has cubic coefficient $N-1/3$.  Hence the global learning-rate window remains $0<\eta\le4$, while the angular region in which the cubic term is small contracts on the scale $N^{-1/2}$.

\begin{theorem}
\label{thm:inexact_cubic}
Suppose an implemented rank-one step at $\eta=2$ satisfies
\begin{equation}
|d_{t+1}-|\Phi_\ell(d_t)||\le\delta_t.
\end{equation}
For sufficiently small $d_t$ there are constants $C_{5,\ell},C_{7,\ell}>0$ such that
\begin{equation}
d_{t+1}\le\left(\ell-\frac13\right)d_t^3
+C_{5,\ell}d_t^5+C_{7,\ell}d_t^7+\delta_t.
\label{eq:inexact_cubic_bound}
\end{equation}
Consequently cubic convergence persists while $\delta_t=o(d_t^3)$.
With a fixed distance-error floor $\delta$, the cubic regime ends at $d_t=O(\delta^{1/3})$; maintaining it asymptotically requires a growing shot budget that makes $\delta_t/d_t^3\to0$.
\end{theorem}
\begin{proof}
Apply Taylor's theorem to \cref{eq:embedding_rank_one_map} and add the implementation error by the triangle inequality.
\end{proof}

For the principal defect channels, the integrable flow has a direct selection rule.
For Slater targets in number-conserving free-fermion circuits and for fermionic-Gaussian targets in matchgate circuits, quotient QNG rescales all active principal defect channels by the same scalar factor.
Equivalently, the curves $\log m_{a}(t)$ differ only by constants as long as the deterministic cominuscule model applies.
This gives a falsifiable trajectory signature and an algorithmic limitation: quotient QNG contracts the existing defect profile but does not redistribute error between principal channels.

The integrability mechanism differs from the classical double-bracket picture.
Completely integrable gradient flows on orbits are well known for moment-map expectation objectives: double-bracket flows sort spectra, diagonalize matrices, and realize Toda-type dynamics as gradient flows on adjoint orbits \cite{Brockett1991,BlochBrockettRatiu1992}.
The contribution here is the identification of conserved principal-defect ratios, a scalar reduction, and an exact Lie-retracted map for the QNG fidelity vector field on embedded quantum-state orbits.
The underlying coherent-state, principal-angle, and symmetric-space geometry is classical \cite{Perelomov1986,Berceanu1997GeometryCoherentStates,Berceanu2004GeometricalPhases}.
The resulting map also differs from cubically convergent Grassmann invariant-subspace iterations, which use another objective and vector field \cite{Absil2004CubicGrassmann}, and from geodesic corrections to QNG \cite{Halla2025QNGGeodesic}.
The proofs of \cref{lem:hc_frame,lem:flat_fisher_block,thm:fidelity_integrability,cor:cominuscule_discrete}, together with the local expansion underlying \cref{eq:cominuscule_newton_cubic}, are given in Appendix~\ref{app:cominuscule_fidelity}.

\subsection{Physical conditioning limits and Morse--Bott landscapes}
\label{subsec:qng_limit}

Cominuscule fidelity is exactly solvable because its maximal-flat defects share one contraction factor.
A general energy objective in the full-control sector has a different controlling mechanism: physical excitation gaps set the intrinsic Hessian anisotropy.
Exact redundancy is removed by quotienting or pseudoinversion. At a critical point where the quotient-to-orbit map is locally diffeomorphic, quotient QNG removes conditioning caused solely by the faithful circuit chart, while the gap-dependent anisotropy remains in the intrinsic objective operator.

\begin{theorem}
\label{thm:qng_limit}
Let $H_{\mathrm{obj}}$ be a Hermitian operator on a fixed sector $\cK_{\lambda}$ of dimension $D$, and consider the objective
\begin{equation}
\cL_{H}(\psi):=\bra{\psi}H_{\mathrm{obj}}\ket{\psi}
\end{equation}
restricted to rays in $\cK_{\lambda}$.
Assume that the restricted dynamical Lie algebra on $\cK_{\lambda}$ satisfies $\mathfrak g_{\lambda}=\mathfrak{su}(\cK_{\lambda})$, and let $\ket{\psi_{\star}}=\ket{0}$ be a nondegenerate ground state of $H_{\mathrm{obj}}\vert_{\cK_{\lambda}}$ with eigenbasis $\{\ket{k}\}_{k=0}^{D-1}$ and eigenvalues $E_{0}<E_{1}\le \cdots \le E_{D-1}$.
Let $\mathbf{H}_{\star}$ denote the Euclidean Hessian matrix of $\cL_{H}$ in a local generator chart centered at a representative of the minimizing ray $\bbC\ket{\psi_{\star}}$.
Then the following hold.
\begin{enumerate}
\item The full generator-coordinate Hessian has $(D-1)^{2}$ zero eigenvalues coming from stabilizer directions. On the horizontal tangent space, its nonzero spectrum is
\begin{equation}
\{E_{1}-E_{0},E_{1}-E_{0},\dots,E_{D-1}-E_{0},E_{D-1}-E_{0}\}.
\end{equation}
In particular,
\begin{equation}
\kappa_{\mathrm{Eucl}}=\frac{E_{D-1}-E_{0}}{E_{1}-E_{0}}.
\end{equation}
\item The Jacobian of the linearized continuous-time quotient-QNG flow on horizontal directions is $-\mathbf{M}_{\star}$, where
\begin{equation}
\mathbf{M}_{\star}:=F(\psi_{\star})^{+}\mathbf{H}_{\star}.
\end{equation}
The nonzero eigenvalues of $\mathbf{M}_{\star}$ are
\begin{equation}
\frac{1}{2}(E_{1}-E_{0}),\frac{1}{2}(E_{1}-E_{0}),\dots,\frac{1}{2}(E_{D-1}-E_{0}),\frac{1}{2}(E_{D-1}-E_{0}).
\end{equation}
Hence
\begin{equation}
\kappa_{\mathrm{QNG}}=\kappa_{\mathrm{Eucl}}=\frac{E_{D-1}-E_{0}}{E_{1}-E_{0}}.
\end{equation}
\end{enumerate}
\end{theorem}

QNG still eliminates geometric ill-conditioning from redundant parameterizations, as established in Section~\ref{sec:quotient}, but it cannot remove the physical anisotropy encoded in the excitation gaps $E_{k}-E_{0}$.
The local condition number $\kappa_{\mathrm{QNG}}=\kappa_{\mathrm{Eucl}}$ is the precise expression of this boundary.

The same excitation structure determines the global critical geometry of linear expectation objectives: eigenspaces form the critical manifolds, and energy differences give the transverse Hessian spectrum.

\begin{theorem}
\label{thm:morse_bott_landscapes}
\begin{enumerate}
\item[(a)]
In the full projective sector, let $\mathfrak g_{\lambda}$ be the restricted dynamical Lie algebra on $\cK_{\lambda}$ and assume \cref{eq:sector_core_assumptions}.
Fix any reference ray $[\psi_{0}]\in \mathbb{P}(\cK_{\lambda})$ and let $P_{0}=\ket{\psi_{0}}\!\bra{\psi_{0}}$.
Then the reachable pure-state sector orbit is
\begin{equation}
\mathcal{O}_{\lambda}:=G_{\lambda}\cdot [\psi_{0}]=\mathbb{P}(\cK_{\lambda})
\cong SU(D)/S(U(1)\times U(D-1))\cong \bbC\mathrm{P}^{D-1},
\end{equation}
equivalently the coadjoint orbit of the rank-one projector $P_{0}$.
For a Hermitian operator $H_{\mathrm{obj}}$ on $\cK_{\lambda}$, define the linear expectation objective
\begin{equation}
\cL_{H}([\psi])=\bra{\psi}H_{\mathrm{obj}}\ket{\psi}=\Tr(H_{\mathrm{obj}}P),
\qquad P=\ket{\psi}\!\bra{\psi},
\end{equation}
on rays $[\psi]\in \mathcal{O}_{\lambda}$.
Let
\begin{equation}
\cK_{\lambda}=\mathcal E_{0}\oplus \mathcal E_{1}\oplus \cdots \oplus \mathcal E_{J}
\end{equation}
be the orthogonal decomposition into eigenspaces of $H_{\mathrm{obj}}\vert_{\cK_{\lambda}}$ corresponding to distinct eigenvalues
\begin{equation}
E_{0}<E_{1}<\cdots<E_{J}.
\end{equation}
Then the following hold.
\begin{enumerate}
\item The critical set of $\cL_{H}$ is exactly the disjoint union
\begin{equation}
\mathrm{Crit}(\cL_H)=\bigsqcup_{j=0}^{J} \mathbb{P}(\mathcal E_{j}).
\end{equation}
\item At each point $[\psi]\in \mathbb{P}(\mathcal E_{j})$, the Hessian kernel of $\cL_H$ is precisely the tangent space
\begin{equation}
\ker(\mathrm{Hess}_{[\psi]}\cL_H)=T_{[\psi]}\mathbb{P}(\mathcal E_{j}).
\end{equation}
Hence $\cL_H$ is Morse--Bott on $\mathcal{O}_{\lambda}$.
\item Writing the projective tangent space as
\begin{equation}
T_{[\psi]}\bbC\mathrm{P}^{D-1}
\cong (\mathcal E_{j}\cap \psi^{\perp}) \oplus \bigoplus_{k\neq j} \mathcal E_{k}
\end{equation}
as a complex orthogonal sum, the Hessian vanishes on $\mathcal E_{j}\cap \psi^{\perp}$ and on each complex summand $\mathcal E_{k}$ with $k\neq j$ it acts, in real coordinates, by the scalar
\begin{equation}
2(E_{k}-E_{j})
\end{equation}
with multiplicity $2\dim \mathcal E_{k}$.
\item Consequently, $\mathbb{P}(\mathcal E_{0})$ is the global minimum manifold, $\mathbb{P}(\mathcal E_{J})$ is the global maximum manifold, and every intermediate critical manifold $\mathbb{P}(\mathcal E_{j})$ with $0<j<J$ is a saddle manifold.
\end{enumerate}
In particular, $\cL_H$ has no spurious local minima or local maxima on the quotient sector orbit.
\item[(b)]
In the highest-weight flag-orbit setting of \cref{thm:flag_root_spectrum}, if
\begin{equation}
\mathcal O_{\lambda}=G_{\lambda}\cdot[\Lambda]
\end{equation}
and the traceless part of $H_{\mathrm{obj}}$ belongs to $\mathfrak g_{\lambda}$, then the expectation objective is a Hamiltonian moment-map component on $\mathcal O_{\lambda}$.
It is a perfect Morse--Bott function; every critical manifold has even Morse index, every local minimum or maximum is global, and there are no spurious local extrema on the flag orbit.
\end{enumerate}
\end{theorem}

The factor~$2$ in the real-coordinate Hessian is the standard affine-chart convention on projective space; in the trace-normalized chart of \cref{thm:qng_limit} the coordinate rescaling $z=t/\sqrt{2}$ absorbs it, giving eigenvalues $E_k-E_j$.
The landscape guarantee applies to linear expectations on the reachable orbit; nonlinear losses or observables outside $\mathfrak g_{\lambda}$ can have landscape features not covered by the moment-map theorem.

Wiersema and Killoran derived the corresponding Riemannian circuit flow, whose rank-one projector form is the classical double-bracket evolution \cite{Wiersema2023Riemannian,Brockett1991}.  In the QFIM normalization used here, the cited benchmark is
\begin{equation}
\dot P(t)=-\frac{\eta}{2}\bigl[[H_{\mathrm{obj}},P(t)],P(t)\bigr].
\label{eq:double_bracket_qng}
\end{equation}
This known full-control benchmark sorts the objective spectrum, whereas the partial-control fidelity flow in \cref{thm:fidelity_integrability} preserves principal-defect ratios.

\subsection{State-preparation dynamics}
\label{subsec:dynamics}

Full-control state preparation is the rank-one intersection of the projective landscape and the cominuscule fidelity flow.
General sector expectation flows converge to the ground-state manifold from almost every initial ray, while the rank-one loss further admits closed continuous and discrete dynamics.

\begin{lemma}
\label{lem:sector_flow_convergence}
Assume the full-control setting of part~(a) of \cref{thm:morse_bott_landscapes} and consider the continuous-time quotient gradient flow on the sector orbit,
\begin{equation}
\dot\gamma(t)=-\eta\,\grad_{\bar F}\cL_H(\gamma(t)),
\qquad \gamma(0)\in \mathcal O_{\lambda},\quad \eta>0,
\end{equation}
where $\bar F$ is the descended QFIM metric.
Then every trajectory converges to a critical manifold of $\cL_H$.
Moreover, with respect to the Fubini--Study volume measure on $\mathcal O_{\lambda}\cong\bbC\mathrm{P}^{D-1}$, the set of initial rays whose trajectories converge to a saddle manifold $\mathbb P(\mathcal E_j)$ with $0<j<J$ or to the maximum manifold $\mathbb P(\mathcal E_J)$ has measure zero.
Hence the quotient gradient flow converges to the ground-state manifold $\mathbb P(\mathcal E_0)$ for almost every initialization.
\end{lemma}

This almost-everywhere convergence combines analytic gradient-flow convergence on the compact orbit with the center-stable manifold theorem for the nonminimal critical manifolds \cite{AbsilMahonyAndrews2005,HirschPughShub1977,Shub1987}.

\begin{corollary}
\label{cor:full_control_stateprep}
Assume that $\ket{\psi_{\mathrm{tar}}}\in \cK_{\lambda}$ and that the restricted dynamical Lie algebra on $\cK_{\lambda}$ satisfies the full-control condition in \cref{eq:sector_core_assumptions}.
For the standard state-preparation loss
\begin{equation}
\cL_{\mathrm{sp}}(\psi):=1-\left|\braket{\psi_{\mathrm{tar}}}{\psi}\right|^{2},
\end{equation}
\begin{enumerate}
\item[(a)]
The landscape on $\bbC\mathrm{P}^{D-1}$ has the target ray $\bbC\ket{\psi_{\mathrm{tar}}}$ as its unique minimum and the orthogonal rays as its global maxima.
This is the rank-one special case of part~(a) of \cref{thm:morse_bott_landscapes} with $H_{\mathrm{obj}}=I-P_{\mathrm{tar}}$.
In generator coordinates centered at a minimizing representative $\ket{\psi_{\star}}$, the Euclidean Hessian satisfies
\begin{equation}
\mathbf{H}_{\star}=\frac{1}{2}F(\psi_{\star}).
\label{eq:hessian_half_qfim}
\end{equation}
Consequently, by \cref{eq:full_control_spectrum}, the full generator-coordinate Hessian has $(D-1)^{2}$ zero eigenvalues and $2D-2$ nonzero eigenvalues, all equal to $1$.
In particular, both the Euclidean and quotient-QNG local condition numbers are equal to $1$.
\item[(b)]
The continuous quotient flow is the $r=1$ specialization of \cref{thm:fidelity_integrability}.
For the continuous-time quotient-QNG flow
\begin{equation}
\dot\gamma(t)=-\eta\,\grad_{\bar F}\cL_{\mathrm{sp}}(\gamma(t)),
\qquad \eta>0,
\end{equation}
the loss satisfies the closed scalar equation
\begin{equation}
\dot\cL_{\mathrm{sp}}(t)=-\eta\,\cL_{\mathrm{sp}}(t)\bigl(1-\cL_{\mathrm{sp}}(t)\bigr),
\label{eq:stateprep_logistic_ode}
\end{equation}
and therefore
\begin{equation}
\cL_{\mathrm{sp}}(t)=\frac{\cL_{\mathrm{sp}}(0)e^{-\eta t}}{1-\cL_{\mathrm{sp}}(0)+\cL_{\mathrm{sp}}(0)e^{-\eta t}}.
\label{eq:stateprep_logistic_solution}
\end{equation}
In particular, the continuous-time quotient-QNG flow converges exponentially fast to the target ray for every initial ray not orthogonal to $\ket{\psi_{\mathrm{tar}}}$.
At the discrete level, the update $\Delta\theta=-\eta F(\theta)^{+}\nabla\cL_{\mathrm{sp}}(\theta)$ with $\eta=2$ coincides with Newton's method for the horizontal quadratic model at the optimum.
\item[(c)]
The discrete Lie-retracted quotient update is the $r=1$ specialization of \cref{cor:cominuscule_discrete}.
For the discrete quotient-QNG iteration with the Lie-group retraction,
\begin{equation}
\ket{\psi_{t+1}}
=\exp\!\Big({+i\frac{\eta}{2}\sum_{a=1}^{D^{2}-1}g^{(t)}_{a}T_{a}}\Big)\ket{\psi_{t}},
\qquad
g^{(t)}_{a}=\partial_{a}\cL_{\mathrm{sp}}(\psi_{t}),
\label{eq:discrete_retracted_update}
\end{equation}
let $d_{t}:=\arccos\left|\braket{\psi_{\mathrm{tar}}}{\psi_{t}}\right|\in[0,\pi/2]$ denote the
Fubini--Study distance from the target ray, so that $\cL_{\mathrm{sp}}(\psi_{t})=\sin^{2}d_{t}$.
The iteration closes exactly at the level of the distance:
\begin{equation}
d_{t+1}=\Bigl|\,d_{t}-\frac{\eta}{4}\sin(2d_{t})\,\Bigr|.
\label{eq:discrete_distance_recursion}
\end{equation}
Consequently, for every initial ray with $d_{0}<\pi/2$:
\begin{enumerate}
\item for $0<\eta\le 4$, the sequence $d_{t}$ is strictly decreasing and converges to $0$;
\item for $0<\eta\le 2$, the argument of the absolute value in
\cref{eq:discrete_distance_recursion} is nonnegative, so the iterates approach the target
monotonically along the connecting geodesic without overshooting;
\item for $0<\eta<4$ with $\eta\neq 2$, the local convergence is linear with exact asymptotic
ratio $|1-\eta/2|$, i.e.\ $d_{t+1}=|1-\eta/2|\,d_{t}+O(d_{t}^{3})$;
\item at the Newton step size $\eta=2$, the convergence is cubic:
$d_{t+1}=\tfrac{2}{3}d_{t}^{3}+O(d_{t}^{5})$, equivalently
$\cL_{\mathrm{sp}}(\psi_{t+1})=\tfrac{4}{9}\cL_{\mathrm{sp}}(\psi_{t})^{3}\bigl(1+O(\cL_{\mathrm{sp}}(\psi_{t}))\bigr)$;
\item at $\eta=4$, convergence is sublinear:
$d_{t+1}=d_t-\tfrac43d_t^3+O(d_t^5)$ and
$d_t\sim\sqrt{3/(8t)}$.
\end{enumerate}
For $\eta>4$ the target ray is locally repelling, so $\eta=4$ is the exact discrete stability
boundary.
\end{enumerate}
\end{corollary}

The discrete recursion \cref{eq:discrete_distance_recursion} is the rank-one case of the flat map \eqref{eq:cominuscule_discrete_map}; expanding to first order in $\eta$ recovers the continuous logistic law \cref{eq:stateprep_logistic_ode}.
The linear-stability interval $\eta<4$ matches the classical quadratic-model bound: \cref{cor:free_ng} makes the effective Euclidean step $\eta/2$, and the horizontal Hessian at the optimum is the identity by \cref{eq:hessian_half_qfim}, so strict linear contraction requires $\eta/2<2$.  The boundary value $\eta=4$ remains convergent through the cubic correction but loses a linear rate.
The cubic rather than quadratic rate at $\eta=2$ reflects the vanishing of the third-order term along the connecting geodesic.
From $\cL_{\mathrm{sp}}(0)=0.6$ the Newton-step iteration reaches $1.6\times10^{-3}$ after two steps and $1.9\times10^{-9}$ after three.

In the trace-orthonormal orbit frame of \cref{eq:full_control_spectrum}, the QNG vector field satisfies
\begin{equation}
-\eta F(\psi)^{+}\nabla\cL_{\mathrm{sp}}
=-\frac{\eta}{2}\nabla\cL_{\mathrm{sp}},
\end{equation}
so quotient QNG is exactly the Euclidean horizontal gradient field with learning-rate parameter $\eta/2$ in that frame.
For a general gate chart, the corresponding statement is instead the transfer law \cref{eq:isotropic_spectral_transfer}.

\section{Finite-shot quotient stability and numerical validation}
\label{sec:noise_trainability}

Finite measurements turn a small Fisher eigenvalue into two coupled costs: pseudoinverse amplification and uncertainty in the subspace on which the update is formed.
The decisive information is structural.
A known exact gauge can be projected out before estimation, whereas a statistically unresolved low-curvature direction may remain physical and requires soft regularization.
The Fisher spectral gap controls both costs and links pointwise kernel certification to iterative stability.
Proofs of the finite-shot bounds are collected in Appendix~\ref{app:finite_shot_proofs}.

\subsection{Fisher-gap amplification and kernel certification}
\label{sec:kernel-certification}

A small horizontal Fisher eigenvalue amplifies gradient-estimation noise through the pseudoinverse, while full-space damping converts vertical estimator noise into representative drift. Let $\widehat{g}$ be an unbiased estimator of $\nabla \cL$ with covariance matrix $\Sigma$ and consider the QNG update $\Delta\theta=-\eta F^{+}\widehat{g}$.

\begin{theorem}
\label{thm:finite_shot_amplification}
\begin{enumerate}
\item[(a)]
Assume the objective descends through the quotient, so that $P_{\cV}\nabla\cL(\theta)=0$, where $P_{\cV}$ is the Euclidean orthogonal projector onto the vertical gauge space
\begin{equation}
\cV_{\theta}\subseteq\ker F(\theta).
\end{equation}
Let $\widehat g=\nabla\cL(\theta)+\epsilon$ be an unbiased finite-shot gradient estimator with covariance
\begin{equation}
\mathbb E[\epsilon\epsilon^{T}]=\Sigma.
\end{equation}
If the update is formed with isotropic Tikhonov damping on the full ambient space,
\begin{equation}
\Delta\theta_{\gamma}=-\eta\,(F(\theta)+\gamma I)^{-1}\widehat g,
\qquad \gamma>0,
\end{equation}
then the expected squared vertical representative drift is exactly
\begin{equation}
\mathbb E\left\|P_{\cV}\Delta\theta_{\gamma}\right\|_{2}^{2}
=\frac{\eta^{2}}{\gamma^{2}}\Tr(P_{\cV}\Sigma P_{\cV}).
\label{eq:tikhonov_vertical_drift}
\end{equation}
\item[(b)]
For the undamped pseudoinverse update $\Delta\theta=-\eta F^+\widehat g$, if $F$ has rank $d_F$ and smallest nonzero eigenvalue $f_{\min}^{+}(F)$, then
\begin{equation}
\mathbb{E}\,\Vert \Delta\theta-\mathbb{E}\Delta\theta\Vert_{2}^{2}
= \eta^{2}\Tr(F^{+}\Sigma F^{+})
\le \eta^{2}\Vert \Sigma\Vert_{2}\,\Tr\!\bigl((F^{+})^{2}\bigr)
\le \eta^{2}\Vert \Sigma\Vert_{2}\,\frac{d_F}{\bigl(f_{\min}^{+}(F)\bigr)^{2}}.
\end{equation}
\end{enumerate}
\end{theorem}

Projection or pseudoinversion should remove known vertical directions before numerical stabilization is applied to the horizontal block.
Full-space damping instead amplifies vertical noise at the exact rate in \cref{eq:tikhonov_vertical_drift} along directions that carry no state information.
When $\Vert\Sigma\Vert_{2}$ scales as $1/M$ with the shot count, the bound predicts $M\propto d_F/(f_{\min}^{+})^{2}$ to keep update noise below a fixed threshold.

For a protocol-independent statement, let $r_F(N,\delta)$ be any valid operator-norm confidence radius,
\begin{equation}
\Pr\!\left[\|\widehat F-F\|_{\mathrm{op}}\le r_F(N,\delta)\right]\ge1-\delta.
\label{eq:operator_confidence_interface}
\end{equation}
This interface covers entrywise estimates, matrix-concentration protocols, stochastic approximations, and structured estimators.
Bias and calibration uncertainty may be included whenever the stated radius remains valid.

\begin{lemma}
\label{lem:psd_preprocessing}
Let $F\succeq0$, let $\|\widehat F-F\|\le \rho_F$, and let $\widehat F_+$ be obtained by symmetrizing $\widehat F$ and clipping its negative eigenvalues to zero.
Then
\begin{equation}
\|\widehat F_+-F\|\le2\rho_F.
\label{eq:psd_clip_bound}
\end{equation}
\end{lemma}
After this preprocessing, the effective confidence radius is $2r_F$ unless a sharper projection bound is available.

When a state-preserving action is known analytically, faithfulness can be tested relative to its vertical space without interpreting unexplained small eigenvalues as gauge.

\begin{theorem}
\label{thm:finite_shot_faithfulness}
Let $Q$ be the orthogonal projector onto a known state-preserving vertical space, let $P=I-Q$, and suppose $FQ=QF=0$.
On an event $\|\widehat F-F\|\le r_F$, define
\begin{equation}
\widehat f_{\min}:=\lambda_{\min}\!\left(P\widehat F P\big|_{\operatorname{im}P}\right).
\end{equation}
If $\widehat f_{\min}>r_F$, then
\begin{equation}
F\big|_{\operatorname{im}P}\succeq(\widehat f_{\min}-r_F)I,
\qquad
\ker F=\operatorname{im}Q.
\label{eq:faithfulness_certificate}
\end{equation}
Thus the prescribed gauge is faithful at the tested point with certified horizontal gap $f_{\mathrm{cert}}=\widehat f_{\min}-r_F$.
If
$\widehat f_{\max}=\lambda_{\max}(P\widehat F P|_{\operatorname{im}P})$, then
$F|_{\operatorname{im}P}\preceq(\widehat f_{\max}+r_F)I$.
\end{theorem}
An adaptive implementation may increase shots geometrically and stop at the first stage $j$ for which $\widehat f_{\min,j}>r_j$.
Choosing $\delta_j=6\delta/(\pi^2j^2)$ makes all stagewise confidence events jointly valid with failure probability at most $\delta$.
Failure to trigger is reported as an unresolved horizontal gap, not as evidence for additional exact gauge directions.

A known vertical space supports the direct faithfulness test in \cref{thm:finite_shot_faithfulness}.
Without such a space, a separated Fisher kernel can still be recovered from a finite estimate when the true rank and a nonzero Fisher gap are supplied independently.
If, in addition, $\ker F(\theta)=\cV_\theta$, the recovered kernel is the prescribed vertical space; identifying a global group action still requires independent structural information.

\begin{theorem}
\label{thm:certified-truncation}
Let $F\in\bbR^{p\times p}$ be symmetric positive semidefinite with rank $d_F$ and smallest nonzero eigenvalue $f^+_{\min}$, and let $\widehat F$ be symmetric with
\begin{equation}
\|\widehat F-F\|_2\le\varepsilon,
\qquad
12\varepsilon\le f^+_{\min}.
\end{equation}
Fix any threshold $\tau\in[2\varepsilon,f^+_{\min}-2\varepsilon]$ and let $\widehat F_\tau^{+}$ denote the pseudoinverse of $\widehat F$ truncated below $\tau$.
Then:
\begin{enumerate}
\item exactly $d_F$ eigenvalues of $\widehat F$ exceed $\tau$, so the truncation preserves the true rank $d_F$;
\item $\|\widehat P_{0}-P_{0}\|_2\le 2\varepsilon/f^+_{\min}$, where $\widehat P_{0}$ and $P_{0}$ are the kernel projectors of the truncation and of $F$;
\item $\|\widehat F_\tau^{+}-F^{+}\|_2\le 6\varepsilon/(f^+_{\min})^2$.
\end{enumerate}
\end{theorem}

For bounded unbiased entrywise Fisher estimators, the same perturbation argument gives a sufficient shot count scaling as $p^2/(f^+_{\min})^2$ up to the estimator range and a logarithmic confidence factor; \cref{cor:shot-complexity} records the explicit constant and the corresponding vertical-leakage bound. Together with \cref{thm:finite_shot_amplification}, this shows that the Fisher gap controls both update amplification and subspace recovery.
Without independently specified state-preserving structure, recovering the local Fisher kernel does not identify it as a prescribed gauge action, as made explicit by \cref{prop:soft_indistinguishability}.
Matrix concentration can sharpen the entrywise union bound \cite{Tropp2012}; the stated elementary form exposes all constants.

\subsection{Horizontal implementation and directional gauge tests}

Once an exact vertical space is known, projection converts the singular ambient problem into a positive-definite horizontal solve.
Let the columns of $V_\theta$ span that vertical space and define
\begin{equation}
P_{\cV}=V_\theta(V_\theta^{\mathsf T}V_\theta)^+V_\theta^{\mathsf T},
\qquad
P_{\mathsf H}=I-P_{\cV}.
\end{equation}
For finite-shot estimates $\widehat F$ and $\widehat g$, the corresponding stabilized update is
\begin{equation}
\widehat\Delta_\gamma
=-\eta P_{\mathsf H}
\left(P_{\mathsf H}\widehat F P_{\mathsf H}+\gamma P_{\mathsf H}\right)^+
P_{\mathsf H}\widehat g.
\label{eq:stable_horizontal_update}
\end{equation}
Regularization therefore acts only on the identifiable block.
The Fisher estimator is otherwise unrestricted, and an iterative implementation needs projector applications and horizontal Fisher--vector products rather than an explicit basis for the full dynamical Lie algebra.

The horizontal system can be solved matrix-free by conjugate gradients using projector applications and horizontal Fisher--vector products. The standard condition-number dependence and classical-oracle cost are recorded in \cref{prop:matrix_free_horizontal}; the quantum measurement cost of a Fisher--vector product remains estimator dependent \cite{vanStraaten2021Measurement,Gacon2021SPSA,Kolotouros2024RandomNG,Halla2025SteinQFI,Rath2021RandomizedQFI}.

For a Slater state represented by $Q\in\operatorname{St}(k,n)$, the Grassmann-horizontal projection of a tangent $Z$ is
\begin{equation}
P_{\mathsf H,Q}(Z)=(I-QQ^\dagger)Z,
\end{equation}
with vertical component $Q(Q^\dagger Z)$.
This projector costs $O(nk^2)$ arithmetic operations and never constructs the $\binom nk$-dimensional many-body vector.
Combined with one-body orbital contractions, it gives a polynomial-orbital-dimension horizontal solve for this model class.

A proposed gauge generator can also be tested locally without reconstructing a global group action.

\begin{proposition}
\label{prop:directional_certificate}
Let $v\in T_\theta\Theta$ and $A_v=C_\theta v$. For a pure state,
\begin{equation}
F_\theta(v,v)=4\Var_\psi(A_v),
\qquad
F_\theta(v,v)=0\iff A_v\ket\psi=a\ket\psi
\end{equation}
for some $a\in\bbR$. For a fixed-rank mixed state on a unitary orbit,
\begin{equation}
F_\theta(v,v)=0\iff[A_v,\rho]=0.
\end{equation}
\end{proposition}
A symmetric finite-difference fidelity test estimates this local null criterion without reconstructing a global group action; its $O(s^2)$ bias and finite-shot confidence radius are given in \cref{prop:directional_confidence_bound}. This establishes a local state-null test rather than a global redundancy action. Quantum symmetry tests and variational symmetry-learning methods can supply candidate actions upstream \cite{LaBorde2023TestingSymmetry,Lu2024LearningSymmetries}.

\subsection{Known and estimated horizontal projectors}

Replacing an exact horizontal projector by an estimate introduces a new error channel: leakage into the true gauge space. The exact-projector case isolates Fisher and gradient errors; the estimated-projector case prices this additional leakage separately.

\begin{theorem}
\label{thm:known_projector_update}
Assume $F=PFP$, $g=Pg$, and $F|_{\operatorname{im}P}\succeq f_{\min}I$ with $f_{\min}>0$.
Let
\begin{equation}
\|\widehat F-F\|\le\varepsilon_F,
\qquad
\|\widehat g-g\|\le\varepsilon_g,
\qquad a=f_{\min}+\gamma,
\end{equation}
and define
\begin{equation}
\Delta_\gamma=-\eta(F+\gamma P)^+g,
\qquad
\widehat\Delta_\gamma=-\eta(P\widehat F P+\gamma P)^+P\widehat g.
\end{equation}
If $\varepsilon_F<a$, then
\begin{equation}
\|\widehat\Delta_\gamma-\Delta_\gamma\|
\le\eta\left[
\frac{\varepsilon_g}{a-\varepsilon_F}
+\frac{\varepsilon_F\|g\|}{a(a-\varepsilon_F)}
\right],
\label{eq:known_projector_error}
\end{equation}
and $(I-P)\widehat\Delta_\gamma=0$ exactly.
\end{theorem}
\begin{theorem}
\label{thm:finite_sample_iteration}
Let $P$ be a fixed orthogonal projector, let $\cL$ be $\beta$-smooth on the affine space $\theta_0+\operatorname{im}P$, and assume $\cL\ge\cL_\star$ there.  At iterate $\theta_t$, write
\begin{equation}
g_t=P\nabla\cL(\theta_t),
\qquad
f_{\min}P\preceq F_t=PF_tP\preceq f_{\max}P,
\end{equation}
where $0<f_{\min}\le f_{\max}$ uniformly.  Define
\begin{equation}
a=f_{\min}+\gamma,
\qquad
b=f_{\max}+\gamma,
\end{equation}
and form the estimated update
\begin{equation}
\theta_{t+1}=\theta_t+\widehat\Delta_t,
\qquad
\widehat\Delta_t
=-\eta(P\widehat F_tP+\gamma P)^+P\widehat g_t.
\label{eq:finite_sample_iteration}
\end{equation}
Suppose, simultaneously for $t=0,\ldots,T-1$,
\begin{equation}
\|P\widehat F_tP-F_t\|\le\varepsilon_F,
\qquad
\|P\widehat g_t-g_t\|\le\varepsilon_g,
\end{equation}
and
\begin{equation}
\eta\le\frac{a}{\beta},
\qquad
\varepsilon_F\le\frac{a^2}{64b},
\label{eq:finite_sample_iteration_conditions}
\end{equation}
with the first restriction omitted when $\beta=0$.  Then
\begin{equation}
\frac1T\sum_{t=0}^{T-1}\|g_t\|^2
\le
\frac{4b\bigl(\cL(\theta_0)-\cL_\star\bigr)}{\eta T}
+\frac{144b^2}{a^2}\varepsilon_g^2.
\label{eq:finite_sample_stationarity}
\end{equation}
Moreover,
\begin{equation}
(I-P)(\theta_t-\theta_0)=0
\qquad (t=0,\ldots,T),
\label{eq:finite_sample_zero_drift}
\end{equation}
so a fixed known gauge projector prevents cumulative representative drift exactly.
\end{theorem}

For $\gamma=0$, the intrinsic quotient-gradient norm satisfies
\begin{equation}
\|\grad_{\bar F}\bar\cL\|_{\bar F}^{2}
=g_t^{\mathsf T}F_t^+g_t
\in\left[\frac{\|g_t\|^2}{f_{\max}},
\frac{\|g_t\|^2}{f_{\min}}\right],
\end{equation}
so \cref{eq:finite_sample_stationarity} also controls the natural quotient stationarity measure under the same uniform Fisher-gap assumption.

For bounded entrywise Fisher estimators and simultaneous gradient confidence bounds, \cref{cor:finite_sample_shot_budget} converts \cref{thm:finite_sample_iteration} into explicit sufficient per-iteration shot budgets and an entrywise-protocol total repetition complexity $\widetilde O(\varepsilon^{-4})$.

The exact cumulative identity \cref{eq:finite_sample_zero_drift} uses a fixed linear projector.  For a parameter-dependent horizontal distribution $P_\theta$, each update is horizontal at its current iterate, but a gauge-fixing retraction is still needed to control the stored representative over finite steps.

When the Fisher and gradient confidence radii both scale as $N^{-1/2}$, the leading local error bound also gives an optimal Fisher-versus-gradient shot split; the explicit rule is \cref{eq:shot_allocation} in Appendix~\ref{app:finite_shot_aux}.

For the estimated-projector update $\widehat\Delta$ and the exact-projector target $\Delta_\gamma$, \cref{thm:estimated_projector_update} gives an explicit $O(\varepsilon_P)$ gauge-leakage term together with the Fisher, gradient, and resolvent errors. In its notation, set $a=f_{\min}+\gamma$, $q=g^{\mathsf T}(F+\gamma P)^+g$, and $\varepsilon_\Delta=\|\widehat\Delta-\Delta_\gamma\|$. If the quotient loss is $\beta$-smooth and $\eta\le a/\beta$, with no restriction on $\eta$ when $\beta=0$, the descent lemma gives
\begin{equation}
\cL(\theta+\widehat\Delta)-\cL(\theta)
\le-\frac{\eta}{2}q
+\left(\|g\|+\beta\eta\sqrt{q/a}\right)\varepsilon_\Delta
+\frac{\beta}{2}\varepsilon_\Delta^2.
\label{eq:finite_shot_descent}
\end{equation}
Together with \cref{eq:estimated_projector_error}, this converts subspace and estimator accuracy into a checkable sufficient condition for actual loss decrease.

\subsection{Exact gauge directions and physical soft modes}

Finite data refine the geometric decomposition asymmetrically: a known structural gauge is removable, a direction with a positive lower confidence bound is certified physical, and a confidence interval that overlaps zero leaves the direction unresolved. Spectral evidence alone cannot promote that third class to an exact gauge.

\begin{proposition}
\label{prop:soft_indistinguishability}
Suppose only $\|\widehat F-F\|\le \rho_F$ is known and
$\widehat F=\diag(\rho_F,f_2,\ldots,f_p)$ with $f_j>2\rho_F$.
Both
\begin{equation}
F_0=\diag(0,f_2,\ldots,f_p),
\qquad
F_1=\diag(2\rho_F,f_2,\ldots,f_p)
\end{equation}
belong to the same confidence ball.
No decision rule using only $\widehat F$ and $\rho_F$ can certify whether the first direction is an exact zero or a physical mode of size $2\rho_F$.
\end{proposition}

The construction is a confidence-set identifiability boundary rather than a sampling-distribution lower bound. Using only $\widehat F$ and $\rho_F$, the actionable spectral classes are therefore ``certified nonzero'' and ``unresolved'' outside any independently known gauge space. If the true rank and a nonzero Fisher gap are supplied independently, \cref{thm:certified-truncation} recovers the local Fisher kernel and licenses projection onto that kernel. Identifying the recovered kernel with a prescribed or global gauge action requires independent structural information about the state-preserving action, as in the setup of \cref{thm:finite_shot_faithfulness}.

\begin{proposition}
\label{prop:near_symmetry_kernel}
Let $F_0\succeq0$ have an $s$-dimensional kernel $\cV$ and nonzero spectral gap $f_{\mathrm{gap}}$.
If $F=F_0+R\succeq0$, $\|R\|\le\varepsilon<f_{\mathrm{gap}}/2$, then the $s$ smallest eigenvalues of $F$ lie in $[0,\varepsilon]$, while all others are at least $f_{\mathrm{gap}}-\varepsilon$.
For $\varepsilon\le f_{\mathrm{gap}}/4$, the bottom-$s$ spectral projector obeys
\begin{equation}
\|P_{\mathrm{soft}}-P_{\cV}\|\le\frac{2\varepsilon}{f_{\mathrm{gap}}}.
\label{eq:soft_subspace_bound}
\end{equation}
\end{proposition}
Closeness to an exact gauge subspace does not justify deleting a soft mode that carries task gradient.
The interaction between a proposed soft sector and the remaining physical directions can be described exactly, without assuming invariance.
Work on the active space after known exact gauge directions have been projected out, let $P$ project onto a candidate soft subspace, set $Q=I-P$, and define
\begin{equation}
A=PFP+\gamma P,
\qquad
D=QFQ+\gamma Q,
\qquad
C=PFQ,
\label{eq:soft_block_operators}
\end{equation}
where $\gamma>0$ and every inverse is restricted to its indicated subspace.

\begin{proposition}
\label{prop:soft_block_coupling}
For the full active-space ridge update
\begin{equation}
\Delta_\gamma=-\eta(F+\gamma I)^{-1}g,
\end{equation}
write $g_S=Pg$, $g_Q=Qg$, and define the Schur complement
\begin{equation}
\mathcal S_\gamma=A-CD^{-1}C^*.
\end{equation}
Then the soft component is exactly
\begin{equation}
P\Delta_\gamma
=-\eta\mathcal S_\gamma^{-1}
\left(g_S-CD^{-1}g_Q\right).
\label{eq:soft_mode_schur}
\end{equation}
\end{proposition}
If $e=\|C\|$ and $a=\lambda_{\min}(A|_{\operatorname{im}P})$ satisfy $e^2<a\gamma$, \cref{prop:soft_coupling_bound} bounds the deviation from the decoupled soft-block update by explicit $e/\gamma$ and $e^2/\gamma$ terms.
If the soft subspace is $F$-invariant, then $C=0$ and the general formula reduces to
\begin{equation}
P\Delta_\gamma
=-\eta(PFP+\gamma P)^+g_S.
\label{eq:soft_mode_bias}
\end{equation}
Thus invariance is a zero-coupling special case rather than a prerequisite for analysis.
More generally, hard deletion has zero soft-block bias exactly when the effective gradient $g_S-CD^{-1}g_Q$ vanishes.
For one invariant eigenmode with Fisher value $f$ and gradient component $c$, the squared hard-deletion bias remains $\eta^2c^2/(f+\gamma)^2$.

\begin{proposition}
\label{prop:optimal_scalar_ridge}
Consider one physical Fisher eigenmode with known $f>0$ and true gradient component $c$.
Assume that only the gradient is noisy, $\widehat c=c+\xi$, with $\mathbb E\xi=0$ and $\Var(\xi)=s^2$, while $f$ and $c$ are treated as fixed oracle quantities for this scalar risk calculation.
For the target $u_*=-c/f$ and ridge update $\widehat u_\gamma=-(c+\xi)/(f+\gamma)$,
\begin{equation}
\mathbb E(\widehat u_\gamma-u_*)^2
=\frac{c^2\gamma^2/f^2+s^2}{(f+\gamma)^2}.
\label{eq:scalar_ridge_risk}
\end{equation}
If $c\ne0$, this risk is minimized by
\begin{equation}
\gamma_*=\frac{s^2f}{c^2}.
\label{eq:scalar_ridge_optimum}
\end{equation}
If $c=0$, complete suppression is optimal for this scalar model.
\end{proposition}
The oracle optimum also yields a finite-sample rule once signal and noise uncertainty are separated by an independent pilot stage.

An independent pilot stage converts the oracle rule into a confidence-certified minimax scalar ridge: \cref{prop:confidence_scalar_ridge} uses a signal confidence radius and an independent variance bound to choose $\gamma_\delta$ and bound the conditional update risk.
For bounded pilot measurements, the signal radius can be obtained from Hoeffding's inequality or another valid bounded-data concentration bound, while the variance bound can come from a known measurement range or independent samples from exact gauge directions under a shared-noise model \cite{Hoeffding1963}. The numerical study evaluates a lower-cost proxy that reuses the observed exact-gauge block as an internal noise monitor; \cref{prop:confidence_scalar_ridge} supplies the certified independent-pilot counterpart.

\subsection{Mixed-state geometry and the depolarizing boundary}
\label{sec:mixed_extension}

Mixed states add a genuinely intrinsic source of small Fisher scales: population gaps. For $\rho=\sum_{j=1}^{D}p_j\ket j\!\bra j$, use the trace-orthonormal off-diagonal generators
\begin{equation}
X_{jk}=\frac{\ket j\!\bra k+\ket k\!\bra j}{\sqrt2},
\qquad
Y_{jk}=\frac{i(\ket j\!\bra k-\ket k\!\bra j)}{\sqrt2}.
\end{equation}
The standard symmetric-logarithmic-derivative (SLD) spectrum on the unitary orbit \cite{Hubner1992Bures,Liu2020QFIM,Koczor2022Noisy} is diagonal in these planes, with
\begin{equation}
f_{jk}=2\,\frac{(p_j-p_k)^2}{p_j+p_k},
\label{eq:standard_sld_orbit_spectrum}
\end{equation}
and $f_{jk}=0$ when $p_j+p_k=0$. Equal populations make the corresponding generator direction vertical, while a small nonzero population gap produces a genuinely physical soft direction whose Fisher scale collapses quadratically.

\begin{theorem}
\label{thm:mixed_qng_linearization}
Let
\begin{equation}
\rho_{\star}=\sum_{j=1}^{D}p_j\ket{j}\!\bra{j},
\qquad p_1>p_2>\cdots>p_D>0,
\end{equation}
and
\begin{equation}
H_{\mathrm{obj}}=\sum_{j=1}^{D}E_j\ket{j}\!\bra{j},
\qquad E_1<E_2<\cdots<E_D,
\end{equation}
so that $\rho_{\star}$ is a critical point of $\cL_H(\rho)=\Tr(H_{\mathrm{obj}}\rho)$ on its unitary orbit. For each pair $X_{jk},Y_{jk}$ with $j<k$, the Euclidean Hessian eigenvalue is
\begin{equation}
h_{jk}=(E_k-E_j)(p_j-p_k),
\end{equation}
and the linearized continuous-time SLD quotient-QNG flow has Jacobian $-\eta M_\star$ on the horizontal tangent space, with
\begin{equation}
M_\star=(F^{\mathrm{SLD}}_\star)^+H_{\mathrm{Eucl},\star}.
\end{equation}
Each pair $X_{jk},Y_{jk}$ carries the QNG eigenvalue
\begin{equation}
\mu_{jk}=\frac12(E_k-E_j)\frac{p_j+p_k}{p_j-p_k}.
\label{eq:mixed_qng_jacobian}
\end{equation}
\end{theorem}
As $p_j-p_k\to0$, the SLD Fisher scale collapses quadratically while the intrinsic QNG linearization grows as $(p_j+p_k)/(p_j-p_k)$. At exact degeneracy the orbit dimension drops and the direction becomes vertical, so the divergence marks approach to a lower stratum rather than a valid update.

Global depolarization isolates the statistical consequence of this intrinsic collapse because it preserves isotropy on the surviving full-control orbit.
\begin{theorem}
\label{thm:depolarizing_noise_geometry}
Let
\begin{equation}
\rho(q)=(1-q)\ket{0}\!\bra{0}+q\frac{I}{D},
\qquad 0\le q<1,
\end{equation}
on a $D$-dimensional full-control sector. Along its unitary orbit, the horizontal SLD-QFIM has rank $2D-2$, is isotropic, and has nonzero eigenvalue
\begin{equation}
f(q)=2\,\frac{(1-q)^2}{1-q+2q/D}.
\label{eq:depolarized_fisher_eigenvalue}
\end{equation}
For a linear objective diagonal in the basis containing $\ket0$, the deterministic horizontal quotient-QNG update is rescaled relative to the corresponding pure-state update by
\begin{equation}
\mu(q)=\frac{1-q+2q/D}{1-q}.
\label{eq:depolarized_qng_amplifier}
\end{equation}
More generally, in any scalar Fisher direction with gradient $g$, Fisher value $f$, and an unbiased $M$-shot gradient estimator of variance $\sigma_g^2/M$,
\begin{equation}
\operatorname{SNR}_{\mathrm{update}}^2
=\frac{(g/f)^2}{\sigma_g^2/(Mf^2)}
=\frac{Mg^2}{\sigma_g^2}.
\label{eq:qng_update_snr}
\end{equation}
Thus the inverse-Fisher factor cancels from the update signal-to-noise ratio. If $g(q)=(1-q)g_0$ while $\sigma_g^2$ remains bounded away from zero, then
\begin{equation}
\operatorname{SNR}_{\mathrm{update}}^2(q)=O((1-q)^2).
\label{eq:depolarized_update_snr}
\end{equation}
\end{theorem}

The deterministic and statistical scales therefore move in opposite senses. In the many-body regime $D(1-q)\gg1$, $\mu(q)=1+o(1)$ can nearly restore the pure-state mean-update scale, while the inverse-Fisher contribution to update variance is
\begin{equation}
f(q)^{-2}=\frac{(1-q+2q/D)^2}{4(1-q)^4}.
\label{eq:depolarized_variance_amplifier}
\end{equation}
If a depth-$L$ circuit has a global depolarizing factor $1-p$ per layer, $1-q_L=(1-p)^L$ and
\begin{align}
\mu(q_L)&=1+\frac{2}{D}\bigl((1-p)^{-L}-1\bigr),\\
f(q_L)^{-2}&=\frac{\left((1-p)^L+2(1-(1-p)^L)/D\right)^2}{4(1-p)^{4L}}.
\end{align}
Inverse-Fisher scaling can therefore normalize a deterministic update while amplifying its fluctuations by the same preconditioning mechanism; it cannot restore signal-to-noise or noiseless shot complexity under this depolarizing model \cite{Wang2021Noise}.

\subsection{Slater/Givens numerical validation}
\label{subsec:slater_experiment}

The numerical study uses a classically simulable number-conserving fermionic model to test consequences of the transfer and finite-shot results rather than to redraw a closed formula.
Eight modes and four particles give a fixed-particle Hilbert space of dimension $\binom84=70$.
A 28-parameter real-Givens circuit explores the totally real slice
$\operatorname{Gr}_{\bbR}(4,8)\subset\operatorname{Gr}_{\bbC}(4,8)$.
A dense coordinate mixing hides 16 physical and 12 exact gauge combinations, while a prescribed physical chart introduces nontrivial coordinate conditioning.
\Cref{fig:slater_exact_bridge} reports the resulting exact circuit-to-orbit consistency checks.

\begin{figure}[t]
\centering
\includegraphics[width=0.82\linewidth,height=0.3143333333\linewidth]{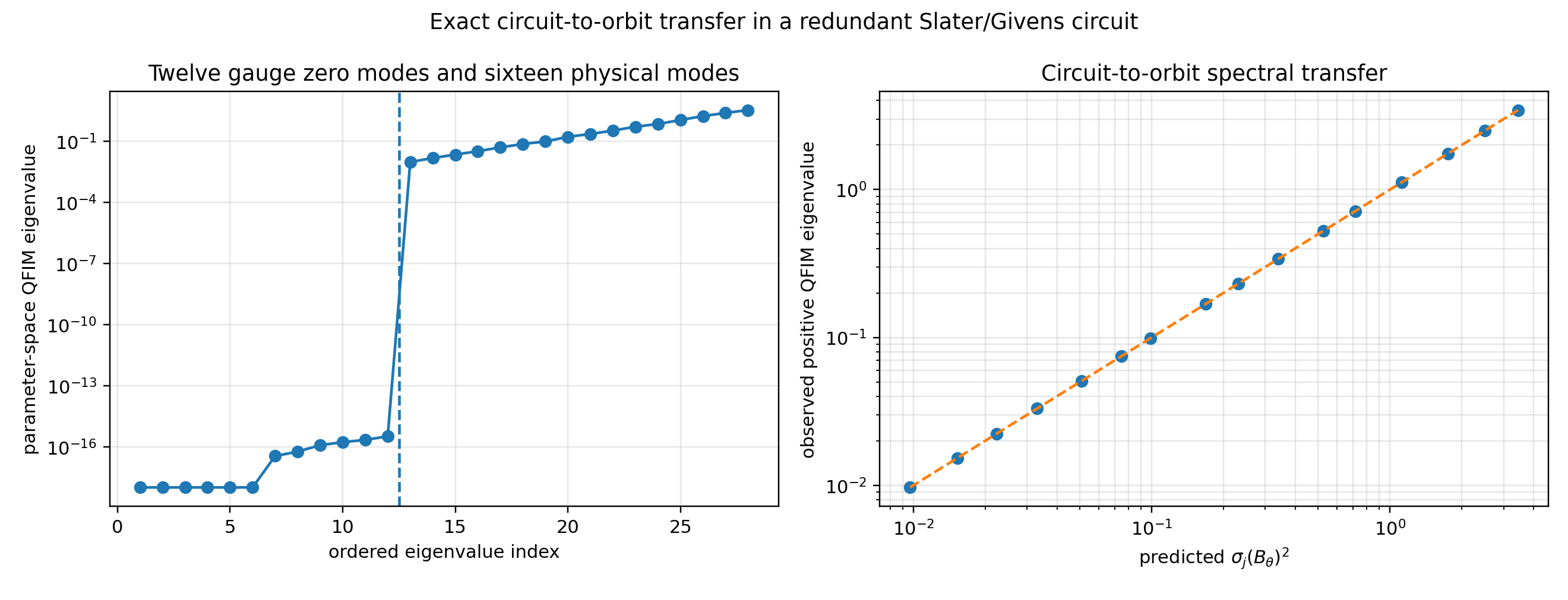}
\caption{Circuit-to-orbit Fisher transfer in the redundant Slater/Givens circuit.
The intrinsic real-Grassmann Fisher operator is isotropic, while the gate-coordinate QFIM has condition number $356.548$.
The factorization $F=B^{\mathsf T}B$ has relative error $1.273\times10^{-10}$, and the exact pseudoinverse update induces the intrinsic real-Grassmann gradient to angular error $1.207\times10^{-6}$ degrees.}
\label{fig:slater_exact_bridge}
\end{figure}

The exact calculation gives $\rank F=16$, $\dim\ker F=12$, and
\begin{equation}
\frac{\|F-B^{\mathsf T}B\|_{\mathrm F}}{\|F\|_{\mathrm F}}
=1.273\times10^{-10},
\qquad
\kappa^+(F)=\kappa(B)^2=356.548.
\end{equation}
It therefore verifies the chart-induced singular-value law in
\cref{thm:circuit_orbit_factorization,cor:isotropic_spectral_transfer} in a circuit with a hidden redundant mixing.

For the finite-shot comparison, QFIM entries are estimated from symmetric directional-overlap measurements with Bernoulli noise, and energy gradients from multinomial occupation measurements.
The ambient and quotient blocks require 868 and 304 distinct shifted-circuit settings, respectively, under this particular protocol.
They are compared at equal total budgets of $1.5$, $6$, and $24$ million shots.
Before reporting angle, energy change, or gauge displacement, every estimated update is rescaled to the same exact physical projector-tangent norm $0.08$.
Thus the comparison isolates direction quality and gauge leakage; it does not compare automatically selected step sizes.
The equal-budget results are summarized in \cref{fig:slater_finite_shot}.

\begin{figure}[t]
\centering
\includegraphics[width=0.88\linewidth,height=0.264\linewidth]{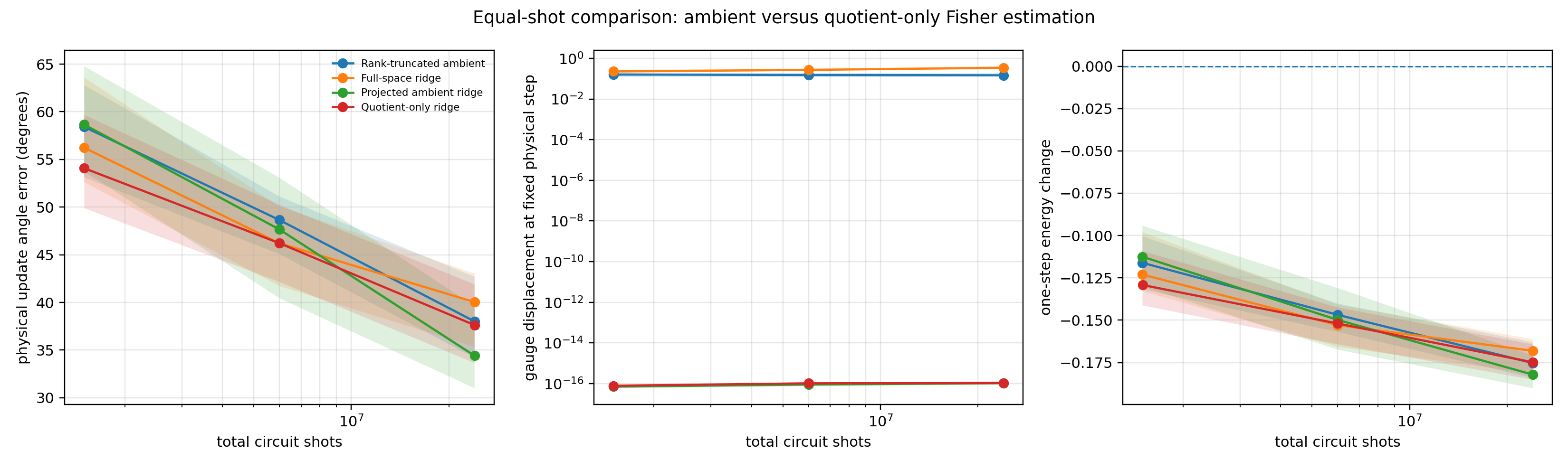}
\caption{Equal-total-shot comparison for ambient and quotient estimators.
Curves show medians and interquartile ranges over 64 replicates.
Structural projection eliminates gauge displacement, while quotient-only estimation spends the fixed budget on the smaller physical block.}
\label{fig:slater_finite_shot}
\end{figure}

At the largest budget, the projected ambient estimator has the smallest median physical-angle error, $34.40^\circ$.
The quotient-only estimator gives $37.62^\circ$, compared with $38.01^\circ$ for rank-truncated ambient estimation and $40.04^\circ$ for full-space ridge.
At the low and intermediate budgets, quotient-only estimation is more accurate because the smaller structural block receives more shots per setting.
Both structurally projected methods keep gauge displacement at numerical zero; full-space ridge has median gauge displacement $3.473\times10^{-1}$ at the largest budget.
The remaining $34$--$40^\circ$ angle errors show that even the largest budget is statistically noisy.

Quotient-only estimation is the lower-error method at the smallest budget and, at the two larger budgets, stays within $3.3^\circ$ of projected ambient estimation while using 304 rather than 868 shifted-circuit settings. The small separation at the larger budgets makes their fine ranking descriptive rather than structural; the stable conclusion is the resource--accuracy tradeoff together with exact removal of gauge displacement by structural projection.

The release experiment keeps six directions exactly redundant and opens six candidate directions continuously; \cref{fig:slater_soft_release} shows the resulting bias--leakage tradeoff.
\begin{figure}[t]
\centering
\includegraphics[width=0.88\linewidth,height=0.264\linewidth]{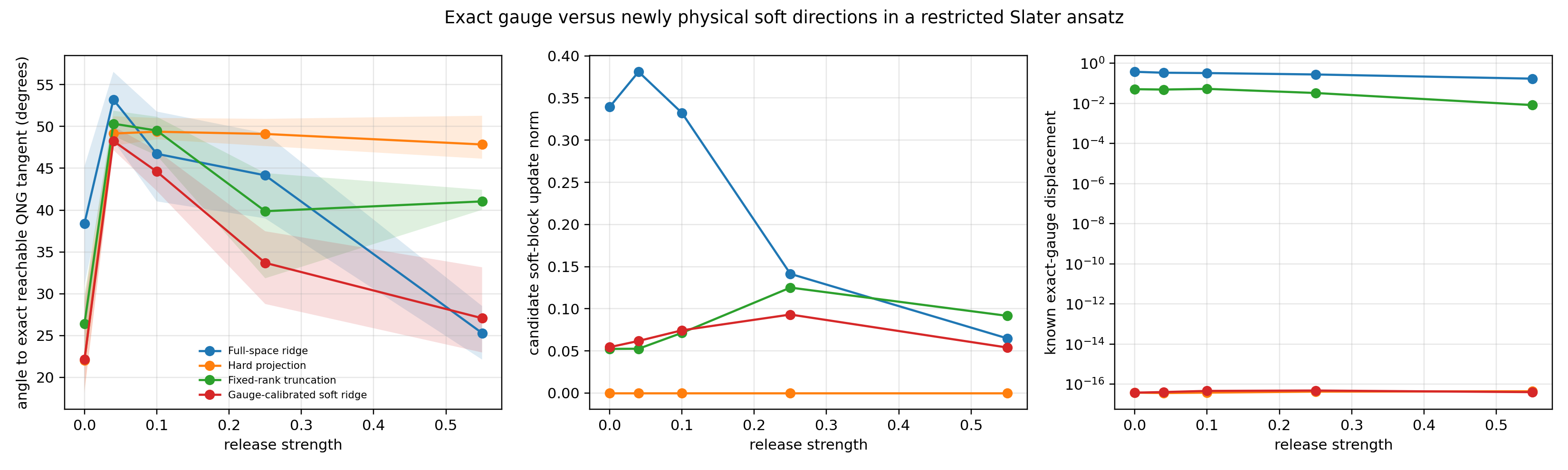}
\caption{Exact-gauge and physical-soft-mode tradeoff under continuous release.
Six candidate directions acquire physical Fisher weight while six known exact gauge directions remain redundant.
Hard projection becomes biased after release, whereas the gauge-calibrated anisotropic rule retains candidate motion and eliminates displacement in the known gauge block.
The plotted rule is the low-cost gauge-noise proxy motivated by \cref{prop:optimal_scalar_ridge}; \cref{prop:confidence_scalar_ridge} supplies the corresponding independent-pilot, confidence-certified scalar alternative.}
\label{fig:slater_soft_release}
\end{figure}

At the largest release strength, full-space ridge has median angle $25.29^\circ$ but incurs known-gauge displacement $1.698\times10^{-1}$.
Hard projection removes all candidate motion and reaches only $47.85^\circ$.
The gauge-calibrated soft rule reaches $27.07^\circ$, retains candidate-block norm $5.403\times10^{-2}$, and keeps known-gauge displacement at numerical zero.
At zero release the candidate block is also physically inactive, although the reported known-gauge diagnostic intentionally tracks only the six directions that remain redundant throughout the release path.

These experiments test the factorization and exact-gauge versus soft-mode mechanism on one engineered free-fermion instance with ideal finite-shot sampling and classically computed references. The quotient-only comparison exhibits a resource--accuracy tradeoff whose fine ranking is descriptive rather than structural; hardware scaling remains open, and the plotted soft rule is a mechanism-level proxy whose certified counterpart is \cref{prop:confidence_scalar_ridge}.

\section{Conclusion and outlook}

The main consequence of the quotient viewpoint is that singularity is not itself an optimization prescription. The state map determines first which parameter velocities are physically invisible and which orbit tangents the circuit can realize; the factorization $F_\theta=B_\theta^*M_\rho B_\theta$ then separates this structural information from coordinate distortion and intrinsic Fisher geometry. When the prescribed redundancy is faithful, the Moore--Penrose step is an intrinsic quotient gradient represented by its minimum-norm horizontal lift. When the quotient-to-orbit map is locally diffeomorphic, chart singular values cancel from the linearized QNG dynamics, leaving the intrinsic operator $M^{-1}H$ as the conditioning problem that cannot be removed by a change of circuit coordinates.

This separation exposes two complementary boundaries. The first is geometric: in the linearized critical-point dynamics, QNG removes faithful chart conditioning but does not erase anisotropy already present on the state orbit. In the trace-orthonormal full-control generator frame, the excitation-gap spread gives $\kappa_{\mathrm{QNG}}=\kappa_{\mathrm{Eucl}}$, while the mixed-state result shows that an individual closing population gap produces a divergent intrinsic QNG mode scale on approach to a lower orbit stratum. The second boundary is statistical: an estimated QFIM inside an operator-norm confidence ball cannot, by itself, distinguish an exact zero from a small physical mode. Known state-preserving structure licenses gauge projection directly; independently supplied rank-and-gap information can instead certify and recover the local Fisher kernel, while identifying that kernel with a global gauge action still requires structural information. Unresolved modes remain part of the physical optimization problem and should be regularized softly rather than deleted.

The representation-theoretic calculations show what becomes possible once the intrinsic factor is isolated. Total Fisher weight is tied to generalized entanglement through the moment map, root data resolve its distribution, and minuscule families collapse to one intrinsic scale. On cominuscule embeddings the same geometry makes fidelity QNG integrable: continuous trajectories preserve principal-defect ratios and reduce to one scalar contraction, while finite Lie-retracted steps reveal the discretization boundary through second-order ratio drift, cubic convergence at $\eta=2$, and stability edge $\eta=4$. Under depolarization, isotropy survives but inverse-Fisher scaling moves the deterministic update and its fluctuations together, so lost signal-to-noise is not recovered by preconditioning.

The Slater/Givens experiment tests these mechanisms in a redundant free-fermion chart and shows how the structural distinction changes finite-shot resource allocation. The next questions are correspondingly structural rather than merely numerical: which homogeneous or nonhomogeneous orbits admit comparable integrable fidelity laws, how learned redundancy projectors can acquire end-to-end confidence guarantees, and how the intrinsic mixed-state boundary changes under correlated or nonunital noise.

\appendix

\section{Quantum Fisher information and quotient-geometric preliminaries}
\label{app:qfi_geometry}

The ingredients used in the main text are standard, but they come from several neighboring subjects.
We follow Nielsen and Chuang for finite-dimensional quantum states, unitary dynamics, and observables \cite{NielsenChuang2010}; Amari and Nagaoka for Fisher geometry and natural gradients \cite{AmariNagaoka2000}; and Petz for quantum statistical models and SLD quantum Fisher information \cite{Petz2008}.
Smooth group actions, orbit spaces, and quotient manifolds are used in the sense of Lee \cite{Lee2013}, while horizontal lifts and optimization on matrix and quotient manifolds follow Absil, Mahony, and Sepulchre and Boumal \cite{AbsilMahonySepulchre2008,Boumal2023}.
This appendix fixes the conventions needed to connect these structures when a circuit parametrization contains state-preserving redundancy.

\subsection{State-space quantum Fisher information}
\label{app:sld_qfi}

Quantum Fisher information (QFI) quantifies the local statistical distinguishability of a smooth family of quantum states.
For a one-parameter variation, the SLD QFI is the supremum of the classical Fisher information over all quantum measurements \cite{BraunsteinCaves1994}.
The associated multiparameter object is most naturally expressed as a bilinear form on state-space tangent vectors before any circuit coordinates are chosen.

Let $X$ and $Y$ be Hermitian tangent variations at a density operator $\rho$ along a smooth fixed-rank family.
Symmetric logarithmic derivatives $L_X$ and $L_Y$ satisfy
\begin{equation}
X=\frac{1}{2}\bigl(\rho L_X+L_X\rho\bigr),
\qquad
Y=\frac{1}{2}\bigl(\rho L_Y+L_Y\rho\bigr).
\label{eq:sld_tangent_definition}
\end{equation}
The SLD quantum Fisher bilinear form is
\begin{equation}
g^{\mathrm{SLD}}_{\rho}(X,Y)
:=\frac{1}{2}\Tr\!\left[\rho\{L_X,L_Y\}\right]
=\mathrm{Re}\,\Tr\!\left(\rho L_XL_Y\right).
\label{eq:sld_qfi_bilinear}
\end{equation}
If $\rho=\sum_a p_a\ket{a}\!\bra{a}$, the same form can be written without choosing representatives of the SLDs as
\begin{equation}
g^{\mathrm{SLD}}_{\rho}(X,Y)
=2\!\sum_{a,b:\,p_a+p_b>0}
\frac{\mathrm{Re}\!\left(\bra{a}X\ket{b}\bra{b}Y\ket{a}\right)}{p_a+p_b}.
\label{eq:sld_qfi_spectral}
\end{equation}
For rank-deficient states, an SLD need not be unique on the kernel of $\rho$, but \cref{eq:sld_qfi_bilinear,eq:sld_qfi_spectral} define the same unambiguous pairing on tangent variations satisfying \cref{eq:sld_tangent_definition}.
Working on a fixed-rank stratum avoids changes in the tangent model at rank-changing boundary points; the mixed-state results in Section~\ref{sec:mixed_extension} are formulated on unitary orbits, where the spectrum and hence the rank are fixed.

\subsection{Pullback QFIM and coordinate formulas}
\label{app:pullback_geometry}

For a smooth state map $\theta\mapsto\rho(\theta)$, the parameter-space quantum Fisher tensor is the pullback $F=\rho^{*}g^{\mathrm{SLD}}$, as written in \cref{eq:qfim_pullback}.
In local coordinates, let $L_i$ solve the SLD equation for $\partial_i\rho$:
\begin{equation}
\partial_i\rho
=\frac{1}{2}\bigl(\rho L_i+L_i\rho\bigr).
\end{equation}
The matrix seen in quantum-machine-learning calculations is
\begin{equation}
F_{ij}(\theta)
=g^{\mathrm{SLD}}_{\rho(\theta)}(\partial_i\rho,\partial_j\rho)
=\frac{1}{2}\Tr\!\left[\rho\{L_i,L_j\}\right]
=\mathrm{Re}\,\Tr\!\left(\rho L_iL_j\right).
\label{eq:qfim_sld_matrix}
\end{equation}
Thus $F_{\theta}(v,w)=v^{\mathsf T}F(\theta)w$ in the chosen coordinate basis.
The term QFI refers to the intrinsic state-space pairing or to its value along a specified state direction, whereas QFIM refers to the coordinate matrix induced by a parametrized family.

Two elementary pullback properties drive the quotient construction.
First, if $d\rho_{\theta}(v)=0$, then bilinearity gives
\begin{equation}
F_{\theta}(v,w)
=g^{\mathrm{SLD}}_{\rho(\theta)}\!\left(0,d\rho_{\theta}(w)\right)=0
\qquad \forall w\in T_{\theta}\Theta.
\end{equation}
Hence $\ker d\rho_{\theta}\subseteq\ker F_{\theta}$.
Second, if a diffeomorphism $\Psi$ of the parameter manifold preserves the state map, $\rho\circ\Psi=\rho$, then functoriality of pullback gives
\begin{equation}
\Psi^{*}F
=\Psi^{*}\rho^{*}g^{\mathrm{SLD}}
=(\rho\circ\Psi)^{*}g^{\mathrm{SLD}}
=F.
\end{equation}
These are precisely the kernel and invariance mechanisms used in the main text.

Under a smooth change of coordinates $\theta=\theta(\varphi)$ with Jacobian $J=\partial\theta/\partial\varphi$, the coordinate matrix transforms as
\begin{equation}
F_{\varphi}=J^{\mathsf T}F_{\theta}J.
\end{equation}
This congruence law is the matrix expression of the coordinate invariance of the pullback tensor.

\subsection{Pure-state specialization}
\label{app:pure_qfi}

For a normalized pure state $\rho=\ket{\psi}\!\bra{\psi}$, differentiation of $\rho^{2}=\rho$ shows that one may choose the SLD of a tangent variation $d\rho$ as $L=2d\rho$.
Substitution into \cref{eq:sld_qfi_bilinear} yields four times the Fubini--Study metric:
\begin{equation}
g^{\mathrm{SLD}}_{\rho}(d\rho,d\rho')
=4\,\mathrm{Re}\!\left(
\braket{d\psi}{d\psi'}
-\braket{d\psi}{\psi}\braket{\psi}{d\psi'}
\right).
\end{equation}
For $\ket{\psi(\theta)}=U(\theta)\ket{\psi_0}$, the local generator
\begin{equation}
A_i=i(\partial_iU)U^{\dagger}
\end{equation}
is Hermitian and satisfies $\partial_i\ket{\psi}=-iA_i\ket{\psi}$.
Inserting this identity into the Fubini--Study expression gives the derivative-overlap and generator-covariance forms collected in \cref{eq:QFIM_pure}.
In particular, the convention used throughout the paper is
\begin{equation}
F_{ij}=4\,\mathrm{Re}\,\mathrm{Cov}_{\psi}(A_i,A_j)
=2\,\bra{\psi}\{\Delta A_i,\Delta A_j\}\ket{\psi}.
\end{equation}
This factor convention is responsible for the nonzero full-control eigenvalue~$2$ in \cref{eq:full_control_spectrum}; conventions based directly on the Fubini--Study tensor differ by an overall factor of four.

\subsection{Natural gradients and pseudoinverses}
\label{app:natural_gradient_background}

On a manifold with a positive-definite metric $h$, the natural or Riemannian gradient of a smooth function $f$ is defined by
\begin{equation}
h(\operatorname{grad}_{h}f,\cdot)=df.
\end{equation}
In local coordinates this becomes $\operatorname{grad}_{h}f=H^{-1}\nabla f$, where $H$ is the metric matrix and $\nabla f$ is the Euclidean coordinate gradient.
For the pulled-back QFI geometry this yields the nonsingular QNG formula in the main text.

If $F$ is symmetric positive semidefinite, the Moore--Penrose pseudoinverse acts as the inverse on $(\ker F)^{\perp}=\operatorname{im}F$ and vanishes on $\ker F$:
\begin{equation}
F^{+}=\bigl(F|_{(\ker F)^{\perp}}\bigr)^{-1}P_{(\ker F)^{\perp}}.
\end{equation}
Accordingly, $F^{+}b$ is the minimum-Euclidean-norm solution of $Fx=b$ when $b\in\operatorname{im}F$, and the minimum-norm least-squares solution otherwise.
The qualifier ``Euclidean'' is important: the selected representative depends on the background inner product of the parameter chart, whereas its projection to a faithful quotient is intrinsic.
Theorem~\ref{thm:quotient_geometry_update} identifies the conditions under which the pseudoinverse QNG step is exactly the horizontal lift of the quotient natural gradient.

For a background Riemannian metric $g$ with Levi--Civita connection $\nabla$, the Riemannian Hessian used in the gauge-fixed critical-point analysis is the bilinear form
\begin{equation}
\mathrm{Hess}\,f(v,w):=\langle \nabla_v(\mathrm{grad}\,f),w\rangle_g.
\end{equation}

\subsection{Group actions and quotient tensors}
\label{app:quotient_preliminaries}

Let a Lie group $\Gamma$ act smoothly on a manifold $U$, and assume that the orbit space $U/\Gamma$ is a smooth manifold for which the canonical projection $\pi:U\to U/\Gamma$ is a submersion.
The kernel of $d\pi_{\theta}$ is the tangent space to the orbit through $\theta$, called the vertical space.
A smooth covariant two-tensor $B$ on $U$ is the pullback of a tensor $\bar B$ on $U/\Gamma$ precisely when it is invariant under the action and annihilates vertical directions in each argument:
\begin{equation}
\Phi_g^{*}B=B,
\qquad
B_{\theta}(v,\cdot)=0\quad \text{for every }v\in\ker d\pi_{\theta}.
\end{equation}
Vertical annihilation makes the value independent of the chosen tangent lift, while invariance makes it independent of the chosen representative point on an orbit.
If $B$ is positive semidefinite and $\ker B_{\theta}=\ker d\pi_{\theta}$, then the descended tensor is positive definite and therefore defines a Riemannian metric on $U/\Gamma$.
These facts underlie the definitions of a $\Gamma$-basic and faithful QFIM in \cref{def:basic_fisher} \cite{Lee2013,AbsilMahonySepulchre2008,Boumal2023}.
No smoothness of the quotient is inferred merely from constant orbit type; it is imposed explicitly in the main results.

\section{Proofs for quotient geometry, gauge fixing, and sector reduction}

\subsection{Quotient metric and pseudoinverse lift}
\begin{proof}[Proof of \cref{thm:quotient_geometry_update}]
First consider descent of the QFIM tensor.
For $[\theta]\in U/\Gamma$ and tangent classes $[u],[v]\in T_{[\theta]}(U/\Gamma)$, choose representatives $u,v\in T_{\theta}U$ and define
\[
\bar F_{[\theta]}([u],[v]) := F_{\theta}(u,v).
\]
If $u$ is replaced by $u+w$ with $w\in\cV_{\theta}$, then $F_{\theta}(w,v)=0$ because $F$ is basic.
The same argument applies to the second argument.
If the representative point is changed from $\theta$ to $g\cdot\theta$, the $\Gamma$-invariance part of basicness gives
\[
F_{g\cdot\theta}\!\left((d\Phi_g)_{\theta}u,(d\Phi_g)_{\theta}v\right)=F_{\theta}(u,v).
\]
Thus $\bar F$ is well defined on the quotient.
Smoothness follows from the smoothness of $F$ and the assumed smooth quotient structure.
If $\ker F_{\theta}=\cV_{\theta}$, a nonzero tangent class has no representative lying entirely in $\ker F_{\theta}$, so the descended tensor is positive definite and hence a Riemannian metric.

For the pseudoinverse lift, the faithful regime gives the orthogonal decomposition
\[
T_{\theta}U=\cV_{\theta}\oplus \mathsf{H}_{\theta},
\qquad
\mathsf{H}_{\theta}:=(\ker F(\theta))^{\perp}=\operatorname{im}F(\theta),
\]
and $d\pi_{\theta}$ restricts to an isomorphism from $\mathsf H_{\theta}$ to $T_{[\theta]}(U/\Gamma)$.
Let $\delta[\theta]=-\eta\,\mathrm{grad}_{\bar F}\bar{\cL}([\theta])$.
Its unique horizontal lift $\Delta\theta\in\mathsf H_{\theta}$ is characterized by
\[
F(\theta)(\Delta\theta,v)
=-\eta\,d\bar{\cL}_{[\theta]}(d\pi_{\theta}v)
=-\eta\,d\cL_{\theta}(v)
\qquad \forall v\in\mathsf H_{\theta},
\]
where $\cL=\bar{\cL}\circ\pi$.
In the ambient Euclidean coordinates this becomes
\[
F(\theta)(\Delta\theta,v)=-\eta\,\langle \nabla\cL(\theta),v\rangle
\qquad \forall v\in\mathsf{H}_{\theta}.
\]
Since $F(\theta)|_{\mathsf H_{\theta}}$ is positive definite, the unique solution is
\[
\Delta\theta=-\eta\bigl(F(\theta)|_{\mathsf H_{\theta}}\bigr)^{-1}P_{\mathsf H}\nabla\cL(\theta),
\]
which is \cref{eq:quotient_horizontal}.
The Moore--Penrose pseudoinverse of a symmetric positive semidefinite matrix satisfies
\[
F(\theta)^{+}=\bigl(F(\theta)|_{\mathsf H_{\theta}}\bigr)^{-1}P_{\mathsf H},
\]
so \cref{eq:quotient_pinv} follows.
Since $d\pi_{\theta}\Delta\theta=\delta[\theta]$ by construction, \cref{eq:quotient_projection_gradient} follows as well.
\end{proof}

\subsection{Sector block structure}
\begin{proof}[Proof of \cref{prop:sector_blocks}]
For part~(a), let $\rho_{0}=\ket{\psi_{0}}\!\bra{\psi_{0}}$ be supported on a single isotypic component $\cH_{\lambda_{0}}$.
Under the commutant ansatz \cref{eq:block_ansatz}, each parameter block $\theta^{(\lambda)}$ acts only on $\cH_{\lambda}$ and is the identity elsewhere.
Therefore the evolved state $\ket{\psi(\theta)}$ remains inside $\cH_{\lambda_{0}}$.
If an index $i$ belongs to a block $\mu\neq\lambda_{0}$, then the corresponding tangent generator $A_i$ vanishes on $\cH_{\lambda_{0}}$.
Hence $A_i\ket{\psi}=0$, so $\bra{\psi}A_i\ket{\psi}=0$ and $\bra{\psi}A_iA_j\ket{\psi}=0$ for every $j$.
Substituting these identities into the covariance formula \cref{eq:QFIM_pure} shows that every matrix entry involving such an index vanishes.
Thus only the $\lambda_0$-block can be nonzero, which proves \cref{eq:qfim_block}.

For part~(b), write $\rho_0=\bigoplus_{\lambda}\rho_{0,\lambda}$ with respect to \cref{eq:isotypic}.
Because the ansatz \cref{eq:block_ansatz} is block diagonal, the evolved state remains block diagonal:
\[
\rho(\theta)=\bigoplus_{\lambda}\rho_{\lambda}(\theta^{(\lambda)}),\qquad
\rho_{\lambda}(\theta^{(\lambda)})=(I_{V_\lambda}\otimes U_{\lambda})\rho_{0,\lambda}(I_{V_\lambda}\otimes U_{\lambda})^{\dagger}.
\]
A tangent direction supported in the $\mu$-block has derivative $\partial_i\rho(\theta)$ supported in the same block.
The SLD equation
\[
\partial_i\rho(\theta)=\tfrac12\bigl(L_i\rho(\theta)+\rho(\theta)L_i\bigr)
\]
therefore decouples blockwise, so one may choose $L_i$ with support only on the $\mu$-block as well.
If $i$ and $j$ belong to different blocks, then $L_iL_j$ is block off-diagonal and its trace against the block-diagonal state $\rho(\theta)$ vanishes.
Hence $F^{\mathrm{SLD}}_{ij}=\operatorname{Re}\Tr(\rho L_iL_j)=0$ for cross-block indices, proving block diagonality.
\end{proof}

\subsection{Hessian zero modes and gauge-fixed representatives}
\begin{proof}[Proof of \cref{lem:zero_modes_transversality}]
For part~(a), let $\xi\in\mathfrak{r}$ and let $X_{\xi}$ be the induced vector field.
Invariance implies $df(X_{\xi})=0$ everywhere.
For any vector field $Y$, compute the covariant derivative:
\begin{equation}
(\nabla_{Y}df)(X_{\xi}) = Y(df(X_{\xi}))-df(\nabla_{Y}X_{\xi})=-df(\nabla_{Y}X_{\xi}).
\end{equation}
By definition $(\nabla_{Y}df)(X_{\xi})=\mathrm{Hess}\,f(Y,X_{\xi})$.
At a critical point, $df=0$, so the right-hand side vanishes, giving $\mathrm{Hess}\,f(Y,X_{\xi})=0$ for all $Y$.
Symmetry of the Hessian yields $\mathrm{Hess}\,f(X_{\xi},Y)=0$ for all $Y$.
Since $\cV_{\theta_{\star}}$ is spanned by $\{X_{\xi}(\theta_{\star})\}$, the claim follows.

For part~(b), let $H$ be the self-adjoint Hessian operator at $\theta_\star$ with respect to $g$.
By part~(a), $\cV_{\theta_\star}\subseteq \ker(H)$.
The assumption that the quotient Hessian at $[\theta_\star]$ is positive definite means exactly that the induced quadratic form on
$T_{\theta_\star}\Theta/\cV_{\theta_\star}$ is positive definite.
Using the $g$-orthogonal decomposition
\[
T_{\theta_\star}\Theta=\cV_{\theta_\star}\oplus \mathsf H_{\theta_\star},
\]
this is equivalent to positive definiteness of $H|_{\mathsf H_{\theta_\star}}$ and therefore to
$H\succeq0$ together with $\ker(H)=\cV_{\theta_\star}$.
Consequently $H$ has zero eigenvalues precisely along the vertical directions and strictly positive eigenvalues on the slice.
The displayed formula for $\kappa_{\mathrm{slice}}$ is then the ordinary spectral condition number of the positive-definite restriction.
Since $H$ has a nontrivial kernel in ambient coordinates, its ambient condition number is infinite.
\end{proof}

\begin{proposition}
\label{prop:soft_gauge}
In the setting of part~(b) of \cref{lem:zero_modes_transversality}, choose local coordinates around $\theta_{\star}$ and let $P_{\cV}$ be the fixed $g$-orthogonal projection onto $\cV_{\theta_{\star}}$.
For
\begin{equation}
f_{\beta}(\theta)=f(\theta)+\frac{\beta}{2}\,\Vert P_{\cV}(\theta-\theta_{\star})\Vert_{g}^{2},
\qquad \beta>0,
\end{equation}
the Hessian at $\theta_{\star}$ has eigenvalue $\beta$ along $\cV_{\theta_{\star}}$ and the original spectrum of $H\vert_{\mathsf{H}_{\theta_{\star}}}$ along $\mathsf{H}_{\theta_{\star}}$.
Its condition number is
\begin{equation}
\kappa(H_{\beta})=\frac{\max\{\lambda_{\max}(H\vert_{\mathsf{H}_{\theta_{\star}}}),\,\beta\}}
{\min\{\lambda_{\min}(H\vert_{\mathsf{H}_{\theta_{\star}}}),\,\beta\}}.
\end{equation}
\end{proposition}

\begin{proof}[Proof of \cref{prop:soft_gauge}]
In the chosen local coordinates, the penalty term is a fixed quadratic form,
\[
q(\theta):=\frac{\beta}{2}\,\|P_{\cV}(\theta-\theta_\star)\|_g^2.
\]
Its Hessian at $\theta_\star$ equals $\beta P_{\cV}$.
Therefore
\[
H_\beta = H + \beta P_{\cV}.
\]
Because $P_{\cV}$ vanishes on $\mathsf H_{\theta_\star}$ and is the identity on $\cV_{\theta_\star}$,
$H_\beta$ acts as $H|_{\mathsf H_{\theta_\star}}$ on horizontal directions and as multiplication by $\beta$ on vertical directions.
Thus the spectrum of $H_\beta$ is exactly the union of the spectrum of $H|_{\mathsf H_{\theta_\star}}$ and the repeated eigenvalue $\beta$ on $\cV_{\theta_\star}$.
The condition-number formula follows immediately from taking the ratio of the largest and smallest positive eigenvalues.
\end{proof}

\begin{proof}[Proof of \cref{thm:representative_convergence}]
For part~(a), let $\gamma:[0,\infty)\to K\subset V$ be as stated and define $\vartheta=s\circ\gamma$.
Because $s$ is smooth on the compact set $K$, its differential is uniformly bounded there:
\[
\|Ds_{x}\|\le C_{K}\qquad (x\in K)
\]
for some finite constant $C_{K}$, where the norm is taken from the quotient tangent norm induced by $\bar F$ to the ambient Euclidean norm on $S\subset\bbR^{p}$.
Since $\gamma$ is absolutely continuous, so is $\vartheta$, and for almost every $t$ one has
\[
\|\dot\vartheta(t)\|_{2}=\|Ds_{\gamma(t)}\dot\gamma(t)\|_{2}\le C_{K}\,\|\dot\gamma(t)\|_{\bar F}.
\]
Integrating gives
\[
\operatorname{Len}_{\bbR^{p}}(\vartheta)=\int_{0}^{\infty}\|\dot\vartheta(t)\|_{2}\,dt
\le C_{K}\int_{0}^{\infty}\|\dot\gamma(t)\|_{\bar F}\,dt
=C_{K}\,\operatorname{Len}_{\bar F}(\gamma),
\]
which proves the length bound.
If moreover $\gamma(t)\to [\theta_{\star}]$, then continuity of $s$ implies
\[
\vartheta(t)=s(\gamma(t))\to s([\theta_{\star}])=\theta_{\star},
\]
so the gauge-fixed representatives converge in ambient coordinates.

For part~(b), the analyticity of $V$ and $\bar F$ makes the negative gradient vector field of $\bar\cL$ real analytic.
The hypothesis that $\gamma([0,\infty))$ has compact closure implies that the trajectory is bounded and has a nonempty compact $\omega$-limit set.
Since
\[
\frac{d}{dt}\bar\cL(\gamma(t))= -\eta\,\|\grad_{\bar F}\bar\cL(\gamma(t))\|_{\bar F}^{2}\le 0,
\]
$\bar\cL(\gamma(t))$ is monotone decreasing and therefore converges to some limit value $L_\infty$.
By the Lojasiewicz--Simon gradient inequality for real-analytic functions on analytic Riemannian manifolds, every accumulation point $x_{\infty}$ of the trajectory admits a neighborhood $W$ and constants $C>0$ and $\alpha\in [1/2,1)$ such that
\[
|\bar\cL(x)-\bar\cL(x_{\infty})|^{\alpha}
\le C\,\|\grad_{\bar F}\bar\cL(x)\|_{\bar F}
\qquad (x\in W)
\]
\cite{Simon1983,AbsilMahonyAndrews2005}.
The standard Lojasiewicz argument then implies that once the trajectory enters $W$, the integral
\[
\int_{t_{0}}^{\infty}\|\dot\gamma(t)\|_{\bar F}\,dt
=\eta\int_{t_{0}}^{\infty}\|\grad_{\bar F}\bar\cL(\gamma(t))\|_{\bar F}\,dt
\]
is finite, so the trajectory has finite $\bar F$-length and converges to a single critical point $[\theta_{\infty}]$ of $\bar\cL$.
Applying part~(a) to the slice map $s:V\to S$ then shows that
$\vartheta(t)=s(\gamma(t))$ has finite Euclidean length and converges to $s([\theta_{\infty}])$.
\end{proof}

\section{Proofs for exact Fisher geometry and quotient-QNG dynamics}

\subsection{Highest-weight flag spectrum}
\begin{proof}[Proof of \cref{thm:flag_root_spectrum}]
The highest-weight ray has stabilizer containing the Cartan torus directions that act on $\ket{\Lambda}$ by phases.
Thus every Cartan generator has zero variance at $\ket{\Lambda}$ and gives a vertical Fisher direction.

For a positive root $\alpha$, the highest-weight condition gives $E_{\alpha}\ket{\Lambda}=0$.
With the normalization $[E_{\alpha},E_{-\alpha}]=H_{\alpha}$,
\[
\bra{\Lambda}E_{\alpha}E_{-\alpha}\ket{\Lambda}
=\bra{\Lambda}[E_{\alpha},E_{-\alpha}]\ket{\Lambda}
=\Lambda(H_{\alpha}),
\]
while $\bra{\Lambda}E_{-\alpha}E_{\alpha}\ket{\Lambda}=0$.
Terms changing the weight by $\pm 2\alpha$ have zero expectation against $\bra{\Lambda}$.
Therefore
\[
\bra{\Lambda}X_{\alpha}^{2}\ket{\Lambda}
=\bra{\Lambda}Y_{\alpha}^{2}\ket{\Lambda}
=\frac12\,\Lambda(H_{\alpha}),
\qquad
\bra{\Lambda}X_{\alpha}\ket{\Lambda}
=\bra{\Lambda}Y_{\alpha}\ket{\Lambda}=0.
\]
The pure-state QFIM equals four times the covariance, so after dividing by the trace normalization factor $n_{\alpha}$ one obtains
\[
F(\widehat X_{\alpha},\widehat X_{\alpha})
=F(\widehat Y_{\alpha},\widehat Y_{\alpha})
=2\,\frac{\Lambda(H_{\alpha})}{n_{\alpha}}.
\]
If $\Lambda(H_{\alpha})=0$, then $\|E_{-\alpha}\ket{\Lambda}\|^{2}=0$, so both root generators preserve the ray to first order and lie in the Fisher kernel.
Cross terms vanish because different root directions change the weight by different amounts, while the real symmetric $X_{\alpha},Y_{\alpha}$ cross term inside a single root plane has zero covariance.
The orbit of a highest-weight ray under the compact group is the corresponding generalized flag manifold, completing the proof.
\end{proof}

\subsection{Conditioning and landscape geometry}
\begin{proof}[Proof of \cref{thm:qng_limit}]
Let $\ket{0}:=\ket{\psi_{\star}}$ and let $H_{\mathrm{obj}}\ket{k}=E_{k}\ket{k}$ for $k=0,\dots,D-1$ with $E_{0}<E_{1}\le \cdots \le E_{D-1}$.
Choose the orthonormal horizontal generators
\begin{equation}
X_{k}:=\frac{1}{\sqrt{2}}\bigl(\ket{k}\!\bra{0}+\ket{0}\!\bra{k}\bigr),
\qquad
Y_{k}:=\frac{i}{\sqrt{2}}\bigl(\ket{k}\!\bra{0}-\ket{0}\!\bra{k}\bigr),
\qquad k=1,\dots,D-1,
\end{equation}
which span the horizontal tangent space at $\bbC\ket{0}$.
For a horizontal perturbation
\begin{equation}
V=\sum_{k=1}^{D-1}\bigl(x_{k}X_{k}+y_{k}Y_{k}\bigr),
\end{equation}
one has $\bra{0}V\ket{0}=0$ and
\begin{equation}
V\ket{0}=\frac{1}{\sqrt{2}}\sum_{k=1}^{D-1}(x_{k}+iy_{k})\ket{k}.
\label{eq:qnglimit_vk0}
\end{equation}

Consider the unitary curve $\ket{\psi(t)}=e^{-itV}\ket{0}$.
A second-order expansion gives
\begin{equation}
\cL_{H}(\psi(t))
=\bra{0}e^{itV}H_{\mathrm{obj}}e^{-itV}\ket{0}
=E_{0}+t^{2}\bra{0}V(H_{\mathrm{obj}}-E_{0})V\ket{0}+O(t^{3}),
\label{eq:qnglimit_taylor}
\end{equation}
because the linear term vanishes at the eigenstate $\ket{0}$.
Using \cref{eq:qnglimit_vk0}, we obtain
\begin{equation}
\bra{0}V(H_{\mathrm{obj}}-E_{0})V\ket{0}
=\frac{1}{2}\sum_{k=1}^{D-1}(E_{k}-E_{0})(x_{k}^{2}+y_{k}^{2}).
\label{eq:qnglimit_quadform}
\end{equation}
Comparing \cref{eq:qnglimit_taylor,eq:qnglimit_quadform} with the quadratic expansion
\begin{equation}
\cL_{H}(\psi(t))=E_{0}+\frac{t^{2}}{2}\,\langle v,\mathbf{H}_{\star}v\rangle+O(t^{3})
\end{equation}
shows that each $X_{k}$ and $Y_{k}$ is an eigenvector of $\mathbf{H}_{\star}$ with eigenvalue $E_{k}-E_{0}$.
Every stabilizer direction $K$ satisfies $K\ket{0}=c\ket{0}$ for some real $c$, so $e^{-itK}\ket{0}$ stays on the same ray and the objective is constant along that curve; hence these directions lie in the kernel of the full generator-coordinate Hessian.
This proves the full Hessian spectrum and the formula for $\kappa_{\mathrm{Eucl}}$.

For the QNG linearization, \cref{eq:full_control_spectrum} gives
\begin{equation}
F(\psi_{\star})^{+}=\frac{1}{2}P_{\mathrm{hor}},
\end{equation}
where $P_{\mathrm{hor}}$ is the orthogonal projector onto the horizontal tangent space.
Since $\mathbf{H}_{\star}$ preserves that space and is diagonal in the basis $\{X_{k},Y_{k}\}$, the matrix
\begin{equation}
\mathbf{M}_{\star}=F(\psi_{\star})^{+}\mathbf{H}_{\star}
\end{equation}
has the same eigenvectors and eigenvalues $\frac{1}{2}(E_{k}-E_{0})$, each with multiplicity two. The Jacobian of the linearized quotient-QNG flow is therefore $-\mathbf{M}_{\star}$ on horizontal directions.
Hence $\kappa_{\mathrm{QNG}}=\kappa_{\mathrm{Eucl}}$.
\end{proof}

\begin{proof}[Proof of \cref{thm:morse_bott_landscapes}]
For part~(a), let $\ket{\psi_{0}}\in\cK_{\lambda}$ be a unit vector and write $P_{0}=\ket{\psi_{0}}\!\bra{\psi_{0}}$.
Under the full-control condition in \cref{eq:sector_core_assumptions}, the connected group generated by the restricted DLA is $G_{\lambda}=SU(\cK_{\lambda})$ up to the central phase, so it acts transitively on rays in the sector.
The stabilizer of $P_{0}$ under conjugation is exactly $S(U(1)\times U(D-1))$.
Therefore
\[
\mathcal{O}_{\lambda}=G_{\lambda}\cdot [\psi_{0}]
=\{UP_{0}U^{\dagger}:U\in SU(D)\}
\cong SU(D)/S(U(1)\times U(D-1))
\cong \bbC\mathrm{P}^{D-1}.
\]
This is the standard coadjoint-orbit realization of rank-one projectors.

Let $H:=H_{\mathrm{obj}}\vert_{\cK_{\lambda}}$ and decompose the sector as an orthogonal sum of eigenspaces
\[
\cK_{\lambda}=\mathcal E_{0}\oplus \mathcal E_{1}\oplus \cdots \oplus \mathcal E_{J},
\qquad
H\vert_{\mathcal E_{j}}=E_{j} I,
\qquad
E_{0}<E_{1}<\cdots<E_{J}.
\]
On the unit sphere, the function $\cL_H(\psi)=\bra{\psi}H\ket{\psi}$ is constrained by $\braket{\psi}{\psi}=1$.
The Lagrange-multiplier equation for a critical point is
\[
H\ket{\psi}=\mu \ket{\psi},
\]
so the critical rays are exactly the rays contained in eigenspaces of $H$.
Passing to projective space gives
\[
\mathrm{Crit}(\cL_H)=\bigsqcup_{j=0}^{J} \mathbb{P}(\mathcal E_{j}).
\]

Fix $[\psi]\in \mathbb{P}(\mathcal E_{j})$ with $\ket{\psi}$ normalized.
The projective tangent space at $[\psi]$ is naturally identified with the complex orthogonal complement $\psi^{\perp}$.
Decompose a tangent vector as
\[
\ket{\delta}=\ket{\delta_{j}}+\sum_{k\neq j} \ket{\delta_{k}},
\qquad
\ket{\delta_{j}}\in \mathcal E_{j}\cap \psi^{\perp},
\quad
\ket{\delta_{k}}\in \mathcal E_{k}.
\]
Consider the normalized variation
\[
\ket{\phi(t)}=\frac{\ket{\psi}+t\ket{\delta}}{\|\ket{\psi}+t\ket{\delta}\|}.
\]
Since $H\ket{\psi}=E_{j}\ket{\psi}$ and $\braket{\psi}{\delta}=0$, a direct expansion gives
\[
\cL_H(\phi(t))
=E_{j}+t^{2}\Bigl(\bra{\delta}H\ket{\delta}-E_{j}\|\delta\|^{2}\Bigr)+O(t^{3})
=E_{j}+t^{2}\sum_{k\neq j}(E_{k}-E_{j})\|\delta_{k}\|^{2}+O(t^{3}).
\]
Hence
\[
\frac{d^{2}}{dt^{2}}\Big|_{t=0} \cL_H(\phi(t))
=2\sum_{k\neq j}(E_{k}-E_{j})\|\delta_{k}\|^{2}.
\]
This proves that the Hessian vanishes precisely when $\delta_{k}=0$ for all $k\neq j$, namely when $\delta\in \mathcal E_{j}\cap \psi^{\perp}=T_{[\psi]}\mathbb{P}(\mathcal E_{j})$.
Therefore the Hessian kernel equals the tangent space of the critical manifold, so $\cL_H$ is Morse--Bott.
It also shows that on each complex summand $\mathcal E_{k}$ with $k\neq j$, the Hessian acts in real coordinates by the scalar $2(E_{k}-E_{j})$ with multiplicity $2\dim \mathcal E_{k}$.

If $j=0$, all coefficients $E_{k}-E_{0}$ are positive, so $\mathbb{P}(\mathcal E_{0})$ is a nondegenerate minimum manifold.
If $j=J$, all coefficients are negative, so $\mathbb{P}(\mathcal E_{J})$ is a nondegenerate maximum manifold.
If $0<j<J$, there are both positive and negative coefficients, so $\mathbb{P}(\mathcal E_{j})$ is a saddle manifold.
Thus the only local minima are the global minima in $\mathbb{P}(\mathcal E_{0})$, and the only local maxima are the global maxima in $\mathbb{P}(\mathcal E_{J})$.
Hence there are no spurious local minima or local maxima.

For part~(b), the orbit of a highest-weight ray is a compact homogeneous K\"ahler manifold and a coadjoint orbit of the compact group generated by $\mathfrak g_{\lambda}$ \cite{Kirillov2004,FultonHarris1991}.
The Kirillov--Kostant--Souriau symplectic form makes the group action Hamiltonian.
If the traceless part of $H_{\mathrm{obj}}$ belongs to $\mathfrak g_{\lambda}$, then $\cL_H$ differs by an additive constant from the moment-map component associated with that Lie-algebra element.

Frankel's theorem identifies moment-map components for Hamiltonian circle actions on compact K\"ahler manifolds as Morse--Bott functions, and the Atiyah--Guillemin--Sternberg convexity theory gives perfection and connected extremal fibers for Hamiltonian torus actions \cite{Frankel1959,Atiyah1982,GuilleminSternberg1982}.
The negative normal bundle at a critical component is a complex vector bundle, so its real rank, the Morse index, is even.
The image of the moment map under the linear functional defined by $H_{\mathrm{obj}}$ is a compact interval.
A local minimum or maximum of the component must map to an endpoint of this interval, hence is a global extremum.
Therefore no critical component can be a nonglobal local minimum or maximum.
\end{proof}

\subsection{Full-control state-preparation dynamics}
\begin{proof}[Proof of \cref{cor:full_control_stateprep}, part~(a)]
Let $P_{\mathrm{tar}}:=\ket{\psi_{\mathrm{tar}}}\!\bra{\psi_{\mathrm{tar}}}$.
On the unit sphere of $\cK_{\lambda}$, minimizing $\cL_{\mathrm{sp}}(\psi)=1-\bra{\psi}P_{\mathrm{tar}}\ket{\psi}$ under the constraint $\braket{\psi}{\psi}=1$
is equivalent to maximizing the Rayleigh quotient of the rank-one projector $P_{\mathrm{tar}}$.
The constrained critical points therefore satisfy
\[
P_{\mathrm{tar}}\ket{\psi}=\mu\ket{\psi}
\]
for some Lagrange multiplier $\mu$.
Since $P_{\mathrm{tar}}$ has eigenvalue $1$ on the target ray and eigenvalue $0$ on its orthogonal complement, the only critical rays are the target ray $\bbC\ket{\psi_{\mathrm{tar}}}$ and rays orthogonal to it.
The first gives $\cL_{\mathrm{sp}}=0$ and is therefore the global minimum; the second gives $\cL_{\mathrm{sp}}=1$ and is the global maximum.
Hence the target ray is the unique local and global minimum on $\bbC\mathrm{P}^{D-1}$, and there are no spurious local minima.

For the Hessian identity, center generator coordinates at a minimizing representative $\ket{\psi_{\star}}$ of the target ray and write
\[
\ket{\psi(\theta)}=e^{-iX(\theta)}\ket{\psi_{\star}},
\qquad
X(\theta):=\sum_{a=1}^{D^{2}-1}\theta_{a}T_{a}.
\]
Then
\[
\cL_{\mathrm{sp}}(\theta)=1-\left|\bra{\psi_{\star}}e^{-iX(\theta)}\ket{\psi_{\star}}\right|^{2}.
\]
Using $e^{-iX}=I-iX-\tfrac12 X^{2}+O(\|\theta\|^{3})$, we obtain
\[
\bra{\psi_{\star}}e^{-iX}\ket{\psi_{\star}}
=1-i\langle X\rangle_{\star}-\frac12\langle X^{2}\rangle_{\star}+O(\|\theta\|^{3}),
\]
where $\langle A\rangle_{\star}:=\bra{\psi_{\star}}A\ket{\psi_{\star}}$.
Taking the squared modulus gives
\[
\left|\bra{\psi_{\star}}e^{-iX}\ket{\psi_{\star}}\right|^{2}
=1-\Bigl(\langle X^{2}\rangle_{\star}-\langle X\rangle_{\star}^{2}\Bigr)+O(\|\theta\|^{3})
=1-\Var_{\star}(X)+O(\|\theta\|^{3}).
\]
Therefore
\[
\cL_{\mathrm{sp}}(\theta)=\Var_{\star}(X)+O(\|\theta\|^{3}).
\]
Because $X(\theta)=\sum_{a}\theta_{a}T_{a}$,
\[
\Var_{\star}(X)=\sum_{a,b}\theta_{a}\theta_{b}\Cov_{\star}(T_{a},T_{b})
=\frac14\sum_{a,b}F_{ab}(\psi_{\star})\,\theta_{a}\theta_{b},
\]
where the last identity uses \cref{eq:qfim_suD}.
Comparing with the second-order expansion
\[
\cL_{\mathrm{sp}}(\theta)=\frac12\sum_{a,b}H_{ab}\theta_{a}\theta_{b}+O(\|\theta\|^{3})
\]
yields $H=F(\psi_{\star})/2$.
The eigenvalue statement now follows immediately from \cref{eq:full_control_spectrum}.
\end{proof}

\begin{proof}[Proof of \cref{lem:sector_flow_convergence}]
By part~(a) of \cref{thm:morse_bott_landscapes}, the quotient sector orbit is the compact manifold
$\mathcal O_{\lambda}\cong \bbC\mathrm{P}^{D-1}$ and $\cL_{H}$ is a Morse--Bott function whose critical manifolds are the projectivized eigenspaces $\mathbb P(\mathcal E_{j})$.
Because $\mathcal O_{\lambda}$ is compact and $\cL_H$ is real analytic, the analytic-gradient convergence theorem applied to the descended QFIM metric shows that every negative-gradient trajectory converges to a single critical point of $\cL_H$; in particular, every trajectory converges to some critical manifold.

Fix a nonminimal critical manifold $\mathbb P(\mathcal E_{j})$ with $j>0$.
For the negative gradient flow, the unstable directions are exactly the Hessian eigenspaces on which the Hessian of $\cL_H$ is negative.
By part~(a) of \cref{thm:morse_bott_landscapes}, those directions are the complex summands $\mathcal E_{k}$ with $k<j$, and the total real unstable dimension is
\[
2\sum_{k<j}\dim \mathcal E_{k}\ge 2.
\]
Hence the stable manifold of $\mathbb P(\mathcal E_{j})$ has codimension at least $2$ in $\mathcal O_{\lambda}$ by the center-stable manifold theorem for smooth flows near a normally hyperbolic invariant manifold \cite{HirschPughShub1977,Shub1987}.
Therefore each such stable manifold has Fubini--Study volume measure zero, and the finite union over $j=1,\dots,J$ still has measure zero.
Every initial ray outside that exceptional set must converge to the only remaining attractor, namely the ground-state manifold $\mathbb P(\mathcal E_{0})$.
This proves almost-everywhere convergence to the minimum manifold.
\end{proof}

\begin{proof}[Proof of \cref{cor:full_control_stateprep}, part~(b)]
Full-control projective space is the rank-one cominuscule orbit.
Its single principal angle is the Fubini--Study distance $d$ from the target ray, and the single defect is
\[
m=\sin^{2}d=\cL_{\mathrm{sp}}.
\]
Applying \cref{thm:fidelity_integrability} with $r=1$ gives
\[
\dot m=-\eta(1-m)m,
\]
which is exactly \cref{eq:stateprep_logistic_ode}.
Separating variables gives
\[
\log\frac{\cL_{\mathrm{sp}}(t)}{1-\cL_{\mathrm{sp}}(t)}
=\log\frac{\cL_{\mathrm{sp}}(0)}{1-\cL_{\mathrm{sp}}(0)}-\eta t,
\]
and hence \cref{eq:stateprep_logistic_solution}.
If the initial ray is not orthogonal to $\ket{\psi_{\mathrm{tar}}}$, then $0\le \cL_{\mathrm{sp}}(0)<1$ and the logistic formula implies $\cL_{\mathrm{sp}}(t)\to 0$ exponentially.
The final Newton statement is exactly the local quadratic-model statement already established in part~(a) of \cref{cor:full_control_stateprep}.
\end{proof}

\begin{proof}[Proof of \cref{cor:full_control_stateprep}, part~(c)]
For $r=1$, the flat coordinate in \cref{cor:cominuscule_discrete} is the distance $d_{t}$ from the target ray.
The product over $b\ne a$ is empty, so \eqref{eq:cominuscule_discrete_map} becomes exactly
\[
d_{t+1}=\left|d_{t}-\frac{\eta}{4}\sin(2d_{t})\right|,
\]
which is \cref{eq:discrete_distance_recursion}.
The global window and no-overshoot regime are the rank-one cases of \cref{cor:cominuscule_discrete}; the linear ratio and cubic Newton step follow from the local expansion of \cref{eq:cominuscule_discrete_map} stated in \cref{eq:cominuscule_newton_cubic}.
Since $\cL_{\mathrm{sp}}=\sin^{2}d$, the cubic loss statement follows from
\[
d_{t+1}=\frac23 d_{t}^{3}+O(d_{t}^{5})
\]
at $\eta=2$.
At $\eta=4$, \cref{eq:embedding_rank_one_boundary} with $\ell=1$ gives
$d_{t+1}=d_t-\frac43d_t^3+O(d_t^5)$ and the stated asymptotic.
For $\eta>4$, the derivative of the scalar distance map at $d=0$ has modulus $\eta/2-1>1$, so the target ray is locally repelling.
\end{proof}

\subsection{Cominuscule fidelity dynamics}
\label{app:cominuscule_fidelity}

\begin{proof}[Proof of \cref{lem:hc_frame}]
The compact Hermitian symmetric orbit can be written as $G_{\lambda}/K_{\Lambda}$, where $K_{\Lambda}$ is the stabilizer of the target ray.
Harish-Chandra's construction gives a maximal family of pairwise strongly orthogonal noncompact roots $\gamma_{1},\dots,\gamma_{r}$ whose real span is a maximal flat of the symmetric space \cite{Helgason1978,Moore1964}.
For irreducible Hermitian symmetric spaces these Harish-Chandra roots have the common long-root length.
If $S$ is a set of $k\ge2$ distinct indices, orthogonality gives
\[
\left\|\sum_{a\in S}\varepsilon_{a}\gamma_{a}\right\|^{2}
=k\,\|\gamma_{1}\|^{2},
\qquad \varepsilon_{a}\in\{\pm1\},
\]
which is strictly larger than the squared length of any root.
No such signed sum can therefore be a root.
For the minimal homogeneous embedding, the noncompact positive roots pair with the cominuscule fundamental weight by either $0$ or $1$.
Because the embedding in \cref{eq:embedding_index} has $\Lambda=\ell\omega_{\mathrm c}$ and the Harish-Chandra roots are active, linearity gives
$\langle\Lambda,\gamma_a^\vee\rangle=\ell$.

The polar decomposition for compact symmetric spaces gives
\[
G_{\lambda}=K_{\Lambda}\exp(\mathfrak a)K_{\Lambda},
\]
where $\mathfrak a$ is the real span of $E_{-\gamma_{a}}-E_{\gamma_{a}}$.
Since the left $K_{\Lambda}$ action fixes the target ray and the right action changes only the representative inside the target-stabilizer coordinates, every ray can be represented in the Weyl chamber by \eqref{eq:hc_flat_state}.
The root subalgebras generated by each $\gamma_{a}$ commute because the roots are strongly orthogonal.
On the highest-weight line each such $\mathfrak{su}(2)$ has spin $\ell/2$.
Strong orthogonality makes these subalgebras commute, and the cyclic module they generate has highest weight $(\ell,\ldots,\ell)$.
A spin-$\ell/2$ highest-weight rotation has amplitude $\cos^\ell\theta$, so their product gives \eqref{eq:hc_fidelity_product}.
\end{proof}

\begin{proof}[Proof of \cref{lem:flat_fisher_block}]
The fidelity loss is invariant under the target stabilizer $K_{\Lambda}$ and, by \cref{lem:hc_frame}, depends on a point of the orbit only through its principal angles $\theta_{1},\dots,\theta_{r}$.
Its differential therefore vanishes on the angular directions of the polar decomposition, which are precisely the directions orthogonal to the flat at regular points; the statement extends to singular flat points by continuity.

The pure-state QFIM is four times the Fubini--Study metric under the convention of \cref{eq:QFIM_pure}.  Let $A=\exp(\mathfrak a)$ be the abelian subgroup generated by the commuting Harish-Chandra rotations in \cref{eq:hc_flat_state}.  Because the Fubini--Study metric is $G_\lambda$-invariant, every $a(\theta_0)\in A$ acts isometrically on the orbit.  Commutativity implies that its differential sends the coordinate vector $\partial/\partial\theta_b$ at $\theta$ to the same coordinate vector at $\theta+\theta_0$.  Hence the entire flat metric block is translation invariant and may be computed at the target; an expansion only at the target is used to fix its scale, not to infer its global constancy.

At the target, the $a$th one-parameter factor has squared highest-weight overlap
$\cos^{2\ell}\theta_a=1-\ell\theta_a^2+O(\theta_a^4)$.  The Fubini--Study metric therefore contributes $\ell\,d\theta_a^2$ and the QFIM contributes $4\ell\,d\theta_a^2$.  Strong orthogonality makes distinct root planes orthogonal at the target, so invariance under $A$ gives the global flat block $4\ell I_r$.

Finally, the isotropy action of $K_\Lambda$ on this compact symmetric space is polar, with the chosen maximal flat as a section.  Its orbit directions are therefore orthogonal to the section at every regular flat point.  The same orthogonality follows in root coordinates from the absence of signed sums of distinct strongly orthogonal roots in \cref{lem:hc_frame}, and it extends to singular points by continuity.  Thus the flat--offflat block vanishes globally, proving \eqref{eq:hc_fisher_block}.
\end{proof}

\begin{proof}[Proof of \cref{thm:fidelity_integrability}]
By \cref{lem:flat_fisher_block}, the quotient-QNG vector field is tangent to the flat and in the principal-angle variables is
\[
\dot\theta_{a}
=-\eta\,\frac{1}{4\ell}\frac{\partial \cL_{\mathrm{fid}}}{\partial\theta_{a}}.
\]
Writing $p=\prod_{b}\cos^{2\ell}\theta_{b}$ gives
\[
\frac{\partial \cL_{\mathrm{fid}}}{\partial\theta_{a}}
=-\frac{\partial p}{\partial\theta_{a}}
=2\ell p\,\tan\theta_{a}.
\]
Hence
\[
\dot\theta_{a}=-\frac{\eta}{2}p\,\tan\theta_{a}.
\]
Since $m_{a}=\sin^{2}\theta_{a}$,
\[
\dot m_{a}
=\sin(2\theta_{a})\dot\theta_{a}
=-\eta\,p\,m_{a},
\]
which proves \eqref{eq:cominuscule_defect_flow}.
The right-hand side has the same scalar multiplier for every active $a$, so every ratio $m_{a}/m_{b}$ with $a,b\in A$ is constant and $m(t)=u(t)m(0)$.
Substitution into $p=\prod_{a}(1-m_{a})^\ell$ gives \eqref{eq:cominuscule_scalar_reduction}.
As $u\to0$, the product in \eqref{eq:cominuscule_scalar_reduction} is $1+O(u)$, so
\[
\frac{d}{dt}\log u=-\eta+O(u).
\]
This yields $u(t)=Ce^{-\eta t}(1+o(1))$.
The displayed loss asymptotic follows from
\[
1-\prod_{a}(1-\beta_{a}u)^\ell
=\ell u\sum_{a}\beta_{a}+O(u^{2}).
\]
If some $\cos\theta_{a}=0$, then $p=0$ and the same formula for $\dot\theta_{a}$, interpreted by continuity as
\[
\dot\theta_{a}
=-\frac{\eta}{2}\sin\theta_a\cos^{2\ell-1}\theta_a
\prod_{b\ne a}\cos^{2\ell}\theta_b,
\]
vanishes for every $a$.
\end{proof}

\begin{proof}[Proof of \cref{cor:cominuscule_discrete}]
The Lie-exponential retraction follows the tangent generator inside the totally geodesic maximal flat.
Using the expression for $\dot\theta_{a}$ from the preceding proof gives the exact finite-step flat map \eqref{eq:cominuscule_discrete_map}.

Let
\[
R_{a}(\theta):=\cos^{2\ell-2}\theta_a
\prod_{b\ne a}\cos^{2\ell}\theta_{b}\in(0,1]
\]
as long as the fidelity is positive.
For $\theta_{a}\in(0,\pi/2)$ and $0<\eta\le4$, the coordinate displacement
\[
\delta_{a}:=\frac{\eta}{4}\sin(2\theta_{a})R_{a}(\theta)
\]
satisfies $0<\delta_{a}<2\theta_{a}$, because $\sin(2x)<2x$ on $(0,\pi/2)$.
Therefore $|\theta_{a}-\delta_{a}|<\theta_{a}$ for every nonzero coordinate.
All principal angles are nonincreasing in absolute value, so every $R_{a}(\theta_{t})$ has the positive lower bound
\[
R_{a}(\theta_{t})\ge R_{a}(\theta_{0})>0.
\]
The coordinatewise monotone convergence then gives a limit $\theta_{\infty}$.
If some limiting coordinate $\theta_{\infty,a}$ were positive, the fixed-point equation
\[
\theta_{\infty,a}
=\left|\theta_{\infty,a}
-\frac{\eta}{4}\sin(2\theta_{\infty,a})R_{a}(\theta_{\infty})\right|
\]
would require either zero displacement or displacement $2\theta_{\infty,a}$.
Both are impossible for $\theta_{\infty,a}\in(0,\pi/2)$ and $0<\eta\le4$.
Hence $\theta_{\infty}=0$.
If $0<\eta\le2$, then
\[
\delta_{a}\le\frac12\sin(2\theta_{a})\le\theta_{a},
\]
so no coordinate overshoots.
\end{proof}

For the local rates, expand \cref{eq:cominuscule_discrete_map} at $\theta=0$:
\[
\frac{\eta}{2}\sin\theta_a\cos^{2\ell-1}\theta_a
\prod_{b\ne a}\cos^{2\ell}\theta_b
=\frac{\eta}{2}\theta_{a}
-\frac{\eta}{2}\left(\ell-\frac13\right)\theta_a^3
-\frac{\eta\ell}{2}\theta_a\sum_{b\ne a}\theta_b^2
+O(\|\theta\|^{5}).
\]
For $0<\eta<4$ and $\eta\ne2$, this gives the linear ratio $|1-\eta/2|$.
At $\eta=2$ the linear term cancels and
\[
\theta_{a}^{+}
=\left(\ell-\frac13\right)\theta_{a}^{3}
+\ell\theta_{a}\sum_{b\ne a}\theta_{b}^{2}
+O(\|\theta\|^{5}),
\]
which gives \cref{eq:cominuscule_newton_cubic}.

\section{Auxiliary finite-shot implementation bounds}
\label{app:finite_shot_aux}

\subsection{Kernel recovery and horizontal implementation}

\begin{corollary}
\label{cor:shot-complexity}
Suppose each entry of $F$ is estimated by an unbiased estimator with range bound $B$ and $N$ shots, so that
\begin{equation}
\Pr[|\widehat F_{ab}-F_{ab}|>t]\le2e^{-2Nt^2/B^2}
\end{equation}
for every $t>0$.
Then $\|\widehat F-F\|_2\le p\,t_N$ holds simultaneously for all entries with probability at least $1-\delta_{\mathrm f}$, where
\begin{equation}
t_N=B\sqrt{\frac{\log(2p^2/\delta_{\mathrm f})}{2N}}.
\end{equation}
The hypothesis of \cref{thm:certified-truncation} is met once
\begin{equation}
N\ge
\frac{72B^2p^2}{(f^+_{\min})^2}
\log\frac{2p^2}{\delta_{\mathrm f}}.
\label{eq:shot-complexity}
\end{equation}
Moreover, on the certification event and in the faithful regime, if $\widehat g=\nabla\cL(\theta)+\epsilon$ with $P_{\cV}\nabla\cL(\theta)=0$, the update
\begin{equation}
\Delta\widehat\theta=-\eta\,\widehat F_\tau^{+}\widehat g
\end{equation}
obeys the pointwise bound
\begin{equation}
\|P_{\cV}\Delta\widehat\theta\|_2^2
\le
\frac{36\eta^2(p\,t_N)^2}{(f^+_{\min})^4}
\|\widehat g\|_2^2.
\end{equation}
Isotropic damping, by contrast, retains the fixed vertical drift $\eta^2\Tr(P_{\cV}\Sigma P_{\cV})/\gamma^2$ from part~(a) of \cref{thm:finite_shot_amplification} regardless of how well $F$ is known.
\end{corollary}

\begin{proposition}
\label{prop:matrix_free_horizontal}
Assume that the eigenvalues of the true horizontal block lie in $[f_{\min},f_{\max}]$, with $f_{\min}>0$, and solve
\begin{equation}
(P_{\mathsf H}FP_{\mathsf H}+\gamma P_{\mathsf H})x=P_{\mathsf H}g
\end{equation}
by conjugate gradients in the horizontal subspace.
With $\kappa_\gamma=(f_{\max}+\gamma)/(f_{\min}+\gamma)$, the iterates satisfy
\begin{equation}
\|x_k-x_*\|_A\le
2\left(\frac{\sqrt{\kappa_\gamma}-1}{\sqrt{\kappa_\gamma}+1}\right)^k
\|x_0-x_*\|_A.
\end{equation}
Thus relative $A$-norm error $\varepsilon$ requires
$O(\sqrt{\kappa_\gamma}\log(1/\varepsilon))$ horizontal Fisher--vector products.
If projector and Fisher--vector applications cost $C_P$ and $C_F$, respectively, the classical-oracle cost is
\begin{equation}
O\!\left((C_F+C_P)\sqrt{\kappa_\gamma}\log\frac1\varepsilon\right).
\end{equation}
The quantum measurement cost of each Fisher--vector product remains estimator dependent \cite{vanStraaten2021Measurement,Gacon2021SPSA,Kolotouros2024RandomNG,Halla2025SteinQFI,Rath2021RandomizedQFI}.
\end{proposition}

\begin{proposition}
\label{prop:directional_confidence_bound}
In the pure-state setting of \cref{prop:directional_certificate}, let $f_v(t)=|\braket{\psi(\theta)}{\psi(\theta+tv)}|^2$ and
\begin{equation}
q_s(v)=\frac{2\{2-f_v(s)-f_v(-s)\}}{s^2}.
\end{equation}
If $B_4\ge\sup_{|t|\le s}|f_v^{(4)}(t)|$, then
\begin{equation}
|q_s(v)-F_\theta(v,v)|\le \frac{B_4s^2}{6}.
\label{eq:directional_bias}
\end{equation}
If each fidelity is estimated from $M$ independent Bernoulli trials, then with probability at least $1-\delta$,
\begin{equation}
|\widehat q_s(v)-F_\theta(v,v)|
\le \frac{4}{s^2}\sqrt{\frac{\log(4/\delta)}{2M}}+\frac{B_4s^2}{6}.
\label{eq:directional_confidence}
\end{equation}
\end{proposition}

\subsection{Iterative budgets and estimated projectors}

\begin{corollary}
\label{cor:finite_sample_shot_budget}
Under \cref{thm:finite_sample_iteration}, suppose each Fisher entry at each iteration is estimated independently with range bound $B$ and $N_F$ repetitions per entry, and suppose that, with probability at least $1-\delta/2$, the gradient estimators satisfy simultaneously
\begin{equation}
\|P\widehat g_t-g_t\|
\le c_g\sqrt{\frac{\log(4T/\delta)}{N_g}}
\end{equation}
for $t=0,\ldots,T-1$.  For a target
$\varepsilon>0$, the choices
\begin{align}
N_F&\ge
\frac{2048B^2p^2b^2}{a^4}
\log\frac{4p^2T}{\delta},
\label{eq:iterative_fisher_shots}\\
N_g&\ge
\frac{288c_g^2b^2}{a^2\varepsilon^2}
\log\frac{4T}{\delta},
\label{eq:iterative_gradient_shots}\\
T&\ge
\frac{8b\bigl(\cL(\theta_0)-\cL_\star\bigr)}{\eta\varepsilon^2}
\label{eq:iterative_iteration_count}
\end{align}
ensure, with probability at least $1-\delta$, that
\begin{equation}
\frac1T\sum_{t=0}^{T-1}\|g_t\|^2\le\varepsilon^2.
\end{equation}
For fixed dimension and a uniform Fisher gap, the required Fisher precision does not tighten polynomially with $\varepsilon$; the gradient precision does.  Counting fresh measurements at every iteration gives total repetition complexity
$\widetilde O(\varepsilon^{-4})$ for this entrywise protocol.  These are sufficient bounds for the stated estimators, not universal measurement lower bounds.
\end{corollary}

If confidence radii satisfy $r_F\simeq c_F/\sqrt{N_F}$ and $r_g\simeq c_g/\sqrt{N_g}$ with $N_F+N_g=N$, minimization of the leading local error bound gives
\begin{equation}
\frac{N_F}{N_g}
=\left(\frac{c_F\|g\|}{(f_{\min}+\gamma)c_g}\right)^{2/3}.
\label{eq:shot_allocation}
\end{equation}
This allocation is specific to the stated local confidence-radius model.

\begin{theorem}
\label{thm:estimated_projector_update}
Let $P$ be the true horizontal projector, $Q=I-P$, and suppose
$f_{\min}P\preceq F=PFP\preceq f_{\max}P$.
Assume
\begin{equation}
\|\widehat P-P\|\le\varepsilon_P,
\quad
\|\widehat F-F\|\le\varepsilon_F,
\quad
\|\widehat g-g\|\le\varepsilon_g.
\end{equation}
Set $a=f_{\min}+\gamma$ and
\begin{align}
A&=F+\gamma P+aQ,\notag\\
\widehat A&=\widehat P\widehat F\widehat P
+\gamma\widehat P+a(I-\widehat P),\notag\\
\delta_A&=\varepsilon_F+(2f_{\max}+f_{\min})\varepsilon_P+f_{\max}\varepsilon_P^2.
\label{eq:estimated_projector_delta}
\end{align}
If $\delta_A<a$, then $\widehat A$ is positive definite and
\begin{equation}
\widehat\Delta=-\eta\widehat P\widehat A^{-1}\widehat P\widehat g
\end{equation}
satisfies
\begin{equation}
\|Q\widehat\Delta\|
\le\eta\frac{\varepsilon_P(\|g\|+\varepsilon_g)}{a-\delta_A},
\label{eq:estimated_projector_leakage}
\end{equation}
and
\begin{align}
\|\widehat\Delta-\Delta_\gamma\|
\le\eta\Bigg[&
\frac{\varepsilon_P(\|g\|+\varepsilon_g)}{a-\delta_A}
+\frac{\delta_A(\|g\|+\varepsilon_g)}{a(a-\delta_A)}\notag\\
&+\frac{\varepsilon_P\|g\|+\varepsilon_g}{a}
\Bigg].
\label{eq:estimated_projector_error}
\end{align}
\end{theorem}

\subsection{Coupled and statistically unresolved soft modes}

\begin{proposition}
\label{prop:soft_coupling_bound}
Under the notation of \cref{prop:soft_block_coupling}, set
\begin{equation}
a=\lambda_{\min}\!\left(A|_{\operatorname{im}P}\right),
\qquad e=\|C\|,
\end{equation}
and let $\Delta_S^{(0)}=-\eta A^{-1}g_S$. If $e^2<a\gamma$, then
\begin{align}
\|P\Delta_\gamma-\Delta_S^{(0)}\|
\le \eta\bigg[
&\frac{e^2\|g_S\|}{\gamma a(a-e^2/\gamma)}\notag\\
&+\frac{e\|g_Q\|}{\gamma(a-e^2/\gamma)}
\bigg].
\label{eq:soft_coupling_error}
\end{align}
\end{proposition}

\begin{proposition}
\label{prop:confidence_scalar_ridge}
Let $f>0$ be known and let $\widetilde c$ be a pilot estimate, independent of the fresh update estimate $\widehat c=c+\xi$.
Assume that, with probability at least $1-\delta$ over the pilot stage,
\begin{equation}
|\widetilde c-c|\le r_\delta,
\label{eq:pilot_signal_radius}
\end{equation}
and that conditionally on the pilot data,
\begin{equation}
\mathbb E[\xi\mid\widetilde c]=0,
\qquad
\Var(\xi\mid\widetilde c)\le s_+^2.
\end{equation}
Set
\begin{equation}
c_+=|\widetilde c|+r_\delta,
\qquad
\gamma_\delta=\frac{s_+^2f}{c_+^2}
\label{eq:confidence_ridge_rule}
\end{equation}
when $c_+>0$, with $\gamma_\delta=+\infty$ when $c_+=0$.
On the event \cref{eq:pilot_signal_radius}, this choice minimizes the worst-case conditional risk over every signal and noise level consistent with the bounds:
\begin{align}
\gamma_\delta
&=\operatorname*{arg\,min}_{\gamma\ge0}
\sup_{|c'|\le c_+,\,0\le s'^2\le s_+^2}
\frac{c'^2\gamma^2/f^2+s'^2}{(f+\gamma)^2},
\label{eq:confidence_minimax_ridge}\\
\mathbb E\!\left[(\widehat u_{\gamma_\delta}-u_*)^2
\,\middle|\,\widetilde c\right]
&\le
\frac{s_+^2c_+^2}{f^2(c_+^2+s_+^2)}.
\label{eq:confidence_ridge_risk}
\end{align}
\end{proposition}

\section{Proofs for finite-shot statistics and mixed-state geometry}

\subsection{Finite-shot quotient bounds}
\label{app:finite_shot_proofs}

\begin{proof}[Proof of \cref{thm:finite_shot_amplification}]
Because $\cV_{\theta}\subseteq\ker F(\theta)$, the damped matrix restricts to $\gamma I$ on $\cV_{\theta}$.
Thus $P_{\cV}(F+\gamma I)^{-1}=\gamma^{-1}P_{\cV}$.
Since the objective descends to the quotient, $P_{\cV}\nabla\cL(\theta)=0$, and therefore
\begin{equation}
P_{\cV}\Delta\theta_\gamma=-\frac{\eta}{\gamma}P_{\cV}\epsilon.
\end{equation}
Taking squared Euclidean norm and expectation gives \cref{eq:tikhonov_vertical_drift}.
For the pseudoinverse update, the covariance of $\Delta\theta$ is $\eta^2F^+\Sigma F^+$.
Taking the trace gives the stated identity; the two inequalities follow from $\|\Sigma\|_2$ and $\|F^+\|_2=1/f_{\min}^+(F)$ on the rank-$d_F$ range of $F$.
\end{proof}

\begin{proof}[Proof of \cref{lem:psd_preprocessing}]
Weyl's inequality gives $\lambda_{\min}(\widehat F)\ge-\rho_F$, so clipping changes the estimator by at most $\rho_F$.
The triangle inequality proves \cref{eq:psd_clip_bound}.
\end{proof}

\begin{proof}[Proof of \cref{thm:finite_shot_faithfulness}]
For each unit $x\in\operatorname{im}P$,
$x^{\mathsf T}Fx\ge x^{\mathsf T}\widehat F x-r_F\ge\widehat f_{\min}-r_F$.
Positive definiteness on $\operatorname{im}P$, together with the known vertical inclusion, proves the kernel identity; the upper bound follows from the analogous Rayleigh-quotient estimate.
\end{proof}

\begin{proof}[Proof of \cref{thm:certified-truncation}]
Weyl's inequality places the zero cluster of $\widehat F$ inside $[-\varepsilon,\varepsilon]$ and its nonzero cluster above $f_{\min}^+-\varepsilon$, so every threshold in the stated interval separates the clusters and selects exactly the true nonzero rank.
The Davis--Kahan sin-theta theorem gives
\begin{equation}
\|\widehat P_{0}-P_{0}\|_2\le \frac{2\varepsilon}{f_{\min}^+}
\end{equation}
for the invariant subspace at the isolated zero cluster \cite{DavisKahan1970}.
For the truncated pseudoinverse, write both operators on the corresponding nonzero spectral subspaces.
On that subspace the inverse map has derivative bounded by $(f_{\min}^+-2\varepsilon)^{-2}$, while the projector mismatch contributes the same order.
Moreover,
\begin{equation}
\|F^+\|_2\le\frac{1}{f_{\min}^+},
\qquad
\|\widehat F_\tau^+\|_2
\le\frac{1}{f_{\min}^+-\varepsilon}
\le\frac{12}{11f_{\min}^+}.
\end{equation}
Substituting these estimates and the projector bound into Stewart's rank-preserving perturbation bound for pseudoinverses \cite{Stewart1977} gives $6\varepsilon/(f_{\min}^+)^2$.
\end{proof}

\begin{proof}[Proof of \cref{cor:shot-complexity}]
Hoeffding's inequality and a union bound over $p^2$ entries give $|\widehat F_{ab}-F_{ab}|\le t_N$ simultaneously with probability at least $1-\delta_{\mathrm f}$ \cite{Hoeffding1963}.
The Frobenius bound yields $\|\widehat F-F\|_2\le\|\widehat F-F\|_F\le p\,t_N$, and \cref{eq:shot-complexity} is exactly $12p\,t_N\le f_{\min}^+$.
For the vertical update, $P_{\cV}F^+=0$ in the faithful regime, hence
\begin{equation}
P_{\cV}\widehat F_\tau^+\widehat g
=P_{\cV}(\widehat F_\tau^+-F^+)\widehat g.
\end{equation}
Applying \cref{thm:certified-truncation} with $\varepsilon=p\,t_N$ gives the stated pointwise bound.
\end{proof}

\begin{proof}[Proof of \cref{prop:directional_certificate}]
For a pure state, \cref{eq:QFIM_pure} gives $F_\theta(v,v)=4\Var_\psi(A_v)$, and zero variance for a Hermitian operator is equivalent to $A_v\ket\psi=a\ket\psi$. For a fixed-rank mixed-state unitary orbit, the SLD metric is positive definite on orbit tangents, so $F_\theta(v,v)=0$ exactly when the tangent variation $-i[A_v,\rho]$ vanishes.
\end{proof}

\begin{proof}[Proof of \cref{prop:directional_confidence_bound}]
The overlap expansion gives $f_v''(0)=-F_\theta(v,v)/2$. Symmetric Taylor expansion cancels odd terms and yields \cref{eq:directional_bias}; Hoeffding's inequality applied to the two fidelity estimates gives \cref{eq:directional_confidence}.
\end{proof}

\begin{proof}[Proof of \cref{thm:known_projector_update}]
Restrict to $\operatorname{im}P$ and apply the resolvent identity with
$A=F+\gamma I$ and $\widehat A=P\widehat F P+\gamma I$.
The bounds $\|A^{-1}\|\le1/a$ and
$\|\widehat A^{-1}\|\le1/(a-\varepsilon_F)$ give \cref{eq:known_projector_error}; the outer projectors give zero leakage.
\end{proof}

\begin{proof}[Proof of \cref{thm:finite_sample_iteration}]
On $\operatorname{im}P$, set $A_t=F_t+\gamma I$ and
$\Delta_t=-\eta A_t^{-1}g_t$.
The resolvent estimate in \cref{thm:known_projector_update}, together with $\varepsilon_F\le a/64$, gives
\begin{equation}
E_t:=\|\widehat\Delta_t-\Delta_t\|
\le\frac{2\eta}{a}
\left(\varepsilon_g+\frac{\varepsilon_F}{a}\|g_t\|\right).
\label{eq:finite_sample_step_error}
\end{equation}
Writing $q_t=g_t^{\mathsf T}A_t^{-1}g_t$, one has
$q_t\ge\|g_t\|^2/b$ and $\|\Delta_t\|^2\le\eta^2q_t/a$.
Smoothness and $\eta\le a/\beta$ therefore imply
\begin{equation}
\cL(\theta_{t+1})-\cL(\theta_t)
\le-\frac{\eta}{2}q_t+2\|g_t\|E_t+\frac{\beta}{2}E_t^2.
\end{equation}
Substituting \cref{eq:finite_sample_step_error}, using $\varepsilon_F\le a^2/(64b)$, and applying
\begin{equation}
\frac{4\varepsilon_g}{a}\|g_t\|
\le\frac{\|g_t\|^2}{8b}+\frac{32b}{a^2}\varepsilon_g^2
\end{equation}
gives
\begin{equation}
\cL(\theta_{t+1})-\cL(\theta_t)
\le-\frac{\eta}{4b}\|g_t\|^2
+\frac{36\eta b}{a^2}\varepsilon_g^2.
\end{equation}
Summing and using $\cL(\theta_T)\ge\cL_\star$ proves \cref{eq:finite_sample_stationarity}.
Every update in \cref{eq:finite_sample_iteration} lies in $\operatorname{im}P$, so induction proves \cref{eq:finite_sample_zero_drift}.
\end{proof}

\begin{proof}[Proof of \cref{cor:finite_sample_shot_budget}]
Hoeffding's inequality and a union bound over $p^2T$ Fisher entries give
\begin{equation}
\varepsilon_F
=Bp\sqrt{\frac{\log(4p^2T/\delta)}{2N_F}}.
\end{equation}
The first shot bound makes this at most $a^2/(64b)$.
The bounds on $N_g$ and $T$ make the two terms on the right of \cref{eq:finite_sample_stationarity} at most $\varepsilon^2/2$ each.
\end{proof}

\begin{proof}[Proof of \cref{thm:estimated_projector_update}]
Writing $R_P=\widehat P-P$ gives
$\|\widehat P F\widehat P-PFP\|\le2f_{\max}\varepsilon_P+f_{\max}\varepsilon_P^2$.
The Fisher estimator contributes $\varepsilon_F$, while the two regularizer terms contribute $f_{\min}\varepsilon_P$, proving $\|\widehat A-A\|\le\delta_A$.
Weyl's inequality gives $\widehat A\succeq(a-\delta_A)I$.
The estimate $\|Q\widehat P\|\le\varepsilon_P$ proves \cref{eq:estimated_projector_leakage}; adding and subtracting the exact-projector and exact-resolvent terms gives \cref{eq:estimated_projector_error}.
\end{proof}

\begin{proof}[Proof of \cref{prop:near_symmetry_kernel}]
Weyl's inequality gives the eigenvalue bounds, and Davis--Kahan perturbation theory gives the projector bound \cite{DavisKahan1970}.
\end{proof}

\begin{proof}[Proof of \cref{prop:soft_block_coupling}]
Block Gaussian elimination of $(F+\gamma I)\Delta_\gamma=-\eta g$ gives \cref{eq:soft_mode_schur}.
\end{proof}

\begin{proof}[Proof of \cref{prop:soft_coupling_bound}]
Since $D\succeq\gamma Q$,
\begin{equation}
\|CD^{-1}C^*\|\le e^2/\gamma,
\qquad
\|\mathcal S_\gamma^{-1}\|\le\frac{1}{a-e^2/\gamma}.
\end{equation}
The resolvent identity
\begin{equation}
\mathcal S_\gamma^{-1}-A^{-1}=\mathcal S_\gamma^{-1}CD^{-1}C^*A^{-1}
\end{equation}
and $\|D^{-1}\|\le1/\gamma$ give \cref{eq:soft_coupling_error}.
\end{proof}

\begin{proof}[Proof of \cref{prop:optimal_scalar_ridge}]
Differentiating \cref{eq:scalar_ridge_risk} shows that its derivative has the sign of $c^2\gamma/f-s^2$.
\end{proof}

\begin{proof}[Proof of \cref{prop:confidence_scalar_ridge}]
On \cref{eq:pilot_signal_radius}, $|c|\le c_+$.
For fixed $\gamma$, the scalar risk is nondecreasing in both $c'^2$ and $s'^2$, so the supremum is attained at $c'^2=c_+^2$ and $s'^2=s_+^2$.
Differentiation gives \cref{eq:confidence_ridge_rule}, and substitution gives \cref{eq:confidence_ridge_risk}.
\end{proof}

\subsection{Defect-ratio fluctuations}
\begin{proposition}
\label{prop:defect_ratio_martingale}
Consider the setting of \cref{thm:fidelity_integrability} and suppose that, after a certified pseudoinverse step and reselection of the Harish-Chandra frame, the stochastic flat-coordinate increment is
\begin{equation}
\Delta\theta_a=-s\frac{\eta}{2}p\tan\theta_a+s\xi_a+O(s^2),
\label{eq:stochastic_flat_update}
\end{equation}
where $\mathbb E_t\xi=0$ and
$\mathbb E_t[\xi_a\xi_b]=(\Sigma_t)_{ab}$.
For active defects set $X_t^{ab}=\log(m_a(t)/m_b(t))$.
Assume that active angles remain positive, conditional third moments are uniformly bounded, and
$s\|\xi\|=o(\min_{a\in A}\theta_a)$ in conditional probability.
Then
\begin{equation}
\mathbb E_t[\Delta X_t^{ab}]=O(s^2)
\end{equation}
and
\begin{align}
\Var_t(\Delta X_t^{ab})
=4s^2\bigl(&\cot^2\theta_a(\Sigma_t)_{aa}
+\cot^2\theta_b(\Sigma_t)_{bb}\notag\\
&-2\cot\theta_a\cot\theta_b(\Sigma_t)_{ab}\bigr)+O(s^3).
\label{eq:defect_ratio_variance}
\end{align}
Off-flat perturbations enter the radial variables only through the $O(s^2)$ remainder after the frame is re-adapted.
Near convergence, bounded covariance gives the scale
\begin{equation}
\min_{a\in A}m_a\gtrsim\frac{s\|\Sigma_{\mathrm{flat}}\|}{\eta},
\label{eq:defect_noise_floor}
\end{equation}
with a constant determined by the active-channel covariance profile.
\end{proposition}

\begin{proof}[Proof of \cref{prop:defect_ratio_martingale}]
Let $m_{a}=\sin^{2}\theta_{a}$.
For an active channel,
\[
\frac{d}{d\theta_{a}}\log m_{a}=2\cot\theta_{a},
\qquad
\frac{d^{2}}{d\theta_{a}^{2}}\log m_{a}=-2\csc^{2}\theta_{a}.
\]
Taylor expansion of one step gives
\[
\Delta\log m_{a}
=2\cot\theta_{a}\,\Delta\theta_{a}
+O\!\left(\csc^{2}\theta_{a}\,(\Delta\theta_{a})^{2}\right).
\]
Substituting \eqref{eq:stochastic_flat_update} yields
\[
\Delta\log m_{a}
=-s\eta p+2s\cot\theta_{a}\,\xi_{a}+O(s^{2}),
\]
where the $O(s^{2})$ term is uniform in the stated small-step regime.
The deterministic contribution $-s\eta p$ is independent of $a$, so it cancels in
\[
\Delta X_{t}^{ab}
=\Delta\log m_{a}-\Delta\log m_{b}
=2s(\cot\theta_{a}\xi_{a}-\cot\theta_{b}\xi_{b})+O(s^{2}).
\]
Taking conditional expectation gives the $O(s^{2})$ drift.
Taking conditional variance and using $\mathbb E_{t}[\xi_{a}\xi_{b}]=(\Sigma_{t})_{ab}$ gives \eqref{eq:defect_ratio_variance}.

The polar decomposition used in \cref{lem:hc_frame} makes the principal angles first-order insensitive to angular, off-flat perturbations at the adapted flat.
After each noisy step, reselecting the Harish-Chandra frame therefore sends off-flat noise into the radial variables only through the second-order Taylor remainder.
Finally, near the target $\cot^{2}\theta_{a}\simeq 1/m_{a}$.
Along the deterministic flow $m_{a}(t)$ decays at rate $\eta$ to leading order, so over $O(1)$ flow time the cumulative leading conditional variance has size bounded by a constant multiple of
\[
\frac{s\|\Sigma_{\mathrm{flat}}\|}{\eta\,\min_{a\in A}m_{a}},
\]
which gives \eqref{eq:defect_noise_floor}.
\end{proof}

\subsection{Mixed-state SLD geometry}
\label{app:mixed_sld}
\begin{proof}[Proof of \cref{thm:mixed_qng_linearization}]
It suffices to compute the second variation in one two-level plane.
Let $K=X_{jk}$ or $Y_{jk}$ and consider
\[
\rho(t)=e^{-itK}\rho_{\star}e^{itK}.
\]
Then
\[
\dot\rho(0)=-i[K,\rho_{\star}],\qquad
\ddot\rho(0)=-[K,[K,\rho_{\star}]].
\]
Since $H_{\mathrm{obj}}$ and $\rho_{\star}$ are diagonal in the same basis, the first derivative of $\Tr(H_{\mathrm{obj}}\rho(t))$ vanishes.
A direct two-dimensional calculation using $\Tr(K^{2}\ket{j}\!\bra{j})=\Tr(K^{2}\ket{k}\!\bra{k})=1/2$ gives
\[
\frac{d^{2}}{dt^{2}}\Big|_{t=0}\Tr(H_{\mathrm{obj}}\rho(t))
=(E_k-E_j)(p_j-p_k).
\]
Thus the Euclidean Hessian eigenvalue along both $X_{jk}$ and $Y_{jk}$ is $h_{jk}=(E_k-E_j)(p_j-p_k)$.
By the standard SLD spectrum \cref{eq:standard_sld_orbit_spectrum}, the corresponding Fisher eigenvalue is
\[
f_{jk}=2\,\frac{(p_j-p_k)^2}{p_j+p_k}.
\]
On the horizontal tangent space, the linearized QNG matrix is $M_{\star}=(F^{\mathrm{SLD}}_{\star})^{+}H_{\mathrm{Eucl},\star}$, so
\[
\mu_{jk}=\frac{h_{jk}}{f_{jk}}
=\frac12(E_k-E_j)\frac{p_j+p_k}{p_j-p_k}.
\]
The Jacobian of the negative-gradient flow is $-\eta M_{\star}$.
\end{proof}

\begin{proof}[Proof of \cref{thm:depolarizing_noise_geometry}]
The eigenvalues of $\rho(q)$ are $p_0=1-q+q/D$ and $p_k=q/D$ for $k>0$. The only nonzero unitary-orbit directions connect $\ket0$ to the orthogonal subspace, so the real rank is $2D-2$. Applying \cref{eq:standard_sld_orbit_spectrum} to each pair $(0,k)$ gives
\begin{equation}
f(q)=2\,\frac{(p_0-p_k)^2}{p_0+p_k}=2\,\frac{(1-q)^2}{1-q+2q/D},
\end{equation}
independent of $k$, proving isotropy. The Euclidean gradient of a diagonal linear objective scales with $p_0-p_k=1-q$, while the SLD inverse contributes the reciprocal Fisher scale; comparison with the pure-state value gives \cref{eq:depolarized_qng_amplifier}. A gradient perturbation is multiplied by the same scalar $f^{-1}$ as the mean update, so that factor cancels in \cref{eq:qng_update_snr}; substituting $g(q)=(1-q)g_0$ gives \cref{eq:depolarized_update_snr}.
\end{proof}

\end{document}